\documentclass[11pt,reqno]{article}%{article}
\usepackage{fullpage} % decreases margins
\usepackage{amssymb} % AMS-LaTeX + Theorems + Symbols
\usepackage{graphicx} % to include graphics files
\usepackage{amsmath} % AMS-LaTeX + Theorems + Symbols
\usepackage{amsthm,amssymb,latexsym} % AMS-LaTeX + Theorems + Symbols
\usepackage{amstext, amsfonts,amsbsy} % not sure why I need it
\usepackage{mathrsfs} % for \mathscr
\usepackage{enumerate} % easy customization of enumeration
\usepackage{upgreek} % for \upalpha, etc.
\usepackage{float} % for \upalpha, etc.
\usepackage{color} % adds color
\usepackage{accents} % for \undertilde
\usepackage{mathtools} % for dcases, etc.
\usepackage{tikz}
\usetikzlibrary{matrix,graphs,arrows,positioning,calc,decorations.markings,shapes.symbols,shapes.geometric,shapes.misc}
\usepackage{calc, ifthen, ae, xstring, etoolbox, xkeyval,pgfplots}
\usepackage{tikz-cd}
\usepackage[pdftex,bookmarks,colorlinks,breaklinks]{hyperref}  % PDF hyperlinks, with coloured links
\definecolor{dullmagenta}{rgb}{0.4,0,0.4}   % #660066
\definecolor{darkgreen}{rgb}{0.0,0.5,0.0}   % #660066
\definecolor{darkblue}{rgb}{0,0,0.4}
\hypersetup{linkcolor=red,citecolor=blue,filecolor=dullmagenta,urlcolor=darkblue} % coloured links
\usepackage{cleveref}

\newtheorem{theorem}{Theorem}[section]
\newtheorem{proposition}[theorem]{Proposition}

\newtheorem{lemma}[theorem]{Lemma}
\newtheorem{definition}[theorem]{Definition}

\newtheorem{maintheorem}{Theorem}

\theoremstyle{definition}

\newtheorem{notation}[theorem]{Notation}
\newtheorem*{notation*}{Notation} 

\theoremstyle{remark}
\newtheorem{remark}[theorem]{Remark}
\AddToHook{env/remark/begin}{\crefalias{theorem}{remark}}
\AddToHook{env/lemma/begin}{\crefalias{theorem}{lemma}}
\AddToHook{env/definition/begin}{\crefalias{theorem}{definition}}
\AddToHook{env/proposition/begin}{\crefalias{theorem}{proposition}}
\AddToHook{env/corollary/begin}{\crefalias{theorem}{corollary}}
\AddToHook{env/assumption/begin}{\crefalias{theorem}{assumption}}

\numberwithin{equation}{section}
\newcommand{\Pain}[1]{\text{P}_{\mathrm{#1}}}

\DeclareRobustCommand{\nought}[1]{\accentset{\circ}{#1}}

\newcommand{\C}{\mathbb{C}}

\newcommand{\Z}{\mathbb{Z}}
\newcommand{\p}{\mathbb{P}}
\DeclareMathOperator{\Pic}{Pic}

\DeclareMathOperator{\Aut}{Aut}
\DeclareMathOperator{\Span}{Span}

\newcommand{\E}{\mathcal{E}}
\newcommand{\F}{\mathcal{F}}

\newcommand{\M}{\mathcal{M}}
\newcommand{\K}{\mathcal{K}}
\newcommand{\X}{\mathcal{X}}
\newcommand{\Y}{\mathcal{Y}}
\newcommand{\h}{\mathcal{H}}
\renewcommand{\L}{\mathcal{L}}
\newcommand{\pain}[1]{\text{P}_{\mathrm{#1}}}
\newcommand{\dpain}[1]{\text{dP}_{\mathrm{#1}}}

\newcommand\ubar[1]{%
  \underaccent{\bar}{#1}}
\newcommand{\Bl}{\operatorname{Bl}}

\usepackage{todonotes}
\usepackage{dynkin-diagrams}

\begin{document}

{\noindent\Large\bf On the discrete Painlev\'e equivalence problem, non-conjugate translations and nodal curves% (v1)\\[-3pt]
}
\medskip

\begin{flushleft}

\textbf{Anton Dzhamay}\\
Beijing Institute of Mathematical Sciences and Applications (BIMSA)\\ 
No.~544, Hefangkou Village Huaibei Town, Huairou District, Beijing 101408, China\\
E-mail: \href{mailto:adzham@bimsa.cn}{\texttt{adzham@bimsa.cn}}\qquad 
ORCID ID: \href{https://orcid.org/0000-0001-8400-2406}{\texttt{0000-0001-8400-2406}}\\[5pt]

\textbf{Galina Filipuk}\\
Institute of Mathematics, University of Warsaw, Banacha 2, Warsaw, 02-097, Poland\\
E-mail: \href{mailto:filipuk@mimuw.edu.pl}{\texttt{filipuk@mimuw.edu.pl}}\qquad
ORCID ID: \href{https://orcid.org/0000-0003-2623-5361}{\texttt{0000-0003-2623-5361}}\\[5pt] 

\textbf{Alexander Stokes}\\
Waseda Institute for Advanced Study, Waseda University, Nishi-Waseda Building, 1-21-1 Nishi Waseda, Shinjuku-ku, Tokyo 1-21-1 Nishi Waseda, Tokyo 169-0051, Japan\\
E-mail: \href{mailto:stokes@aoni.waseda.jp}{\texttt{stokes@aoni.waseda.jp}}\qquad
ORCID ID: \href{https://orcid.org/0000-0001-6874-7141}{\texttt{0000-0001-6874-7141}}\\[8pt]

\emph{Keywords}: Orthogonal polynomials, recurrence coefficients, discrete Painlev\'e equations, affine Weyl groups, nodal curves.\\[3pt]

\emph{MSC2020}: 
33E17,      % Painlevé-type functions
33C45.      %Orthogonal polynomials and functions of hypergeometric type (Jacobi, Laguerre, Hermite, Askey scheme, etc.)
39A36.    %Integrable difference and lattice equations; integrability tests
%33D45,     %Basic orthogonal polynomials and functions (Askey-Wilson polynomials, etc.)
%34M55,     %Painlevé and other special ordinary differential equations in the complex domain; classification, hierarchies
%34M56,  	%Isomonodromic deformations for ordinary differential equations in the complex domain
%14E07,     %Birational automorphisms, Cremona group and generalizations
%14J26,  	%Rational and ruled surfaces
%14J81  	%Relationships between surfaces, higher-dimensional varieties, and physics
%39A13      %Difference equations, scaling ($q$-differences) 
%39A45  	%Difference equations in the complex domain
%37F10  	%Dynamics of complex polynomials, rational maps, entire and meromorphic functions; Fatou and Julia sets
\end{flushleft}

\date{}

\begin{abstract}
We consider several examples of nonautonomous systems of difference equations coming from semi-classical orthogonal polynomials via recurrence coefficients and ladder operators, with respect to various generalisations of Laguerre and Meixner weights.
We identify these as discrete Painlev\'e equations and establish their types in the Sakai classification scheme in terms of the associated rational surfaces. 
In particular, we find examples which come from different weights and share a common surface type $D_5^{(1)}$ but are inequivalent in two ways. 
First, their dynamics are generated by non-conjugate elements of $\widehat{W}(A_3^{(1)})$. 
Second, some of the examples have associated surfaces being non-generic in the sense of having nodal curves. 
The symmetries of these examples form subgroups of the generic symmetry group, which we compute. 
In particular, we find $(W(A_1^{(1)})\times W(A_1^{(1)}))\rtimes \Z/2\Z$.
These examples give further weight to the argument that any correspondence between different weights and the Sakai classification  should make use of the refined version of the discrete Painlev\'e equivalence problem, which takes into account not just surface type, but also the group elements generating the dynamics as well as parameter constraints, e.g. those corresponding to nodal curves.
\end{abstract}

\section{Introduction}

Second-order systems of differential or difference equations satisfied by various quantities related to semi-classical orthogonal polynomials often turn out to be differential or discrete Painlev\'e equations. 
The earliest example is perhaps that of a second-order difference equation shown by Shohat  \cite{shohat} to be satisfied by the coefficients in the three-term recurrence for polynomials $P_n(x)$ orthogonal with respect to the weight $w(x;t) = e^{t x^2 - x^4}$. 
This turns out to be a special case of a discrete Painlev\'e equation, which was first dubbed a discrete Painlev\'e-I equation when it was found, also in relation to orthogonal polynomials, in \cite{fokasitskitaev}, due to it admitting a continuum limit to the first Painlev\'e equation $\pain{I}$.
Since then, numerous cases of polynomials orthogonal with respect to semi-classical weights have been shown to be related to Painlev\'e equations, differential or discrete, via systems of differential or difference equations satisfied by recurrence coefficients or other associated quantities (see e.g. the monograph \cite{walterbook} and references therein).

However, the systems may not immediately take the standard form of a Painlev\'e equation, but rather are written in coordinates coming from the orthogonal polynomials side - they are Painlev\'e equations in disguise. 
The problem of unmasking such an example, i.e. finding a birational change of variables to some standard form of a Painlev\'e equation, is known as a ``Painlev\'e equivalence problem" \cite{Cla:2019:OPPE}.
Sometimes, solving it is possible by inspection of the equations themselves, but at other times it is far from obvious.
In \cite{DzhFilSto:2020:RCDOPWHWDPE}, an algorithmic procedure was presented for recognising a discrete Painlev\'e equation written in non-standard coordinates, determining its type in the Sakai classification scheme \cite{SAKAI2001}, and transforming it to a standard form of a discrete Painlev\'e equation of the same type, e.g. as appearing in \cite{KajNouYam:2017:GAPE} (see \cite{DzhFilSto:2022:DERCSOPTRPEGA} for the differential case and \cite{hamiltonians} for considerations of Hamiltonian structures).

The Sakai classification scheme for discrete Painlev\'e equations is based on a classification of certain rational surfaces.
Sakai defined a generalised Halphen surface as a smooth projective rational surface over $\C$ with an effective anticanonical divisor $D= \sum_i m_i D_i$ such that the intersection numbers $D_i\bullet D=0$ for all $i$ (such an anticanonical divisor $D$ is said to be of canonical type).
A generalised Halphen surface is either a Halphen surface of index one (a kind of rational elliptic surface), or has unique effective anticanonical divisor, and in this paper we refer to this latter type as \emph{Sakai surfaces}, which were classified in \cite{SAKAI2001}.
Symmetries of Sakai surfaces give rise to discrete Painlev\'e equations.

Part of the data specifying the type of a discrete Painlev\'e equation is a pair of affine Dynkin diagrams associated to the \emph{surface type} $\mathcal{R}$ and \emph{symmetry type} $\mathcal{R}^{\perp}$.
The surface type $\mathcal{R}$ describes the unique effective anticanonical divisor, in a way similar to the Kodaira classification of singular fibres of elliptic surfaces.
The symmetry type $\mathcal{R}^{\perp}$ is that of the affine Weyl group in terms of which the symmetries of the surface are described, and is determined by $\mathcal{R}$.

When there are multiple instances of discrete Painlev\'e equations arising from some common situation where there is a classification present, it is natural to ask how this matches with the Sakai scheme. 
This is the case for discrete Painlev\'e equations appearing in relation to recurrence coefficients of semiclassical orthogonal polynomials and gap probabilities in related probabilistic models, and one may ask how different weights on the orthogonal polynomials side match up with different types of discrete Painlev\'e equations. 
Among the discrete systems coming from these contexts that have been identified as discrete Painlev\'e equations and placed in the Sakai classification are examples of surface/symmetry type $(\mathcal{R},\mathcal{R}^{\perp})$ being 
$(A_1^{(1)},E_7^{(1)})$ \cite{antonalisa}, 
$(D_4^{(1)},D_4^{(1)})$ \cite{DzhFilSto:2020:RCDOPWHWDPE,borodin2003,borodinboyarchenko}, 
$(D_5^{(1)},A_3^{(1)})$ \cite{borodinboyarchenko,HuDzhChe:2020:PLUEDPE}, 
$(D_6^{(1)}, (A_1+A_1)^{(1)})$ \cite{D6lixing},
$(E_6^{(1)},A_2^{(1)})$ \cite{shohat, DzhFilSto:2022:DERCSOPTRPEGA, antonelizaveta}, $(A_5^{(1)}, (A_2+A_1)^{(1)})$ \cite{qlaguerre}.

However, the pair $(\mathcal{R},\mathcal{R}^{\perp})$ does not uniquely determine a discrete Painlev\'e equation, even up to equivalence via birational changes of variables, and there can be infinitely many inequivalent discrete Painlev\'e equations of the same surface type, corresponding to non-conjugate (quasi-)translation elements of the symmetry group corresponding to $\mathcal{R}^{\perp}$.
% \cite{dPsinfinite}
% see \Cref{rem:nonconjugatetranslations}.
Further, while each surface type determines a unique symmetry type, this refers to symmetries of families of surfaces in the \emph{generic} case. 
When parameters take special values, the symmetries compatible with these form a subgroup of that from the generic case, which may not appear explicitly in the Sakai list. 
The existence of such `exotic' symmetry types arising from non-generic Sakai surfaces has been known about for quite some time, at least as far back as \cite{takenawaD5}, in which the extended affine Weyl group $\widetilde{W}((A_1+A_3)^{(1)})$ was identified as a subgroup of the generic $\widetilde{W}(D_5^{(1)})$ symmetry for surfaces of type $A_3^{(1)}$. 
For more examples of exotic symmetry types being identified in discrete Painlev\'e equations, see \cite{AHJN,yangtranslations,dP2symmetry}. 
Such non-generic symmetry types can arise when the surface has nodal curves (here meaning $(-2)$-curves disjoint from the anticanonical divisor, see \Cref{def:nodalcurve}) or when there are constraints on parameters in the surface and associated equations related to projective reduction \cite{KNT}. 
Combinations of the two are also possible, see \Cref{rem:RCGeqn} in the main text.

The point of this paper is to demonstrate that any correspondence between different weights and different types of discrete Painlev\'e equations should be \emph{finer} than the type $(\mathcal{R},\mathcal{R}^{\perp})$.
When a discrete Painlev\'e equation appearing in the context of orthogonal polynomials (or another of the many areas of mathematics and physics in which they play a role), to specify the equation up to equivalence via birational changes of variables one should provide a tuple $(\mathcal{R},\mathcal{R}^{\perp}, \mathcal{C}, \mathcal{S}, [w])$, where:
\begin{itemize}
    \item $\mathcal{R}$ is the surface type; 
    \item $\mathcal{R}^{\perp}$ is the generic symmetry type, which is determined by $\mathcal{R}$;
    \item $\mathcal{C}$ is the specification of any parameter constraints, which when written in terms of root variables defines a subset of the set of isomorphism classes of surfaces of type $\mathcal{R}$;
    \item $\mathcal{S}$ is the symmetry type, which corresponds to a subgroup (which may be proper if $\mathcal{C}$ is nontrivial) of the generic one corresponding to $\mathcal{R}^{\perp}$;
    \item $[w]$ is the conjugacy class in the symmetry group of the element generating the equation. In generic cases this refers to conjugacy in the generic symmetry group determined by $\mathcal{R}$, but in non-generic cases requires consideration of parameter constraints, see \Cref{rem:conjugacyinsymmetrygroups} in the main text.
\end{itemize}
We demonstrate the necessity to specify not just the pair $(\mathcal{R},\mathcal{R}^{\perp})$ through several examples coming from orthogonal polynomials with respect to different weights.
The examples we consider all share the same $\mathcal{R}$ and $\mathcal{R}^{\perp}$, but have different $\mathcal{C}$, $\mathcal{S}$ and $[w]$.
In \cite{D6lixing}, two of us, together with Li and Zhang, pointed out the necessity to specify $[w]$ when considering a discrete Painlev\'e identification problem. 
In this paper we demonstrate the necessity to specify $\mathcal{C}$ and the $\mathcal{S}$ it determines, and the implications of this for the notion of symmetry types for discrete Painlev\'e equations.

\subsection{Discrete systems from various generalisations of Meixner and Laguerre orthogonal polynomials}
\label{subsec:discretesystemsfromorthogonalpolynomials}

The way in which the examples we consider arise from orthogonal polynomials is via three-term recurrence relations for the polynomials, or the related ladder operator construction.
The former is classical, see e.g. \cite{szegobook, Chiharabook,walterbook}, and the latter was developed by Chen and Ismail \cite{ChenIsmail} and is by now well-known, see \cite[Ch. 3]{Ismailbook}.
We remark that discrete Painlev\'e equations can also be derived from orthogonal polynomials via isomonodromy deformations of linear differential equations satisfied by the polynomials and associated Riemann-Hilbert problems as in, e.g.,  \cite{fokasitskitaev,borodinboyarchenko}.

In each case we consider below, there are two quantities, related to the recurrence coefficients and/or ladder operators of the relevant orthogonal polynomials, which satisfy a system of difference equations with respect to the index $n$ of the polynomials, and a system of differential equations with respect to a parameter in the weight. 

The examples we will consider are as follows. We will not reproduce the derivations of the systems and instead defer these to the references given below.
\subsubsection{Laguerre weight on a finite interval (L)}
This weight was considered in \cite{LC17} and is given by 
    \begin{equation}\label{weight-L}
        w^{\operatorname{L}}(x;\alpha) = x^{\alpha} \exp(-x), \quad x \in [0,s], \quad s \geq 0, \alpha>0.   
    \end{equation}
    This leads to a system of difference equations with respect to $n$ and a system of differential equations with respect to $s$ for functions $f=f_n(s)$, $g=g_n(s)$ coming from the ladder operators for the polynomials orthogonal with respect to the weight \eqref{weight-L}.
    Explicitly the discrete system is given by
\begin{equation}\label{eq-LUE-disc} \tag{L}
	\left\{
		% \bar{f} f = \frac{g^2 - (\alpha+ 2n)g + (\alpha+n)n}{g^2} ,\qquad
		\bar{f} f = \frac{(g - n)(g-\alpha-n)}{g^2} ,\qquad
		\ubar{g} + g = \frac{(1-\alpha + s -2n) f + \alpha + 2n +1}{(f-1)^2},
	\right.
\end{equation}
where $\bar{f}=f_{n+1}$, $\ubar{g}=g_{n-1}$. 
Similar notation will be used to indicate discrete evolution with $n$ throughout the paper.

\subsubsection{Perturbed Laguerre weight on the nonnegative real line (pL)}
    This weight was considered in \cite{MC20}, and also earlier in \cite{forresterormerod}, and is given by 
    \begin{equation} \label{weight-pL}
        w^{\operatorname{pL}}(x;s,\alpha,\gamma) = x^{\alpha} \exp(-x) |x-s|^{\gamma}(A+ B \theta(x-s)), \quad x\in [0,\infty), \quad 
        \begin{gathered}
            s \geq 0, \,\alpha > 0, \,\gamma > 0,\\
            A\geq 0, \,A+B\geq 0.
        \end{gathered} 
    \end{equation}
    This leads to a system of difference equations with respect to $n$ and a system of differential equations with respect to $s$ for functions $f=f_n(s)$, $g=g_n(s)$ coming from the ladder operators for the polynomials orthogonal with respect to the weight \eqref{weight-pL}, in which $\alpha$ and $\gamma$ are regarded as constant parameters.
     Explicitly, the discrete system is given by
   \begin{equation}\label{eq-pLUE-disc} \tag{pL}
	\left\{
		\bar{f} f = \frac{(g-  n)(g-\alpha-n)}{g(g+\gamma)} ,\quad
		\ubar{g} + g = \frac{-\gamma f^2  + ( 1-\alpha + s -2n + \gamma) f + \alpha + 2n +1}{(f-1)^2},
	\right.
    \end{equation}
and reduces to the system \eqref{eq-LUE-disc} when $\gamma=0$.
 \subsubsection{Semi-classical Meixner weight on $\mathbb{Z}_{\geq 0}$ (M) }
    This weight was considered in \cite{BoelenFilipukVanAssche} and is given by
    \begin{equation} \label{weight-M}
        w^{\operatorname{M}}(k;c,\gamma) = \frac{(\gamma)_k c^k}{(k!)^2}, \quad k\in \Z_{\geq 0}, \quad c > 0, \,\gamma > 0.
    \end{equation}    
    This leads to a system of difference equations with respect to $n$ and a system of differential equations with respect to $c$ for functions $x=x_n(c)$, $y=y_n(c)$ related to the coefficients $a_n^2$, $b_n$ from the three-term recurrence satisfied by the polynomials orthonormal with respect to the weight \eqref{weight-M}. 
    Explicitly the discrete system is
\begin{equation}\label{eq-degen-Meixner-disc} \tag{M}
	\left\{
		(\bar{x} + y)(x + y) = \frac{(\gamma-1)}{c^2} y (y-c)^2  ,\quad
		(x + \ubar{y})(x + y) = \frac{x(x+c)^2}{x- \tfrac{cn}{\gamma-1}}.
	\right.
\end{equation}
 
     \subsubsection{Generalised semi-classical Meixner weight on $\mathbb{Z}_{\geq 0}$ (gM)}
    The weight was considered in \cite{SmetVanAssche,FilipukVanAssche} and is given by 
    \begin{equation}\label{weight-gM}
        w^{\operatorname{gM}}(k;c,\beta,\gamma) = \frac{(\gamma)_k c^k}{(\beta)_k k!}, \quad k\in \Z_{\geq 0}, \quad c > 0, \,\beta>0, \, \gamma > 0.
    \end{equation}   
     This leads to a system of difference equations with respect to $n$ and a system of differential equations with respect to $c$ for functions $x=x_n(c)$, $y=y_n(c)$ related to the coefficients $a_n^2$, $b_n$ from the three-term recurrence satisfied by the polynomials orthonormal with respect to the weight \eqref{weight-gM}.
       Explicitly the discrete system is
\begin{equation}\label{eq-Meixner-disc} \tag{gM}
	\left\{
		(\bar{x} + y)(x + y) = \frac{(\gamma-1)}{c^2} y (y-c) \left(y- c \frac{\gamma-\beta}{\gamma-1}\right)  ,\quad
		(x + \ubar{y})(x + y) = \frac{x(x+c)}{x- \tfrac{cn}{\gamma-1}} \left(x + c \frac{\gamma-\beta}{\gamma-1}\right),
	\right.
\end{equation}
and reduces to the system \eqref{eq-degen-Meixner-disc} when $\beta=1$.

\subsubsection{Known relations to discrete Painlev\'e equations}
Regarding connections between the above weights and differential and discrete Painlev\'e equations, we make the following remarks. 

For \eqref{eq-LUE-disc}, the discrete system is a discrete Painlev\'e equation of surface type $D_5^{(1)}$, and the associated differential system with respect to $s$ can be transformed to a special case of one of the standard Hamiltonian forms of the fifth Painlev\'e equation $\pain{V}$, which shares that surface type.
This was established in \cite{HuDzhChe:2020:PLUEDPE}, using the approach of \cite{DzhFilSto:2020:RCDOPWHWDPE, DzhFilSto:2022:DERCSOPTRPEGA}.

For \eqref{eq-pLUE-disc}, we remark that the $\gamma=0$ case of the weight $w^{\operatorname{pL}}$ in equation \eqref{weight-pL} was studied by Basor and Chen in \cite{BasorChenPV}, and they derived a particular case of $\pain{V}$, which coincides with that associated with \eqref{eq-LUE-disc}.
In the general case, the ladder operator construction was performed by Forrester and Ormerod \cite{forresterormerod}, who also compared this with the approach involving a Riemann-Hilbert problem for the orthogonal polynomials along the lines of \cite{fokasitskitaev}.
They connected the differential system with respect to $s$ with $\pain{V}$, which was also done in \cite{MC20} (in a way consistent with the results of \cite{BasorChenPV} when $\gamma=0$), but the discrete system itself was not identified as a discrete Painlev\'e equation.
We remark that the discrete system studied by Zhu, Chen and Zhang in \cite{siqi} is very similar but comes from a different weight, and was identified with a discrete Painlevé equation of surface type $D_5^{(1)}$. This example is the same one with which we identify \eqref{eq-pLUE-disc} below, but the recurrence considered in \cite{siqi} comes with a different initial condition.

For \eqref{eq-degen-Meixner-disc}, the weight $w^{\operatorname{M}}$ in equation \eqref{weight-M} is a special case of one introduced by Ronveaux \cite{ronveaux86}, and when $\gamma=1$, the weight $w^{\operatorname{M}}$ reduces to the classical Charlier weight $\frac{c^k}{k!}$.
In this case, the discrete system breaks down and is not required to compute the recurrence coefficients, which are given by $a_n^2=c n$, $b_n=n+c$.
We remark that the $N=2$ case of the generalised Charlier weight $\frac{c^k}{(k!)^{N}}$ on $\Z_{\geq 0}$ defined in \cite{genCharlier} was shown to be related to a particular case of a discrete Painlev\'e equation known as $\dpain{II}$ in \cite{VanAsscheFoup}, which also shares the surface type $D_5^{(1)}$.

For \eqref{eq-Meixner-disc}, the connection to the fifth Painlev\'e equation was found in \cite{FilipukVanAssche}. 
The discrete system was derived in \cite{SmetVanAssche}, and was identified as a limiting case of a discrete Painlev\'e equation quoted from \cite{GramaniReview} and \cite{VanAsscheDESFOP}.
This limiting procedure involves introducing a parameter $\varepsilon$ into the variables and coefficients in an example of a discrete Painlev\'e equation, then taking a formal limit as $\varepsilon\to0$ to derive the system \eqref{eq-Meixner-disc}.
Therefore neither the type of the discrete system \eqref{eq-Meixner-disc} in the Sakai scheme, nor its possible relation to a standard example of a discrete Painlev\'e equation via birational transformation, were established.
The weight $w^{\operatorname{gM}}$ in equation \eqref{weight-gM} can be considered on either the lattice $\mathbb{N}=\Z_{\geq 0}$ or the shifted lattice $\mathbb{N}+1-\beta$, and these can be combined to give a weight on the bilattice $\mathbb{N}\cup \left(\mathbb{N}+1-\beta\right)$. 
These three situations lead to the same discrete and differential systems but with different initial conditions, as shown in \cite{SmetVanAssche} and explained in \cite{FilipukVanAssche}.
An alternative derivation of the relation between the discrete and differential systems and $\pain{V}$ via Hankel determinant representations and the $\sigma$-form of $\pain{V}$ was given in \cite{ClarksonOPs}.

\subsection{Summary of results}
The systems of difference equations \eqref{eq-LUE-disc}, \eqref{eq-pLUE-disc}, \eqref{eq-degen-Meixner-disc}, and \eqref{eq-Meixner-disc} correspond to discrete Painlev\'e equations which all share a surface type $\mathcal{R}=D_5^{(1)}$, and the systems of differential equations are all related by birational transformations to instances of the fifth Painlev\'e equation $\pain{V}$ in its standard Hamiltonian form, which lives on the same type of surface.
However, the examples of discrete Painlev\'e equations are distinguished by the conjugacy classes of their associated symmetries as well as their symmetry groups, which is related to the presence of nodal curves on the surfaces, as we will elaborate on in \Cref{sec:identificationprocedure}.

The generic symmetry type of a Sakai surface of type $D_5^{(1)}$ is $\mathcal{R}^{\perp}=A_3^{(1)}$.
That is, for a generic surface of type $D_5^{(1)}$, the symmetry group (in the sense of Cremona isometries of the surface and their Cremona action as made precise in \cite{SAKAI2001,ASIDEnotes}) is $\widehat{W}(A_3^{(1)}) := W(A_3^{(1)}) \rtimes \operatorname{Aut}(A_3^{(1)})$, 
where $W(A_3^{(1)})$ is the affine Weyl group of type $A_3^{(1)}$ and $\operatorname{Aut}(A_3^{(1)})$ is the group of automorphisms of the corresponding (extended) Dynkin diagram.
Discrete Painlev\'e equations are generated by translation elements (or more generally elements of infinite order) of this group.

There are two examples of discrete Painlev\'e equations of surface type $D_5^{(1)}$ which are relevant in this paper. 

The first is that appearing in \cite[Sec. 8.1.17]{KajNouYam:2017:GAPE}, which we refer to as the KNY example:  
 \begin{equation} \label{eq-KNY-dP} 
	\left\{
	\begin{aligned}
		\bar{q} 	&= 1 - q - \frac{a_0}{p+t} - \frac{a_2}{p}, \\
		\ubar{p} 	&= - p - t + \frac{a_1}{q} + \frac{a_3}{q-1}, 
	\end{aligned}
	\qquad 
	\begin{aligned}
		\bar{a}_0 &= a_0 + 1, &&\bar{a}_1 = a_1 - 1, \\
		\bar{a}_2 &= a_2 + 1, &&\bar{a}_3 = a_3 - 1,
	\end{aligned}
	\right.
    \tag{$\operatorname{KNY}$}
\end{equation}
where we have adopted the usual convention of writing the non-autonomous nature of a discrete Painlev\'e equation in terms of evolution of parameters $a_i$, which here are normalised so that $a_0+a_1+a_2+a_3=1$.
The associated symmetry is a translation, which we denote by
$T_{\operatorname{KNY}} \in \widehat{W}(A_3^{(1)})$,
associated with a weight of squared length $1$ of the underlying finite $A_3$ type root system. 

The second is one that appeared in \cite[Sec. 7]{SAKAI2001} (though in different coordinates) labelled as $\dpain{IV}$, which we refer to as the Sakai example: 
\begin{equation} \label{eq-Sakai-dP} 
	\left\{
	\begin{aligned}
		\bar{q} 	&= \frac{p+t}{t} \left( 1 - \frac{a_0 + a_1}{ a_0 + (p+t)q} \right) , \\
		\bar{p} 	&= \frac{1}{1-\bar{q}}\left( a_2 + \frac{(a_0+a_1)p}{t+p - t \bar{q}} \right),  
	\end{aligned}
	\qquad 
	\begin{aligned}
		\bar{a}_0 &= a_0 + 1, &&\bar{a}_1 = a_1, \\
		\bar{a}_2 &= a_2, &&\bar{a}_3 = a_3 - 1.
	\end{aligned}
	\right.
    \tag{$\operatorname{Sak}$}
\end{equation}
The associated symmetry is a translation, which we denote by
$T_{\operatorname{Sak}} \in \widehat{W}(A_3^{(1)})$,    
which is associated to a weight of squared length $\tfrac{3}{4}$, so is not conjugate to $T_{\operatorname{KNY}}$.

Our results establishing the types of the systems of difference equations from the various weights above as discrete Painlev\'e equations are summarised as follows.

\begin{maintheorem}
\label{mainthm:pLaguerre}
The system of difference equations \eqref{eq-pLUE-disc} is a discrete Painlev\'e equation of surface type $D_5^{(1)}$ with
        \begin{enumerate}[(1)]
           \item symmetry group element being a translation conjugate to $T_{\operatorname{KNY}}$;
           \item surfaces being generic, in that root variables are free when parametrised by parameters from the weight;
           \item the full symmetry group $\widehat{W}(A_3^{(1)})=W(A_3^{(1)})\rtimes \operatorname{Aut}(A_3^{(1)})$.
        \end{enumerate}       
    The system \eqref{eq-pLUE-disc} is related to the KNY example \eqref{eq-KNY-dP} by a birational change of variables up to appropriate parameter identification.    
\end{maintheorem}

\begin{maintheorem}
\label{mainthm:Laguerre}
The system of difference equations \eqref{eq-LUE-disc} is a discrete Painlev\'e equation of surface type $D_5^{(1)}$ with
        \begin{enumerate}[(1)]
           \item symmetry group element being a translation conjugate to $T_{\operatorname{KNY}}$;
            \item parameter constraint $a_1+a_2=0$ on root variables, which corresponds to the existence of a nodal curve;
           \item symmetry subgroup $W( ( {A_1} + \underset{|\alpha|^2=4}{A_1})^{(1)} ) \rtimes \mathbb{Z}/2\Z \subset \widehat{W}(A_3^{(1)})$, where the generator of $\mathbb{Z}/2\Z$ acts by a particular automorphism of the pair of two Dynkin diagrams of type $A_1^{(1)}$.
        \end{enumerate}       
    The system \eqref{eq-LUE-disc} is related to the $a_1+a_2=0$ case of the KNY example \eqref{eq-KNY-dP} by a birational change of variables up to appropriate parameter identification.    
\end{maintheorem}

\begin{maintheorem}
\label{mainthm:genMeixner}
The system of difference equations \eqref{eq-Meixner-disc} is a discrete Painlev\'e equation of surface type $D_5^{(1)}$ with
        \begin{enumerate}[(1)]
           \item surfaces being generic, in that root variables are free when parametrised by parameters from the weight;
           \item symmetry group element being a translation conjugate to $T_{\operatorname{Sak}}$;
           \item the full symmetry group $\widehat{W}(A_3^{(1)})=W(A_3^{(1)})\rtimes \operatorname{Aut}(A_3^{(1)})$.
        \end{enumerate}       
    The system \eqref{eq-Meixner-disc} is related to the Sakai example \eqref{eq-Sakai-dP} by a birational change of variables up to appropriate parameter identification.    
\end{maintheorem}

\begin{maintheorem}
\label{mainthm:Meixner}
The system of difference equations \eqref{eq-degen-Meixner-disc} is a discrete Painlev\'e equation of surface type $D_5^{(1)}$ with
        \begin{enumerate}[(1)]
           \item symmetry group element being a translation conjugate to $T_{\operatorname{Sak}}$;
           \item parameter constraint $a_2=1$ on root variables, which corresponds to the existence of a nodal curve;
           \item symmetry subgroup $W( ( {A_1} + \underset{|\alpha|^2=4}{A_1})^{(1)} ) \rtimes \mathbb{Z}/2\Z\subset \widehat{W}(A_3^{(1)})$, where the generator of $\mathbb{Z}/2\Z$ acts by a particular automorphism of the pair of two Dynkin diagrams of type $A_1^{(1)}$.
        \end{enumerate}       
    The system \eqref{eq-degen-Meixner-disc} is related to the $a_2=1$ case of the Sakai example \eqref{eq-Sakai-dP} by a birational change of variables up to appropriate parameter identification.    
\end{maintheorem}

The plan of the paper is as follows.
In \Cref{sec:KNY} we recall the realisation of Sakai surfaces of surface type $D_5^{(1)}$ provided in \cite{KajNouYam:2017:GAPE}, as well as the fifth Painlev\'e equation and the KNY and Sakai examples of discrete Painlev\'e equations in coordinates coming from this realisation. 
In \Cref{sec:identificationprocedure} we outline the identification procedure and discuss how nodal curves on Sakai surfaces give rise to restricted symmetry groups, then in Sections \ref{sec:defLUE}-\ref{sec:degenMeixner} we study the four examples and establish the results as summarised above.
We conclude with some discussion in \Cref{sec:discussionconclusion}.

\section{Reference model of Sakai surfaces of type $D_5^{(1)}$} 
\label{sec:KNY}

In this Section, we recall the necessary geometric and algebraic data for the realisation of the $D_5^{(1)}$ surface family and its associated discrete Painlev\'e equations that we will use as our reference model, with which to identify the examples from orthogonal polynomials outlined in \Cref{subsec:discretesystemsfromorthogonalpolynomials}. 
This realisation is the one provided in Section 8.2.18 of the survey paper \cite{KajNouYam:2017:GAPE} of Kajiwara, Noumi and Yamada, most of the details of which can also be found in \cite{HuDzhChe:2020:PLUEDPE}, but we include it here to make the present paper self-contained.
For a formal account of details of the general theory of discrete Painlev\'e equations we refer to \cite{ASIDEnotes} as well as the original paper of Sakai \cite{SAKAI2001}.

To differentiate between the different instances of surfaces of type $D_5^{(1)}$ we encounter in this paper, we provide the following.

\begin{notation} \label{not:KNY}
For the Kajiwara-Noumi-Yamada (KNY) model of $D_5^{(1)}$ surfaces we use the following notation: coordinates $(q,p)$; 
parameters $a_{0},a_{1},a_{2},a_{3}$ (the root variables); continuous independent variable $t$ (extra parameter);
centres of blowups $b_{i}$; exceptional divisors $E_{i}$.	
\end{notation}

\subsection{Point configuration and surfaces}
\label{subsec:KNYpointconfigandsurfaces}

We begin with the sequence of blowups performed in the KNY realisation of Sakai surfaces of surface type $D_5^{(1)}$.
Take $\C^2$ with coordinates $(q,p)$, and compactify this to $\p^1\times\p^1$ with the atlas
\begin{equation} \label{eq:P1P1atlasKNY}
    \p^{1} \times \p^{1} = \C^2_{(q,p)} \cup \C^2_{(Q,p)} \cup \C^2_{(q,P)}\cup \C^2_{(Q,P)},\quad Q=1/q, P=1/p,
\end{equation} 
where here we use subscripts to indicate coordinates.
This provides a compactification of $\C^2_{(q,p)}$, in that we have glued in two copies of $\p^1$ corresponding to $q=\infty$ and $p=\infty$, given respectively by $Q=0$ and $P=0$.
Before we specify the points to be blown up to arrive at the surfaces, we specify our notation for charts.
\begin{notation} \label{not:blowupcharts}
For each blowup, we introduce two coordinate charts according to the following convention:
after blowing up a point $b_i$ in a surface $Y$, we denote the resulting surface by $\Bl_{b_i} Y$, the projection (blowup morphism) by $\pi_i : \Bl_{b_i}Y \rightarrow Y$, and the exceptional divisor of the blowup of $b_i$ by $E_i = \pi^{-1}(b_i) \subset \Bl_{b_i}Y$, $E_i\cong \p^1$.
If $b_i$ is given in some local coordinates $(x,y)$ for an open subset of $Y$ by
$b_i : (x,y) = (x_*, y_*)$,
then the exceptional divisor $E_i$ is covered by two charts on $\Bl_{b_i}Y$ with coordinates $(u_i,v_i)$ and $(U_i,V_i)$, related to those of $Y$ under the blowdown morphism according to
\begin{equation}
\begin{aligned}
x&=  u_i + x_{*}, 			& y &= u_i v_i+ y_{*}  , \\
u_i &= x-x_{*},	 & v_i  &=\frac{y- y_{*}}{x- x_{*} },
\end{aligned}
\qquad \text{and} \qquad 
\begin{aligned}
x &= U_i V_i +x_{*} , 			& y&=V_i+ y_{*} ,\\
U_i &= \frac{x-x_*}{y-y_*},		& V_i &= y - y_*.
\end{aligned}
\end{equation}
In particular, in the coordinates $(u_i,v_i)$ the exceptional divisor $E_i$ is given by the local equation $u_i=0$, with the part of $E_i$ visible in this chart parametrised by $v_i\in\C$.
Similarly, in the coordinates $(U_i,V_i)$, $E_i$ is given by $V_i=0$, with the part of $E_i$ visible in this chart parametrised by $U_i\in\C$.
\end{notation}
\begin{remark}
    There is some variation in the literature in the definitions of charts introduced in the blowup procedure, despite the same symbols $(u_i,v_i)$ and $(U_i,V_i)$ being used. 
    Here we adopt the same convention as \cite{DzhFilSto:2020:RCDOPWHWDPE, DzhFilSto:2022:DERCSOPTRPEGA, hamiltonians}, which is different from, e.g., \cite{quasihamiltonians,takasakipaper}.
\end{remark}
The points to be blown up depend on parameters (root variables) $a_0,a_1,a_2,a_3\in \C$, subject to $a_{0} + a_{1} + a_{2} + a_{3} = 1$, as well as $t\in \C \backslash \{0\}$ (called the `extra parameter' in \cite{SAKAI2001}) which will play the role of independent variable of the Hamiltonian form of $\pain{V}$ which will be given below. 

We initially blow up the four points $b_1,b_3,b_5,b_7$ in $\p^1 \times\p^1$, then perform four subsequent blowups of points which lie on the exceptional divisors of the first four.
We give the locations of these points in coordinates explicitly in Figure \ref{fig:KNY:pointlocations}.

\begin{figure}[htb]
\begin{equation*}
 \begin{array}{ll}
	b_{1} : (q,p) = (\infty,-t) &\leftarrow   \quad b_{2} : (u_1,v_1) = \left(Q, \frac{p+t}{Q} \right) = (0,-a_{0}),\\
	b_{3} : (q,p) = (\infty,0)  &\leftarrow   \quad b_{4} : (u_3,v_3) = \left(Q,\frac{p}{Q}  \right)= (0,-a_{2}),\\
	b_{5} : (q,p) = (0,\infty)  &\leftarrow   \quad b_{6} : (U_5,V_5) = \left(\frac{q}{P},P \right) = (a_{1},0),\\
	b_{7} : (q,p) = (1, \infty)  &\leftarrow  \quad b_{8} : (U_8,V_8) = \left(\frac{q-1}{P} ,P \right) = (a_{3},0).
\end{array}
\end{equation*}
    \caption{Point locations for the KNY realisation of $D_5^{(1)}$ surfaces $X_{\boldsymbol{a},t}$}
    \label{fig:KNY:pointlocations}
\end{figure}
Denote the surface obtained by this sequence of eight blowups by 
\begin{equation*}
\begin{gathered}
    X_{\boldsymbol{a},t} = \Bl_{b_1\cdots b_8}(\p^1\times\p^1), \\ \boldsymbol{a}\in\mathscr{A}= \left\{(a_0,a_1,a_2,a_3)\in \C^4~~|~~  a_0+a_1+a_2+a_3=1 \right\}, \quad t \in\mathscr{T}=\C\setminus\{0\}.
\end{gathered}
\end{equation*}
Here we have introduced the parameter space $\mathscr{A}$ consisting of the root variables, see \cite[Sec. 3.6]{ASIDEnotes}, as well as the extra parameter $t$ which takes values in $\mathscr{T}$. This acts as the independent variable in $\pain{V}$, and $\mathscr{T}$ is the complement in $\C$ of the fixed singularities of $\pain{V}$.

The surfaces $X_{\boldsymbol{a},t}$ for $\boldsymbol{a}\in\mathscr{A}$, $t\in\mathscr{T}$, form a family of Sakai surfaces of type $D_5^{(1)}$ parametrised by root variables, which we will write as 
$$\mathcal{X}\to\mathscr{A}\times\mathscr{T}.$$
We denote the composition of the projection morphisms of the eight blowups by 
$$\pi_{\boldsymbol{a},t} : X_{\boldsymbol{a},t} \rightarrow \p^1\times\p^1.$$
The Picard group of the surface $X_{\boldsymbol{a},t}$, which we think of as the group of classes of divisors on $X_{\boldsymbol{a},t}$ up to linear equivalence, can be written as 
\begin{equation} \label{eq:picXa}
    \Pic(X_{\boldsymbol{a},t}) = \Span_{\Z}\left\{\h_{q},\h_{p},\E_1,\dots,\E_8\right\},
\end{equation}
where $\h_q$ and $\h_p$ are classes of (pullbacks under $\pi_{\boldsymbol{a},t}$ of) fibres of the projections from $\p^1\times\p^1$ onto the two $\p^1$-factors, and $\E_i=[E_i]$ corresponds to the class of the exceptional divisor of the blowup of $b_i$, or more precisely, its pullback under any subsequent blowups. 
We give an illustration of the configuration of points $b_1,\dots,b_8$ and the surface $X_{\boldsymbol{a},t}$, for generic $\boldsymbol{a}$ and $t$, on Figure \ref{fig:KNY-soic-5}, with curves labeled by their classes in $\Pic(X_{\boldsymbol{a},t})$.

\begin{figure}[htb]
	\begin{tikzpicture}[>=stealth,basept/.style={circle, draw=red!100, fill=red!100, thick, inner sep=0pt,minimum size=1.2mm}]
	\begin{scope}[xshift=0cm,yshift=0cm]
	\draw [black, line width = 1pt] (-0.4,0) -- (2.9,0)	node [pos=0,left] {
 \small $p=0$
 } node [pos=1,right] {\small $p=0$};
	\draw [black, line width = 1pt] (-0.4,2.5) -- (2.9,2.5) node [pos=0,left] {
 \small $p=\infty$
 } node [pos=1,right] {\small $p=\infty$};
	\draw [black, line width = 1pt] (0,-0.4) -- (0,2.9) node [pos=0,below] {
  \small $q=0$
 } node [pos=1,above] {\small $q=0$};
	\draw [black, line width = 1pt] (2.5,-0.4) -- (2.5,2.9) node [pos=0,below] {
  \small $q=\infty$
 } node [pos=1,above] {\small $q=\infty$};
	\node (p3) at (2.5,0) [basept,label={[xshift = -8pt, yshift=-18pt] \small $b_{3}$}] {};
	\node (p4) at (3,0.5) [basept,label={[yshift=0pt] \small $b_{4}$}] {};
	\node (p1) at (2.5,1) [basept,label={[xshift = -8pt, yshift=-15pt] \small $b_{1}$}] {};
	\node (p2) at (3,1.5) [basept,label={[yshift=0pt] \small $b_{2}$}] {};
	\node (p5) at (0,2.5) [basept,label={[xshift = 8pt, yshift=0pt] \small $b_{5}$}] {};
	\node (p6) at (-.5,2) [basept,label={[yshift=-18pt] \small $b_{6}$}] {};
	\node (p7) at (1.5,2.5) [basept,label={[xshift = 8pt, yshift=0pt] \small $b_{7}$}] {};
	\node (p8) at (1,2) [basept,label={[yshift=-18pt] \small $b_{8}$}] {};
	\draw [red, line width = 0.8pt, ->] (p2) -- (p1);
	\draw [red, line width = 0.8pt, ->] (p4) -- (p3);
	\draw [red, line width = 0.8pt, ->] (p6) -- (p5);
	\draw [red, line width = 0.8pt, ->] (p8) -- (p7);	
        \node (P1P1) at (1.25,-1.25) {$\p^1\times\p^1$};
	\end{scope}
	\draw [->] (6.5,1)--(4.5,1) node[pos=0.5, below] 
    {$\pi_{\boldsymbol{a},t}$};
	\begin{scope}[xshift=9cm,yshift=0cm]
	\draw [red, line width = 1pt] (-0.4,0) -- (3.5,0)	node [pos=0, left] {\small $\h_{p}-\E_{3}$};
	\draw [red, line width = 1pt] (0,-0.4) -- (0,2.4) node [pos=0, below] {\small $\h_{q}-\E_{5}$};
	\draw [blue, line width = 1pt] (-0.2,1.8) -- (0.8,2.8) node [pos=0, left] {\small $\E_{5}-\E_{6}$};
	\draw [red, line width = 1pt] (-0.1,2.7) -- (0.4,2.2) node [pos=0, above] {\small $\E_{6}$};
	\draw [blue, line width = 1pt] (1.2,1.8) -- (2.2,2.8) node [pos=0, xshift=-14pt, yshift=-5pt] {\small $\E_{7}-\E_{8}$};
	\draw [red, line width = 1pt] (1.6,2.4) -- (2.1,1.9) node [pos=1, below] {\small $\E_{8}$};
	\draw [blue, line width = 1pt] (0.3,2.6) -- (4.2,2.6) node [pos=1,right] {\small $\h_{p} - \E_{5} - \E_{7}$};
	\draw [blue, line width = 1pt] (3,-0.2) -- (4,0.8) node [pos=1,right] {\small $\E_{3} - \E_{4}$};
	\draw [red, line width = 1pt] (3.4,0.4) -- (3.9,-0.1) node [pos=1, below] {\small $\E_{4}$};
	\draw [blue, line width = 1pt] (3.8,0.3) -- (3.8,3) node [pos=1, above] {\small $\h_{q}-\E_{1} - \E_{3}$};	
	\draw [blue, line width = 1pt] (3,1) -- (4,2) node [pos=1,right] {\small $\E_{1} - \E_{2}$};
		\draw [red, line width = 1pt] (3.1,1.9) -- (3.6,1.4) node [pos=0, above] {\small $\E_{2}$};
	\draw [red, line width = 1pt] (-0.4,1.2) -- (3.5,1.2)	node [pos=0, left] {\small $\h_{p}-\E_{1}$};
	\draw [red, line width = 1pt] (1.4,-0.4) -- (1.4,2.4) node [pos=0, below] {\small $\h_{q}-\E_{6}$};
        \node (surfacelabel) at (2,-1.25) {$\mathcal{X}_{\boldsymbol{a},t}$};
	\end{scope}
	\end{tikzpicture}
	\caption{Blowup point configuration and surface $\mathcal{X}_{\boldsymbol{a},t}$ in the KNY model of $D_5^{(1)}$ surfaces}
	\label{fig:KNY-soic-5}
\end{figure}

For any $\boldsymbol{a}\in\mathscr{A}$, $t\in\mathscr{T}$, the surface $X_{\boldsymbol{a},t}$ is a Sakai surface, i.e. a generalised Halphen surface with unique effective anticanonical divisor, see \cite[Def. 3.21 \& Def. 3.22]{ASIDEnotes}.
Here this divisor is given by $D =  - \operatorname{div}\omega_{\boldsymbol{a},t}$, where $\omega_{\boldsymbol{a},t}$ is the rational 2-form on $X_{\boldsymbol{a},t}$ given by the pullback under $\pi_{\boldsymbol{a},t}$ of the 2-form on $\p^1\times\p^1$ given in the four charts in \Cref{eq:P1P1atlasKNY} by
$    dq \wedge dp = - \frac{dQ \wedge dp}{Q^2} = - \frac{dq \wedge dP}{P^2} =  \frac{dQ \wedge dP}{Q^2P^2}$.
In $\Pic(X_{\boldsymbol{a},t})$, the linear equivalence class of $D$ gives the anticanonical divisor class 
$-\mathcal{K}_{X_{\boldsymbol{a},t}}=2\h_q+2\h_p-\E_1-\dots-\E_8$.
\begin{remark}
    If, instead of $a_0+a_1+a_2+a_3=1$, the parameters were subject to $a_0+a_1+a_2+a_3=0$, there would be a pencil of biquadratic curves on $\p^1\times\p^1$ passing through the eight points $b_1,\dots,b_8$ and we would get $\dim|-\mathcal{K}_{X_{\boldsymbol{a},t}}|=1$ and have a Halphen surface of index 1, i.e. a particular kind of rational elliptic surface, with elliptic fibration coming from the linear system $|-\ell \mathcal{K}_{X_{\boldsymbol{a},t}}|$, with $\ell=1$.
\end{remark}

\subsection{Root data and Cremona isometries}
\label{subsec:rootdataandcremonaisometries}

We now introduce bases of simple roots for the surface and symmetry root lattices coming from this realisation of Sakai surfaces of surface type $D_5^{(1)}$, which we use in the algebraic description of symmetries and discrete Painlev\'e equations below.
The geometric realisation of the surfaces in the family $\mathcal{X}$ is interpreted as carrying the data of the sequence of blowups and their enumeration. This gives an identification of all the $\Pic(X_{\boldsymbol{a},t})$ into a single lattice, see \cite[Def. 3.74]{ASIDEnotes}, which we denote by 
$$\Pic(\X) = \Span_{\Z}\left\{\h_{q},\h_{p},\E_1,\dots,\E_8\right\},$$ abusing notation in using the same symbols as for $\Pic(\mathcal{X}_{\boldsymbol{a},t})$ in \Cref{eq:picXa}, including for the anticanonical class $$-\K_{\X}:=2\h_q+2\h_p-\E_1-\E_2-\E_3-\E_4-\E_5-\E_6-\E_7-\E_8\in\Pic(\X),$$
and the symmetric, bilinear product given by the intersection form, defined by extension of 
$$        \h_q\bullet\h_p=1, \quad \h_q\bullet\h_q=\h_p\bullet\h_p=\h_q\bullet\E_i=\h_p\bullet\E_i=0, \quad \E_i\bullet\E_j=-\delta_{ij}, \quad\text{ for }~i,j=1,\dots,8.
$$
\begin{remark}
    Strictly speaking, we are abusing notation here in more than one way, since $\Pic(\X)$ is not the Picard group of the family $\X \to \mathscr{A} \times\mathscr{T}$, but rather an identification of Picard groups of all of its fibres.
\end{remark}
\subsubsection{Surface and symmetry root lattices}
The surface root basis is formed of elements $\delta_i \in \Pic(\X)$ corresponding to classes of the components $D_i$, $i=0,\dots,5$, of the unique effective anticanonical divisor 
$$D=D_0+D_1+2D_2+2D_3+D_4+D_5,$$ 
for each $\X_{\boldsymbol{a},t}$ in the family.
These are indicated in blue on \Cref{fig:KNY-soic-5}.
The elements $\delta_i\in\Pic(\X)$ corresponding to classes of the components $D_i$ are given in \Cref{fig:d-roots-d5-KNY} and play the role of simple roots for the affine root system of type $D_5^{(1)}$, with null root corresponding to the anticanonical class:
$$\delta = \delta_0+\delta_1+2\delta_2+2\delta_3+\delta_4+\delta_5 = - \mathcal{K}_{\X}.$$
These span the surface root lattice $Q(\mathcal{R})$, $\mathcal{R}=D_5^{(1)}$ contained in the orthogonal complement of $\delta$:
\begin{equation}
    Q=Q(D_5^{(1)}) = \Span_{\Z}\left\{ \delta_0, \dots, \delta_5\right\} \subset \delta^{\perp}\subset \Pic(\X).
\end{equation}

\begin{figure}[htb]
\begin{equation*}\label{eq:d-roots-d51}			
	\raisebox{-32.1pt}{\begin{tikzpicture}[
			elt/.style={circle,draw=black!100,thick, inner sep=0pt,minimum size=2mm}]
		\path 	(-1,1) 	node 	(d0) [elt, label={[xshift=-10pt, yshift = -10 pt] $\delta_{0}$} ] {}
		        (-1,-1) node 	(d1) [elt, label={[xshift=-10pt, yshift = -10 pt] $\delta_{1}$} ] {}
		        ( 0,0) 	node  	(d2) [elt, label={[xshift=-10pt, yshift = -11 pt] $\delta_{2}$} ] {}
		        ( 1,0) 	node  	(d3) [elt, label={[xshift=10pt, yshift = -11 pt] $\delta_{3}$} ] {}
		        ( 2,1) 	node  	(d4) [elt, label={[xshift=10pt, yshift = -10 pt] $\delta_{4}$} ] {}
		        ( 2,-1) node 	(d5) [elt, label={[xshift=10pt, yshift = -10 pt] $\delta_{5}$} ] {};
		\draw [black,line width=1pt ] (d0) -- (d2) -- (d1)  (d2) -- (d3) (d4) -- (d3) -- (d5);
	\end{tikzpicture}} \qquad
			\begin{aligned}
            &\begin{alignedat}{2}
			\delta_{0} &= \mathcal{E}_{1} - \mathcal{E}_{2}, &\qquad  \delta_{3} &= \mathcal{H}_{p} - \mathcal{E}_{5} - \mathcal{E}_{7},\\
			\delta_{1} &= \mathcal{E}_{3} - \mathcal{E}_{4}, &\qquad  \delta_{4} &= \mathcal{E}_{5} - \mathcal{E}_{6},\\
			\delta_{2} &= \mathcal{H}_{q} - \mathcal{E}_{1} - \mathcal{E}_{3}, &\qquad  \delta_{5} &= \mathcal{E}_{7} - \mathcal{E}_{8},
			\end{alignedat}
           \\[5pt]
           &~\,\delta  =\delta_0+\delta_1+2\delta_2+2\delta_3+\delta_4+\delta_5.
           \end{aligned}
\end{equation*}
	\caption{The surface root basis for the KNY reference model of $D_{5}^{(1)}$ surfaces}
	\label{fig:d-roots-d5-KNY}	
\end{figure}
The sublattice of $\Pic(\X)$ orthogonal to $Q$ with respect to the intersection form is the symmetry root lattice $Q^{\perp}=Q(\mathcal{R}^{\perp})$, $\mathcal{R}^{\perp}=A_3^{(1)}$, which can be written in terms of simple roots for an affine root system of type $A_3^{(1)}$ as
\begin{equation}
    Q^{\perp}=Q(A_3^{(1)}) = \Span_{\Z}\left\{ \alpha_0, \dots, \alpha_3\right\} \subset \delta^{\perp}\subset \Pic(\X),
\end{equation}
where the simple roots $\alpha_i$ are given in \Cref{fig:a-roots-a3-KNY}.

\begin{figure}[htb]
\begin{equation*}\label{eq:a-roots-a31}			
	\raisebox{-32.1pt}{\begin{tikzpicture}[
			elt/.style={circle,draw=black!100,thick, inner sep=0pt,minimum size=2mm}]
		\path 	(-1,1) 	node 	(a0) [elt, label={[xshift=-10pt, yshift = -10 pt] $\alpha_{0}$} ] {}
		        (-1,-1) node 	(a1) [elt, label={[xshift=-10pt, yshift = -10 pt] $\alpha_{1}$} ] {}
		        ( 1,-1) node  	(a2) [elt, label={[xshift=10pt, yshift = -10 pt] $\alpha_{2}$} ] {}
		        ( 1,1) 	node 	(a3) [elt, label={[xshift=10pt, yshift = -10 pt] $\alpha_{3}$} ] {};
		\draw [black,line width=1pt ] (a0) -- (a1) -- (a2) --  (a3) -- (a0); 
	\end{tikzpicture}} \qquad
			\begin{alignedat}{2}
			\alpha_{0} &= \mathcal{H}_{p} - \mathcal{E}_{1} - \mathcal{E}_{2}, &\qquad  \alpha_{3} &= \mathcal{H}_{q} - \mathcal{E}_{7} - \mathcal{E}_{8},\\
			\alpha_{1} &= \mathcal{H}_{q} - \mathcal{E}_{5} - \mathcal{E}_{6}, &\qquad  \alpha_{2} &= \mathcal{H}_{p} - \mathcal{E}_{3} - \mathcal{E}_{4},
			\\[5pt]
			\delta & = \mathrlap{\alpha_{0} + \alpha_{1} +  \alpha_{2} + \alpha_{3}.} 
			\end{alignedat}
\end{equation*}
	\caption{The symmetry root basis for the KNY reference model of $D_{5}^{(1)}$ surfaces}
	\label{fig:a-roots-a3-KNY}	
\end{figure}

\subsubsection{Extended affine Weyl group acting on $\Pic(\X)$}

The symmetry root basis introduced above is used in the description of the symmetries of the family $\mathcal{X}$ as an extended affine Weyl group, on the level of linear maps on $\Pic(\X)$. 
The affine Weyl group of type $A_3^{(1)}$ has the following presentation by generators and relations, encoded in the Dynkin diagram as in \Cref{fig:a-roots-a3-KNY}:
\begin{equation}
    W(A_3^{(1)}) = \left\langle w_0,w_1,w_2,w_3 ~~|~~
    w_{i}^{2} = 1,  \quad
    % \raisebox{-0.1in}{
    \begin{aligned}
        (w_{i} w_{j})^2 &= 1&  &\text{ when }
   				\raisebox{-0.15in}{\begin{tikzpicture}[
   							elt/.style={circle,draw=black!100,thick, inner sep=0pt,minimum size=1.5mm}]
   						\path   ( 0,0) 	node  	(ai) [elt] {}
   						        ( 0.5,0) 	node  	(aj) [elt] {};
   						\draw [black] (ai)  (aj);
   							\node at ($(ai.south) + (0,-0.2)$) 	{\small ${\alpha_i}$};
   							\node at ($(aj.south) + (0,-0.2)$)  {\small ${\alpha_j}$};
   							\end{tikzpicture}}\\
    (w_{i} w_{j})^3 &= 1& &\text{ when }
   				\raisebox{-0.15in}{\begin{tikzpicture}[
   							elt/.style={circle,draw=black!100,thick, inner sep=0pt,minimum size=1.5mm}]
   						\path   ( 0,0) 	node  	(ai) [elt] {}
   						        ( 0.5,0) 	node  	(aj) [elt] {};
   						\draw [black] (ai) -- (aj);
   							\node at ($(ai.south) + (0,-0.2)$) 	{\small${\alpha_{i}}$};
   							\node at ($(aj.south) + (0,-0.2)$)  {\small${\alpha_{j}}$};
   							\end{tikzpicture}}
    \end{aligned}
    % }
    \right\rangle.
\end{equation}
Here $w_i=r_{\alpha_i}$ is the simple reflection associated to the simple root $\alpha_i$.
    We have an action of $W(A_3^{(1)})$ on $\Pic(\X)$, 
    where the reflections $r_{\alpha_i}$ act according to the formula
\begin{equation}
          r_{\alpha_i} (\F)= \F -2 \frac{\F\bullet\alpha_i }{\alpha_i\bullet\alpha_i}\alpha_i = \F + (\F\bullet\alpha_i )\alpha_i, \quad \F\in\Pic(\X),
\end{equation}   
where $\alpha_i$, $i=0,1,2,3$, are as provided in \Cref{fig:a-roots-a3-KNY}.
This action preserves the intersection form $\bullet$, the anticanonical class $\delta=-\K_{\X}$, and the surface root lattice $Q\subset \Pic(\X)$.

In addition to $W(A_3^{(1)})$, the KNY family of $D_5^{(1)}$ surfaces admits an action of the automorphism group of the $A_3^{(1)}$ Dynkin diagram.
This group $\Aut(A_3^{(1)})$ of graph automorphisms of the $A_3^{(1)}$ Dynkin diagram in \Cref{fig:a-roots-a3-KNY} is isomorphic to the dihedral group $\mathbb{D}_{4}$ of order 8, which we describe as permutations of simple roots $\alpha_0,\alpha_1,\alpha_2,\alpha_3$.
As generators we choose $\sigma_1$ and $\sigma_2$ as follows:
\begin{equation}
		\Aut(A_3^{(1)}) = \langle \sigma_1,\sigma_2\rangle \cong \mathbb{D}_4, \qquad \begin{aligned}
		    \sigma_{1} &= (\alpha_{0}\alpha_{3})(\alpha_{1}\alpha_{2}), \\
            \sigma_{2} &= (\alpha_{0}\alpha_{2}).
		\end{aligned}		
\end{equation}
Their actions on $\Pic(\mathcal{X})$, which act on the simple roots by these permutations, are 
\begin{equation}
    \sigma_1 = (\h_q \h_p)(\E_1 \E_7)(\E_2 \E_8)(\E_3\E_5)(\E_4 \E_6), \quad 
    \sigma_2 = (\E_1 \E_3)(\E_2 \E_4),
\end{equation}
where we again use cycle notation for permutations.
Note that, in constrast to the reflections $w_i$, the Dynkin diagram automorphisms permute the surface roots $\delta_j$.
Together with the action of $W(A_3^{(1)})$ above, this gives an action on $\Pic(\X)$ of the (fully) extended affine Weyl group of type $A_3^{(1)}$,
\begin{equation}
    \widehat{W}(A_3^{(1)}) := W(A_3^{(1)})\rtimes \Aut(A_3^{(1)}),
\end{equation}
where the semi-direct product structure comes from $\sigma w_i \sigma^{-1} = \sigma r_{\alpha_i} \sigma^{-1} = r_{\sigma(\alpha_i)}$, for $\sigma\in\Aut(A_3^{(1)})$.

\subsubsection{Translations}

The group $\widehat{W}(A_3^{(1)})$ contains a normal subgroup of translations, which is contained in a smaller extension of $W(A_3^{(1)})$ by only certain Dynkin diagram automorphisms, which we describe now.
The subgroup of these Dynkin diagram automorphisms in this case is  
\begin{equation}
    \Sigma = \langle \rho := \sigma_1 \sigma_2\rangle \cong \Z/4\Z,    \qquad \rho= (\alpha_0\alpha_1\alpha_2\alpha_3).
\end{equation}
which acts on $\Pic(\mathcal{X})$ according to 
\begin{equation} \label{eq:rhoonpic}
\rho=(\h_q \h_p)(\E_1 \E_5 \E_3 \E_7)(\E_2 \E_6 \E_4 \E_8).
\end{equation}
The extended affine Weyl group of type $A_3^{(1)}$ is
\begin{equation}
    \widetilde{W}(A_3^{(1)}) := W(A_3^{(1)})\rtimes \Sigma \cong W(A_3)\ltimes T_{\nought{P}},
\end{equation}
where $W(A_3)=\langle w_1,w_2,w_3\rangle\subset W(A_3^{(1)})$ is the finite Weyl group of type $A_3$, and $T_{\nought{P}}$ is a subgroup of translations corresponding to the weight lattice of the underlying finite root system. 
The root lattice of the underlying finite root system is $\nought{Q}^{\perp} := \Z \alpha_1+\Z \alpha_1+\Z\alpha_2 \subset Q^{\perp}$, and the weight lattice is
\begin{equation}
    \nought{P} := \left\{ h\in \nought{Q}^{\perp}\otimes \mathbb{Q} ~|~ h\bullet\alpha\in\Z \text{ for all } \alpha\in \nought{Q}^{\perp} \right\}=\Z h_1+\Z h_2+\Z h_3\supset \nought{Q}^{\perp},
\end{equation}
where the fundamental weights $h_{i}$ are 
\begin{equation*}\label{A3w}
    h_1 = \frac{3}{4} \alpha_1 + \frac{1}{2} \alpha_{2} + \frac{1}{4} \alpha_3, \quad
h_2 = \frac{1}{2} \alpha_1 +  \alpha_{2} + \frac{1}{2} \alpha_3, 
\quad
h_3 = \frac{1}{4} \alpha_1 + \frac{1}{2} \alpha_{2} + \frac{3}{4} \alpha_3,
\end{equation*}
so $\alpha_i\bullet h_j=-\delta_{i,j}$.
The associated translations act on $Q^{\perp}$ by 
\begin{equation} \label{eq:translationformula}
    T_{h}(\alpha)=\alpha - (\alpha\bullet h)\delta, \quad \alpha\in Q^{\perp},  h\in \nought{P}.
\end{equation}
Further, $T_{ h}T_{ h'}=T_{ h+ h'}$ for any $ h, h'\in \nought{P}$, and $T_{\nought{P}}=\langle T_{ h_1}, T_{ h_2},T_{ h_3}\rangle$.
On $\Pic(\X)$, the actions of the translations associated to fundamental weights are given by the following compositions of simple reflections and $\rho$ as in \eqref{eq:rhoonpic}:
\begin{equation} \label{eq:translationwords}
         T_{ h_1} = \rho^3 w_{2} w_{3} w_{0}, 
         \quad 
         T_{ h_2} = \rho^2 w_0 w_3 w_1 w_0, 
         \quad 
         T_{ h_3} = \rho w_2 w_1 w_0.
\end{equation}
Note that the action of $T_{\nought{P}}$ on the whole of $\Pic(\X)$, not just $Q^{\perp}$, is not given in general by the formula \eqref{eq:translationformula}, but rather via the expressions in \Cref{eq:translationwords} using the action of $\rho$ on $\Pic(\X)$ given in \Cref{eq:rhoonpic}.
\begin{remark} \label{rem:conjugacynecessaryconditionweight}
    Conjugacy classes of translations in $\widehat{W}(A_3^{(1)})$ correspond to orbits of weights modulo $\Z \delta$, and $T_{w( h)}=w T_h w^{-1}$ for $h\in \nought{P}$, $w\in W(A_3)$.
    In particular, this gives a necessary condition for conjugacy of translations: if two weights $h, h'\in \nought{P}$ are such that $h\bullet h\neq h'\bullet h'$, then $T_h$ and $T_{h'}$ are not conjugate in $\widehat{W}(A_3^{(1)})$.
\end{remark}
\subsection{Cremona action and discrete Painlev\'e equations}

The action of $\widehat{W}(A_3^{(1)})=W(A_3^{(1)})\rtimes \operatorname{Aut}(A_3^{(1)})$ on $\Pic(\X)$ in \Cref{subsec:rootdataandcremonaisometries} is realised by an action of $\widehat{W}(A_3^{(1)})$ on the family $\mathcal{X}\to\mathscr{A}\times\mathscr{T}$, called the Cremona action \cite{SAKAI2001}.
In particular, for each $w\in \widehat{W}(A_3^{(1)})$, there is an action on parameters $\mathscr{A}\times\mathscr{T}\to\mathscr{A}\times\mathscr{T}, (\boldsymbol{a},t)\mapsto(\tilde{\boldsymbol{a}},\tilde{t})$ and isomorphisms $X_{\boldsymbol{a},t}\to X_{\tilde{\boldsymbol{a}},\tilde{t}}$, which induces the action of $w$ on $\Pic(\X)$ by pushforward or pullback depending on convention (see \cite[Sec. 3.6.3]{ASIDEnotes} for details).
In the following lemma we give this action, in the form of the actions of generators of $W(A_3^{(1)})$ and $\Aut(A_3^{(1)})$ on parameters as $\boldsymbol{a}\mapsto\tilde{\boldsymbol{a}}$, as well as the birational map $\C^2\dashrightarrow\C^2, (q,p)\mapsto (\tilde{q},\tilde{p})$ which gives the isomorphism $X_{\boldsymbol{a},t}\to X_{\tilde{\boldsymbol{a}},\tilde{t}}$ when extended to $\p^1\times \p^1$ and lifted under the blowups.

\begin{lemma}\label{thm:bir-weyl-d5} 
	The Cremona action of $\widehat{W}(A_3^{(1)})$ on the family $\mathcal{X}\to\mathscr{A}\times\mathscr{T}$ is given as follows.
    The simple reflections $w_{i}\in W(A_3^{(1)})$ act according to 
	\begin{alignat*}{2}
		w_{0}&: 
		\left(\begin{matrix} a_{0} & a_{1} \\ a_{2} & a_{3} \end{matrix}\ ;\ t\ ;
		\begin{matrix} q \\ p \end{matrix}\right) 
		&&\mapsto 
		\left(\begin{matrix} -a_{0} & a_{0} + a_{1} \\ a_{2} & a_{0} + a_{3} \end{matrix}\ ;\ t\ ;
		\begin{matrix} \displaystyle q + \frac{a_{0}}{p + t} \\ p \end{matrix}\right)\!, \\
		w_{1}&: \left(\begin{matrix} a_{0} & a_{1} \\ a_{2} & a_{3} \end{matrix}\ ;\ t\ ;
		\begin{matrix} q \\ p \end{matrix}\right)
		&&\mapsto 
		\left(\begin{matrix} a_{0} + a_{1} & -a_{1} \\ a_{1} + a_{2} & a_{3} \end{matrix}\ ;\ t\ ;
		\begin{matrix}  q \\ \displaystyle p - \frac{a_{1}}{q} \end{matrix}\right)\!, \\
		w_{2}&: 
		\left(\begin{matrix} a_{0} & a_{1} \\ a_{2} & a_{3} \end{matrix}\ ;\ t\ ;
		\begin{matrix} q \\ p \end{matrix}\right) 
		&&\mapsto 
		\left(\begin{matrix} a_{0} & a_{1} + a_{2} \\ -a_{2} & a_{2} + a_{3} \end{matrix}\ ;\ t\ ;
		\begin{matrix} \displaystyle q + \frac{a_{2}}{p}\\ p \end{matrix}\right)\!, \\
		w_{3}&: 
		\left(\begin{matrix} a_{0} & a_{1} \\ a_{2} & a_{3} \end{matrix}\ ;\ t\ ;
		\begin{matrix} q \\ p \end{matrix}\right) 
		&&\mapsto 
		\left(\begin{matrix} a_{0}+a_{3} & a_{1} \\ a_{2}+a_{3} & -a_{3} \end{matrix}\ ;\ t\ ;
		\begin{matrix} q \\ \displaystyle  p - \frac{a_{3}}{q-1} \end{matrix}\right)\!.
	\end{alignat*}	
    The actions of the generators $\sigma_1,\sigma_2$ of $\Aut(A_3^{(1)})$ are given by 
	\begin{alignat*}{2}
		\sigma_{1}&: 
		\left(\begin{matrix} a_{0} & a_{1} \\ a_{2} & a_{3} \end{matrix}\ ;\ t\ ;
		\begin{matrix} q \\ p \end{matrix}\right) 
		&&\mapsto 
		\left(\begin{matrix} a_{3} & a_{2} \\ a_{1} & a_{0}  \end{matrix}\ ;\ -t\ ;
		\begin{matrix} \displaystyle -\frac{p}{ t} \\ q t \end{matrix}\right)\!, \\
		\sigma_{2}&: 
		\left(\begin{matrix} a_{0} & a_{1} \\ a_{2} & a_{3} \end{matrix}\ ;\ t\ ;
		\begin{matrix} q \\ p \end{matrix}\right) 
		&&\mapsto 
		\left(\begin{matrix} a_{2} & a_{1} \\ a_{0} & a_{3}  \end{matrix}\ ;\ -t\ ;
		\begin{matrix} q \\ p + t \end{matrix}\right)\!.
	\end{alignat*}

\end{lemma}
Via the Cremona action, elements of infinite order in $\widehat{W}(A_3^{(1)})$, e.g. translations, define discrete Painlev\'e equations of surface type $D_5^{(1)}$.
There are two reference examples which will be relevant in this paper, which we present next.

\subsubsection{The KNY example of a discrete Painlev\'e equation of surface type $D_5^{(1)}$} 
\label{subsubsec:KNYexample}

The first reference example is that given in \cite[Sec. 8.1.17]{KajNouYam:2017:GAPE}, which comes from the Cremona action of the translation 
\begin{equation}
    T_{\operatorname{KNY}}  := T_{ h_1- h_2+ h_3}= T_{ h_1} T_{ h_2}^{-1} T_{ h_3} = \rho^2 w_1 w_3 w_2 w_0,
\end{equation}
whose associated weight $h_{\operatorname{KNY}}:=h_1-h_2+h_3$ has squared length $- (h_{\operatorname{KNY}}\bullet h_{\operatorname{KNY}})=1$.
We reproduce the expression of this in coordinates as $(q,p)\mapsto (\bar{q},\bar{p})$, with parameter evolution $\boldsymbol{a}=(a_0,a_1,a_2,a_3)\mapsto (\bar{a}_0,\bar{a}_1,\bar{a}_2,\bar{a}_3)=\bar{\boldsymbol{a}}$, as follows: 
\begin{equation} \label{eq-KNY-dP-S3}
	\left\{
	\begin{aligned}
		\bar{q} 	&= 1 - q - \frac{a_0}{p+t} - \frac{a_2}{p}, \\
		\ubar{p} 	&= - p - t + \frac{a_1}{q} + \frac{a_3}{q-1}, 
	\end{aligned}
	\qquad 
	\begin{aligned}
		\bar{a}_0 &= a_0 + 1, &&\bar{a}_1 = a_1 - 1, \\
		\bar{a}_2 &= a_2 + 1, &&\bar{a}_3 = a_3 - 1.
	\end{aligned}
	\right.
\end{equation}
Regarding this as a (family of) birational map(s) 
$$\varphi^{\operatorname{KNY}}_{\boldsymbol{a},t} : \p^1\times \p^1 \dashrightarrow\p^1 \times\p^1, \quad (q,p)\mapsto(\bar{q},\bar{p}),$$
lifting under the blowups to the surfaces gives the isomorphism(s) 
\begin{equation}
\tilde{\varphi}^{\operatorname{KNY}}_{\boldsymbol{a},t} : X_{\boldsymbol{a},t} \rightarrow X_{\bar{\boldsymbol{a}},t}, 
\end{equation}
which act on $\Pic(\X)$ via pushforward according to
\begin{equation}
T_{\operatorname{KNY}} : \left\{
	\begin{aligned}
		\mathcal{H}_q 	&\mapsto 5 \h_q + 2 \h_p - \E_1 -  \E_2 -  \E_3 -  \E_4 - 2\E_5 - 2\E_6 - 2\E_7-2\E_8, \\\
		\mathcal{H}_p 	&\mapsto  2\h_q + \h_p - \E_5 -\E_6 - \E_7 - \E_8, \\
		\E_1			&\mapsto  2\h_q + \h_p -\E_4 - \E_5 -\E_6 - \E_7 - \E_8, \\
		\E_2			&\mapsto  2\h_q + \h_p -\E_3 - \E_5 -\E_6 - \E_7 - \E_8, \\
		\E_3			&\mapsto  2\h_q + \h_p - \E_2-\E_5 -\E_6 - \E_7 - \E_8, \\
		\E_4			&\mapsto  2\h_q + \h_p -\E_1 - \E_5 -\E_6 - \E_7 - \E_8, \\
		\E_5			&\mapsto  \h_q-\E_8, \\
		\E_6			&\mapsto  \h_q-\E_7, \\
		\E_7			&\mapsto  \h_q-\E_6, \\
		\E_8			&\mapsto  \h_q-\E_5, 
	\end{aligned}
\right.
\end{equation}
the restriction of which to $Q^{\perp}$ gives
\begin{equation} \label{eq:translationonroots-KNY}
T_{\operatorname{KNY}} : 
\left\{
	\begin{aligned}
		\alpha_0 &\mapsto \alpha_0 - \delta, \\
		\alpha_1 &\mapsto \alpha_1 + \delta , \\
		\alpha_2 &\mapsto \alpha_2 - \delta, \\
		\alpha_3 &\mapsto \alpha_3 + \delta.	
	\end{aligned}
\right.
\end{equation}

\subsubsection{The Sakai example of a discrete Painlev\'e equation of surface type $D_5^{(1)}$} 
\label{subsubsec:Sakaiexample}

The second reference example corresponds to that labeled $\operatorname{d}$-$\pain{IV}$ in \cite[Sec. 7]{SAKAI2001}, which comes from the Cremona action of the translation 
\begin{equation}
    T_{\operatorname{Sak}}  := T_{ h_3}= \rho w_2 w_1 w_0,
\end{equation}
whose associated weight $h_{\operatorname{Sak}}:=h_3$ has squared length $- (h_{\operatorname{Sak}}\bullet h_{\operatorname{Sak}})=\frac{3}{4}$. 
So $T_{\operatorname{Sak}}$ and $T_{\operatorname{KNY}}$ are not conjugate in $\widehat{W}(A_3^{(1)})$, see \Cref{rem:conjugacynecessaryconditionweight}.

We reproduce the expression of this discrete Painlev\'e equation, recycling the notation for the action on variables $(q,p)\mapsto (\bar{q},\bar{p})$, and parameter evolution $\boldsymbol{a}=(a_0,a_1,a_2,a_3)\mapsto (\bar{a}_0,\bar{a}_1,\bar{a}_2,\bar{a}_3)=\bar{\boldsymbol{a}}$, as follows:
\begin{equation} \label{eq-Sakai-dP-S3}
	\left\{
	\begin{aligned}
		\bar{q} 	&= \frac{p+t}{t} \left( 1 - \frac{a_0 + a_1}{ a_0 + (p+t)q} \right) , \\
		\bar{p} 	&= \frac{1}{1-\bar{q}}\left( a_2 + \frac{(a_0+a_1)p}{t+p - t \bar{q}} \right),  
	\end{aligned}
	\qquad 
	\begin{aligned}
		\bar{a}_0 &= a_0 + 1, &&\bar{a}_1 = a_1, \\
		\bar{a}_2 &= a_2, &&\bar{a}_3 = a_3 - 1.
	\end{aligned}
	\right.
\end{equation}
Regarding this as a (family of) birational map(s) 
$$\varphi^{\operatorname{Sak}}_{\boldsymbol{a},t} : \p^1\times \p^1 \dashrightarrow\p^1 \times\p^1, \quad (q,p)\mapsto(\bar{q},\bar{p}),$$
this gives the isomorphism(s)
\begin{equation}
\varphi^{\operatorname{Sak}}_{\boldsymbol{a},t} : X_{\boldsymbol{a},t} \rightarrow X_{\bar{\boldsymbol{a}},t}, 
\end{equation}
which acts on $\Pic(\X)$ via pushforward of $\varphi^{\operatorname{Sak}}$ according to
\begin{equation}
T_{\operatorname{Sak}} : \left\{
	\begin{aligned}
		\mathcal{H}_q 	&\mapsto 3 \h_q + 2 \h_p  -  \E_3 -  \E_4 - \E_5 - \E_6 - 2\E_7-2\E_8, \\\
		\mathcal{H}_p 	&\mapsto  2\h_q + \h_p - \E_3 -\E_4 - \E_7 - \E_8, \\
		\E_1			&\mapsto  2\h_q + \h_p - \E_3 -\E_4 -\E_6 - \E_7 - \E_8, \\
		\E_2			&\mapsto  2\h_q + \h_p - \E_3 -\E_4 -\E_5 - \E_7 - \E_8, \\
		\E_3			&\mapsto  \h_q - \E_8, \\
		\E_4			&\mapsto  \h_q- \E_7, \\
		\E_5			&\mapsto  \h_q+\h_p-\E_4-\E_7-\E_8, \\
		\E_6			&\mapsto  \h_q+\h_p-\E_3-\E_7-\E_8, \\
		\E_7			&\mapsto  \E_1, \\
		\E_8			&\mapsto  \E_2,
	\end{aligned}
\right.
\end{equation}
the restriction of which to $Q^{\perp}$ gives
\begin{equation}\label{eq:translationonroots-Sakai}
T_{\operatorname{Sak}} : 
\left\{
	\begin{aligned}
		\alpha_0 &\mapsto \alpha_0 - \delta, \\
		\alpha_1 &\mapsto \alpha_1 , \\
		\alpha_2 &\mapsto \alpha_2 , \\
		\alpha_3 &\mapsto \alpha_3 + \delta.	
	\end{aligned}
\right.
\end{equation}

\subsubsection{Hamiltonian form of $\pain{V}$}

The family of surfaces $\mathcal{X}\to\mathscr{A}\times \mathscr{T}$ also provides the initial value spaces for the fifth Painlev\'e equation $\pain{V}$, in the sense of Okamoto \cite{Oka:1979:FAESOPCFPP}.
For the standard form of $\pain{V}$ as a scalar second-order ODE we take the one in \cite{Oka:1987:SPEIFPEP},
\begin{equation}\label{eq:P5-std}
	\frac{d^{2} w}{dt^2} = \left(\frac{1}{2w} + \frac{1}{w-1}\right)\left(\frac{dw}{dt}\right)^{2} - \frac{1}{t} \frac{dw}{dt} + 
		\frac{(w-1)^{2}}{t^{2}}\left(\alpha w + \frac{\beta}{w}\right) + \frac{\gamma}{t} w + \delta \frac{w(w+1)}{w-1},
\end{equation}
where $\alpha,\beta,\gamma,\delta\in \mathbb{C}$ are parameters, and we assume $\delta\neq 0$ since otherwise \Cref{eq:P5-std} reduces to the third Painlev\'e equation $\pain{III}$.
We consider the form of $\pain{V}$ as a non-autonomous Hamiltonian system as given in \cite{KajNouYam:2017:GAPE}, which is polynomial just as that provided by Okamoto for $\pain{V}$ in \cite{Oka:1980:PHAWPE} (see also \cite{Oka:1987:SPEIFPEP,malmquist} - for the relations between various Hamiltonian forms of Painlev\'e equations on the level of geometry see \cite{hamiltonians}).
This is given by 
\begin{equation}\label{eq-KNY-Ham5-sys}
	\left\{
	\begin{aligned}
		\frac{dq}{dt} &= \frac{\partial \operatorname{H}}{\partial p} =\frac{1}{t}\Big(q(q-1)(2p+t) - a_{1}(q-1) - a_{3}q\Big),\\
		\frac{dp}{dt} &= -\frac{\partial \operatorname{H}}{\partial q} =\frac{1}{t}\Big(p(p+t)(1-2q) + (a_{1} + a_{3})p - a_{2}t\Big),
	\end{aligned}
	\right.
\end{equation}
where the Hamiltonian is
\begin{equation}\label{eq:KNY-Ham5}
	\operatorname{H}(q,p;t) = \frac{1}{t}\Big(q (q-1) p(p+t) - (a_{1} + a_{3})q p + a_{1}p + a_{2} t q\Big).	
\end{equation}
The system \eqref{eq-KNY-Ham5-sys} reduces to the standard Painlev\'e equation \eqref{eq:P5-std} for the variable 
$w(t) = 1 -  \frac{1}{q(t)}$ with the parameter values
\begin{equation}\label{eq:root-pars-PV}
	\alpha = \frac{a_{1}^{2}}{2},\quad \beta = -\frac{a_{3}^{2}}{2},\quad \gamma = a_{0}-a_{2}, \quad \delta = -	\frac{1}{2}.
\end{equation}
For given $\boldsymbol{a}=(a_0,a_1,a_2,a_3)\in\mathscr{A}$, the Okamoto space of initial conditions at $t\in\mathscr{T}$ for system \eqref{eq-KNY-Ham5-sys} is obtained by removing the components of the unique effective anticanonical divisor from $X_{\boldsymbol{a},t}$,
i.e. the surface as in Figure \ref{fig:KNY-soic-5}, with blue curves removed.
The Cremona action in \Cref{thm:bir-weyl-d5} provides B\"acklund transformations for the system \eqref{eq-KNY-Ham5-sys}.

\section{The identification procedure}
\label{sec:identificationprocedure}

We now recall the procedure, proposed in \cite{DzhFilSto:2020:RCDOPWHWDPE}, for identifying discrete Painlev\'e equations written in non-standard form. 
The kind of discrete systems the procedure is applicable to are systems of two first-order difference equations for $(x_n,y_n)$, say
\begin{equation} \label{eq:discretesystemgeneral}
    x_{n+1}= f_n(x_n,y_n), \quad y_{n+1}= g_n(x_n,y_n), 
\end{equation}
such that if we regard $(x_n,y_n)$ and $(x_{n+1},y_{n+1})$ as coordinates for two copies of $\C^2$, the discrete system \eqref{eq:discretesystemgeneral} defines a birational map.
Further, we will be considering non-autonomous discrete systems containing parameters, i.e. the functions $f_n,g_n$ in \Cref{eq:discretesystemgeneral} are rational in $x_n,y_n$ with coefficients being functions of $n$ as well as some complex parameters. 

In the case of systems coming from orthogonal polynomials there is often a differential system satisfied by the recurrence coefficients with respect to some parameter in the weight, say $s\in \mathscr{S}$. 
Here $\mathscr{S}\subset \C$ is the domain of analyticity of this differential system, playing the same role as the independent variable space $\mathscr{T}$ of the Painlev\'e equation.

We consider \Cref{eq:discretesystemgeneral} as defining a family of birational maps
\begin{equation}
    \psi_{\boldsymbol{b}, s} : \C^2 \dashrightarrow \C^2, \quad (x,y)\mapsto (\bar{x},\bar{y}),
\end{equation}
where $(x,y)=(x_{n},y_{n})$ and $(\bar{x},\bar{y})=(x_{n+1},y_{n+1})$. 
We take $\boldsymbol{b} \in\mathscr{B}$, where $\mathscr{B}$ is the space of parameters in the equation \eqref{eq:discretesystemgeneral}, which is a subset of $\C^{r}$ for some $r$, along with the parameter evolution $\boldsymbol{b}\mapsto\bar{\boldsymbol{b}}$ induced by $n \mapsto n+1$.

Suppose such a system is indeed a discrete Painlev\'e equation in disguise, related by a birational change of variables to a standard form coming from the Cremona action on a family of Sakai surfaces parametrised by root variables. 
If the discrete Painlev\'e equation in standard form corresponds, in the same way as the examples in \Cref{subsubsec:KNYexample,subsubsec:Sakaiexample}, to a family of birational maps
\begin{equation}
    \varphi_{\boldsymbol{a},t} : \C^2 \dashrightarrow \C^2, \quad (q,p)\mapsto (\bar{q},\bar{p}),
\end{equation}
with evolution of root variables $\boldsymbol{a}\mapsto \bar{\boldsymbol{a}}$, then the change of variables is a family $\upsilon_{\boldsymbol{b},s} : \C^2 \dashrightarrow\C^2$ of birational maps and an identification of parameters $\rho : \mathscr{B} \times \mathscr{S} \to\mathscr{A} \times \mathscr{T}$ that matches $\varphi_{\boldsymbol{a},t}$ and $\psi_{\boldsymbol{b},s}$. 
That is, the following diagrams commute:
\begin{equation} \label{eq:changeofvarsandparamident}
    \begin{tikzcd}
            \C^2 \arrow[r,dashed, "\psi_{\boldsymbol{b},s}"] \arrow[d,dashed, "\upsilon_{\boldsymbol{b},s}"]
        & \C^2   \arrow[d,dashed, "\upsilon_{\bar{\boldsymbol{b}},s}"]      \\
        \C^2  \arrow[r,dashed, "\varphi_{\boldsymbol{a},t}"]
        & \C^2 
    \end{tikzcd}
    \qquad 
    \begin{tikzcd}
            Y_{\boldsymbol{b},s} \arrow[r, "\tilde{\psi}_{\boldsymbol{b},s}"] \arrow[d, "\tilde{\upsilon}_{\boldsymbol{b},s}"]
        & Y_{\bar{\boldsymbol{b}},s}   \arrow[d, "\tilde{\upsilon}_{\bar{\boldsymbol{b}},s}"]      \\
        X_{\boldsymbol{a},t} \arrow[r,"\tilde{\varphi}_{\boldsymbol{a},t}"]
        & X_{\bar{\boldsymbol{a}},t} 
    \end{tikzcd}
    \qquad
    \begin{tikzcd}
        \mathscr{B} \times \mathscr{S} \arrow[r,"n\mapsto n+1"]  \arrow[d,"\rho"]
        &\mathscr{B} \times \mathscr{S}\arrow[d,"\rho"]\\
        \mathscr{A} \times \mathscr{T} \arrow[r,"{\boldsymbol{a}}\mapsto \bar{\boldsymbol{a}}"]
        &  \mathscr{A} \times \mathscr{T}.        
    \end{tikzcd}    
\end{equation}
In the situation under consideration, where there is a system of differential equations with respect to $s\in \mathscr{S}$ satisfied by $x_n,y_n$, the same change of variables will transform this to the Painlev\'e equation for $q,p$ with independent variable $t\in \mathscr{T}$.

The identification procedure to obtain the change of variables and parameter identification as in \cite{DzhFilSto:2020:RCDOPWHWDPE}, phrased with the notation and conventions above, is along the following lines.

\begin{description}
    \item[(Step 1)] \textbf{Construct the surfaces on which the system is regularised.}
    
    Obtain surfaces $Y_{\boldsymbol{b},s}$, of which the birational maps $\psi_{\boldsymbol{b},s}$ give isomorphisms $\tilde{\psi}_{\boldsymbol{b},s} : Y_{\boldsymbol{b},s}\rightarrow Y_{\bar{\boldsymbol{b}},s}$. 
    This is done by first extending the map $\psi_{\boldsymbol{b},s}$ to some compactification of $\C^2$, usually $\p^1 \times\p^1$, then performing blowups to resolve indeterminacies and contractions.
    Note that sometimes blowdowns, not just blowups, are required to obtain $Y_{\boldsymbol{b},s}$ as a Sakai surface from a given compactification (i.e. the initially obtained surfaces need to be minimised), but we will not encounter such examples in this paper.
    All the Picard groups $\Pic(Y_{\boldsymbol{b},s})$ are naturally identified according to the enumeration of blowups from the construction of $Y_{\boldsymbol{b},s}$ into a single lattice $\Pic(\mathcal{Y})$ similar to $\Pic(\X)$.
    
    \item[(Step 2)] \textbf{Determine the surface type according to the Sakai scheme, and choose a reference model of surfaces of this type.} 
    
    Determine the surface type $\mathcal{R}$ of the Sakai surface $Y_{\boldsymbol{b},s}$. 
    This is done by identifying curves of self-intersection $-2$ that form the irreducible components of the unique effective anticanonical divisor.
    Choose a surface root basis in $\Pic(\mathcal{Y})$ corresponding to the classes of these irreducible components.
    Note that there may be other $(-2)$-curves, and if so these will be relevant in determining the symmetry type of the system. 
    Choose a reference model for surfaces of type $\mathcal{R}$, i.e. a family of Sakai surfaces parametrised by root variables, with data of surface and symmetry root bases as well as Cremona action of the generic symmetry group, e.g. from those in \cite{KajNouYam:2017:GAPE} or \cite{SAKAI2001}.
    
    \item[(Step 3)] \textbf{Find the induced linear action on the Picard lattice.} 
    
    Find the pushforward $(\tilde{\psi}_{\boldsymbol{b},s})_* : \Pic(Y_{\boldsymbol{b},s})\rightarrow \Pic(Y_{\bar{\boldsymbol{b}},s})$ of  the isomorphism $\tilde{\psi}_{\boldsymbol{b},s}$ and the induced automorphism $\Psi : \Pic(\mathcal{Y})\to \Pic(\mathcal{Y})$. 
    
    \item[(Step 4)] \textbf{Find an identification with the reference model on the level of the Picard lattice.}

    Find an isomorphism of lattices $\Upsilon : \Pic(\mathcal{Y})\to \Pic(\X)$ that matches matches $\Psi$ with the action on $\Pic(\X)$ of some element $w\in W(\mathcal{R}^{\perp})\rtimes \Aut(\mathcal{R}^{\perp})$, i.e. such that the following diagram commutes:
    \begin{equation} \label{eq:Picidentcommutativediagram}
        \begin{tikzcd}
            \Pic(\mathcal{Y}) \arrow[r, "\Psi"] \arrow[d, "\Upsilon"] & \Pic(\mathcal{Y}) \arrow[d,"\Upsilon"] \\
            \Pic(\X) \arrow[r,"w"] & \Pic(\X).
        \end{tikzcd}
    \end{equation}
    This is done by first finding a preliminary matching of $\Pic(\mathcal{Y})$ and $\Pic(\X)$, which preserves the intersection forms and effectiveness of divisor classes as well as matches the surface root bases. 
    Then the action of $\Psi$ on the symmetry root basis under this matching can be compared with that of $w$ corresponding to a reference example, and the preliminary matching can be adjusted if necessary.
    
    \item[(Step 5)] \textbf{Find the change of variables and parameter identification.}
    
    Find an isomorphism $\tilde{\upsilon}_{\boldsymbol{b},s} : Y_{\boldsymbol{b},s}\rightarrow X_{\boldsymbol{a},t}$ that induces the identification $\Upsilon : \Pic(\mathcal{Y})\to \Pic(\X)$ above. 
    This requires the identification of parameters $\rho : \mathscr{B} \times \mathscr{S} \to\mathscr{A} \times \mathscr{T}$, which is obtained by computing the root variables of $Y_{\boldsymbol{b},s}$ for the symmetry root basis corresponding under $\Upsilon$ to that of the reference example.
    This will transform the discrete system on the family $Y_{\boldsymbol{b},s}$ to the discrete Painlev\'e equation on the family $X_{\boldsymbol{a},t}$ generated by the element $w\in W(\mathcal{R}^{\perp})\rtimes \Aut(\mathcal{R}^{\perp})$, and also the differential systems with respect to $t$ and $s$.
\end{description}

The result of this procedure is a change of variables and parameter identification as in \Cref{eq:changeofvarsandparamident}, with the isomorphism $\tilde{\upsilon}_{\boldsymbol{b},n,s} : Y_{\boldsymbol{b},n,s}\rightarrow X_{\boldsymbol{a},t}$ providing the birational change of variables $\upsilon_{\boldsymbol{b},s}$ as in \eqref{eq:changeofvarsandparamident} when written in the coordinates $(x,y)$ and $(q,p)$.

\subsection{Parameter constraints and obstructions to symmetries}
\label{subsec:paramconstraintsandobstructionstosymmetries}

The parameter space $\mathscr{B}$ may not exhaust all possible values of root variables, and thus not admit the whole Cremona action. 
Such examples occur in several contexts, where Sakai surfaces with constraints on root variables either appear naturally or are engineered.
Below we will describe two such situations which will be relevant in this paper.

\subsubsection{Projective reduction conditions}

In \cite{SAKAI2001}, discrete Painlev\'e equations were said to arise from translation elements of the corresponding extended affine Weyl group.
While elements of infinite order more generally are now widely regarded as also defining discrete Painlev\'e equations (and in fact some well-studied early examples are associated to non-translation elements \cite{dP2symmetry}), one benefit of working with translations is that the evolution of parameters (root variables) under these is such that the discrete dynamical systems can be neatly written as non-autonomous difference equations of additive, $q$-difference, or elliptic difference type. 

When considering an element of infinite order that is not a translation (which will necessarily be a `quasi-translation'), the evolution of root variables is not purely translational and to write the system down as an equation with $n$-dependent coefficients requires periodic terms.
There is a procedure \cite{KNT} called projective reduction, in which parameters are constrained to a subset on which a non-translation element acts translationally, so the corresponding discrete Painlev\'e equation can be written as a non-autonomous difference equation free from periodic $n$-dependence of coefficients.
We will refer to such parameter constraints as \emph{projective reduction conditions}.
See \cite{yangtranslations} for algebraic details of quasi-translations in affine Weyl groups that appear in the context of discrete Painlev\'e equations.

\subsubsection{Parameter constraints from nodal curves}

A second type of parameter constraint that arises naturally is that corresponding to surfaces that include nodal curves.
\begin{definition} \label{def:nodalcurve}
    Let $X$ be a Sakai surface with effective anticanonical divisor $D$. 
    A \emph{nodal curve} is a smooth rational curve of self-intersection $-2$ disjoint from the irreducible components of $D$.
    The set of classes of nodal curves is $\Delta^{\operatorname{nod}} \subset \Pic(X)$. 
\end{definition}
Note that the genus formula immediately implies that if a smooth rational curve of self-intersection $-2$ on a Sakai surface is not among the components of the anticanonical divisor then it must be disjoint from them.

We now explain parameter constraints corresponding to nodal curves. 
The root variables, which are used as parameters for families of Sakai surfaces as above, are defined in terms of a kind of period map, which is a function $\chi$ on the symmetry root lattice $Q^{\perp}$ defined with respect to a choice of 2-form which is free of zeroes on $X$, and whose divisor of poles gives $D$ (see \cite[Sec. 3.6.2]{ASIDEnotes} and \cite[Sec. 5]{SAKAI2001} for details). 
For additive surface types, like the case of $\mathcal{R}=D_5^{(1)}$ we consider in this paper, the period map takes values in $\C$ and the root variables are $\chi(\alpha_i)=a_i$.
The existence of nodal curves implies constrained root variables due to the following.
\begin{proposition}[\cite{SAKAI2001}, Prop. 22.] \label{prop:nodalcurvesrootvars}
    Let $X$ be a Sakai surface, of surface type $\mathcal{R}$ and symmetry type $\mathcal{R}^{\perp}$.
    Denote by $W^{\operatorname{nod}}\subset W(\mathcal{R}^{\perp})$ the subgroup generated by reflections with respect to the elements of $\Delta^{\operatorname{nod}}$. 
    Then $\alpha \in W^{\operatorname{nod}}\cdot\Delta^{\operatorname{nod}} \iff \chi(\alpha)=0$, where $\chi$ is the period map on $X$.
\end{proposition}

In particular, the surface $X$ having a single nodal curve, $\Delta^{\operatorname{nod}}=\{\alpha\}$, is equivalent to $\chi(\alpha)=0$.
Nodal curves present obstructions to symmetries of Sakai surfaces, and the symmetries compatible with the existence of some set of nodal curves will form a subgroup of the generic one for the corresponding surface type. 
This can be understood both in terms of the parameter constraints that arise as in \Cref{prop:nodalcurvesrootvars} and on the level of Cremona isometries (see \cite{SAKAI2001} for details).
We quote the part of the relevant result which pertains to the surface type $\mathcal{R}=D_5^{(1)}$.

\begin{theorem}[\cite{SAKAI2001}, Th. 26.]
    Let $X$ be a Sakai surface, of surface type $\mathcal{R}\neq A_6^{(1)}$, $A_7^{(1)}$, $A_7^{(1)'}$, $A_8^{(1)}$, $D_7^{(1)}$, $D_8^{(1)}$, and corresponding symmetry type $\mathcal{R}^{\perp}$.
    Then the group of Cremona isometries of $X$ is the setwise stabiliser 
    $$\operatorname{Cr}\left(X\right) \cong \left( W(\mathcal{R}^{\perp})\rtimes \Aut(\mathcal{R}^{\perp}) \right)_{\Delta^{\operatorname{nod}}} = \left\{ w \in W(\mathcal{R}^{\perp})\rtimes \Aut(\mathcal{R}^{\perp}) ~|~ w(\Delta^{\operatorname{nod}})=\Delta^{\operatorname{nod}}\right\}.$$
\end{theorem}

\begin{remark} \label{rem:RCGeqn}
    Note that projective reduction conditions and parameter constraints corresponding to nodal curves are not mutually exclusive. 
    For example, there is an elliptic discrete Painlev\'e equation constructed from the Q4 equation in the Adler-Bobenko-Suris list \cite{ABSlist} (also known as Adler's discrete Krichever-Novikov equation \cite{adlerdiscreteKN}) by reduction \cite{joshietal} and deautonomisation \cite{RCG}, whose symmetry group $W(F_4^{(1)})$ was established in \cite{AHJN}.
    The associated surface is non-generic, with parameter constraints that correspond to four nodal curves appearing and that also ensure the evolution of root variables of the non-translational symmetry corresponding to the elliptic discrete Painlev\'e equation in question is translational.
    Another example of the combination of nodal curves and projective reduction conditions in a single example is that of the $\dpain{XXXIV}$ equation \cite{cresswelljoshi} (equivalent to a special case of that known as $\dpain{II}$). 
    This is associated with a half-translation on a $D_5^{(1)}$ surface with a nodal curve and a projective reduction condition, see \cite{dP2symmetry}.
\end{remark}
\begin{remark} \label{rem:conjugacyinsymmetrygroups}
In the introduction, included among the data required to specify a discrete Painlev\'e equation up to equivalence via birational changes of variables  was the conjugacy class $[w]$ of the element generating the equation. 
In the generic case, this refers to conjugacy within the symmetry group determined by the type $\mathcal{R}^{\perp}$, which in the examples considered in this paper is $\widehat{W}(A_3^{(1)})$.
In the non-generic case, the notion of conjugacy must take into consideration parameter constraints.
For constraints corresponding to nodal curves as relevant to this paper, this is along the following lines.
Suppose that $\mathcal{X}\to \mathcal{A}$ and $\mathcal{X}'\to\mathcal{A}'$ are two families of Sakai surfaces of the same surface type $\mathcal{R}$.
Denote their common generic symmetry group by $\widehat{W}(\mathcal{R}^{\perp})$, and suppose that the families have constrained root variables corresponding to $\Delta^{\operatorname{nod}}\subset \Pic(\mathcal{X})$ and $\Delta^{\operatorname{nod}}{'}\subset \Pic(\mathcal{X}')$ respectively, but are otherwise generic.
Discrete Painlev\'e equations associated to elements $w,w'\in \widehat{W}(\mathcal{R}^{\perp})$ will be related by a birational change of variables if there exists $g\in \widehat{W}(\mathcal{R}^{\perp})$ such that $w' =g w g^{-1}$ and $\Delta^{\operatorname{nod}}{'}=g(\Delta^{\operatorname{nod}})$.
This will also enforce that $g$ be compatible with the parameter constraints corresponding to $\Delta^{\operatorname{nod}}{'}$ and $\Delta^{\operatorname{nod}}$.
% CHECK
% assuming 
% $$w'(\Delta^{\operatorname{nod}}{'})=\Delta^{\operatorname{nod}}{'}, \quad 
% w(\Delta^{\operatorname{nod}})=\Delta^{\operatorname{nod}},\quad w' =g w g^{-1}, \quad g(\Delta^{\operatorname{nod}})=\Delta^{\operatorname{nod}}{'},$$
% check compatibility:
% \begin{align*}
% w'(\Delta^{\operatorname{nod}}{'})&=\Delta^{\operatorname{nod}}{'}\\
% g w g^{-1} (\Delta^{\operatorname{nod}}{'})&=\Delta^{\operatorname{nod}}{'}\\
% w (\Delta^{\operatorname{nod}})&=g^{-1}(\Delta^{\operatorname{nod}}{'}) = \Delta^{\operatorname{nod}}
% \end{align*}

% In such a case, $g$ is in the intersection of the symmetry subgroups of the two families.

\end{remark}

\section{Results}

In this section we present the computational details establishing Theorems \ref{mainthm:pLaguerre}-\ref{mainthm:Meixner}.
We omit details of proofs in cases that these are standard computations, since demonstrations of how to perform these can be found in, e.g., \cite{SAKAI2001,KajNouYam:2017:GAPE,AHJN,DzhTak:2018:SASGTDPE,DzhFilSto:2020:RCDOPWHWDPE,DzhFilSto:2022:DERCSOPTRPEGA}.

\subsection{Perturbed Laguerre weight on the nonnegative real line} % (fold)
\label{sec:defLUE}

We consider the discrete system \eqref{eq-pLUE-disc} and establish \Cref{mainthm:pLaguerre}.
\begin{notation*}
For the discrete and differential systems from the perturbed Laguerre weight \eqref{weight-pL} and associated $D_5^{(1)}$ surfaces we use the following notation: coordinates $(f,g)$; parameters $\alpha,\gamma,n$; continuous independent variable  $s$;
centres of blowups $p_{i}$; exceptional divisors $L_{i}$.	
\end{notation*}

\subsubsection{The discrete and differential systems}

First recall from \cite{MC20} the discrete and differential systems satisfied by $f=f_n(s)$ and $g=g_n(s)$, which are 
\begin{equation}\label{eq-pLUE-disc-body}
	\left\{
		\bar{f} f = \frac{(g-  n)(g-\alpha-n)}{g(g+\gamma)} ,\quad
		\ubar{g} + g = \frac{-\gamma f^2  + ( 1-\alpha + s -2n + \gamma) f + \alpha + 2n +1}{(f-1)^2},
	\right.
    \end{equation}
and 
\begin{equation}\label{eq-pLUE-diff}
	\left\{
	\begin{aligned}
		\frac{df}{ds} &=\frac{(\gamma+2 g) f^2  - 4 f g+ (\alpha-\gamma-s +2n)f + 2 g - \alpha - 2n}{s} ,\\
		\frac{dg}{ds} &= \frac{g^2(1-f^2) - (\alpha+2n+\gamma f^2)g + (\alpha+n)n}{s f},
	\end{aligned}
	\right.
\end{equation}
respectively.
We take $f,g$ each as an affine coordinate on a copy of $\p^1$ and consider the discrete system \eqref{eq-pLUE-disc-body} as a family of birational maps
\begin{equation} \label{eq-pLUEbirationalmap}
\begin{aligned}
    \psi_{\boldsymbol{b},s} : \p^1 \times \p^1 &\dashrightarrow \p^1 \times \p^1,\\
    (f,g)&\mapsto (\bar{f},\bar{g}),
\end{aligned}
\end{equation}
with parameter space $\mathscr{B}\ni \boldsymbol{b}=(\alpha,\gamma,n)$, and evolution of parameters $\boldsymbol{b}\mapsto\bar{\boldsymbol{b}}$ given by $n\mapsto n+1$.
The independent variable space for the differential system \eqref{eq-pLUE-diff} is $\mathscr{S}=\C\setminus\{0\}\ni s$.
\subsubsection{Surfaces}

We first construct the surfaces on which the discrete system \eqref{eq-pLUE-disc-body} (and the differential system \eqref{eq-pLUE-diff}) is regularised, in line with Step 1 of the identification procedure outlined in \Cref{sec:identificationprocedure}.
\begin{lemma}
    Through the sequence of eight blowups of points $p_1,\dots,p_8$ given in \Cref{fig:pLUE:pointlocations}, we obtain a rational surface $Y_{\boldsymbol{b},s} = \Bl_{p_1 \cdots p_8}(\p^1 \times \p^1)$ such that the birational map \eqref{eq-pLUEbirationalmap} defined by the discrete system \eqref{eq-pLUE-disc-body} becomes an isomorphism
    $$ \tilde{\psi}_{\boldsymbol{b},s} : Y_{\boldsymbol{b},s} \longrightarrow Y_{\bar{\boldsymbol{b}},s}.$$
\end{lemma}
\begin{proof}
The construction of the space of initial conditions for the discrete system is done by standard computations, identifying points of indeterminacy and contracted curves in $\p^1 \times \p^1$ of $\psi_{\boldsymbol{b},s}$ and its inverse, and resolving them using blowups. 
\end{proof}

\begin{figure}[htb]
\begin{equation*}
\begin{aligned}
	&p_1 : (f,g)=(0,n),							&	&p_2 : (f,g)=(0,\alpha+n),  \\
	&p_3 : (f,g)=\left( 1, \infty \right) \quad \leftarrow 	&	&p_4 : (u_3,v_3) = \left(  f-1, \frac{1}{g(f-1)}  \right) = (0, 0)  \\
	&										& 	&\uparrow \\
	&										&	&p_5 : 	(u_4,v_4) = \left( u_3, \frac{v_3}{u_3} \right) = \left(0, \frac{1}{s}\right) \\
	&										& 	&\uparrow \\
	&										&	&p_6 : 	(u_5,v_5) = \left( u_4, \frac{v_4 - \frac{1}{s}}{u_4} \right) = \left(0, \frac{\alpha +\gamma - s + 2n -1}{s^2}\right), \\
	&p_7 : (f,g) = (\infty,0),  
    &	&p_8 : (f,g) = (\infty,-\gamma).
\end{aligned}
\end{equation*}
    \caption{Blowup point locations from the perturbed Laguerre weight}
    \label{fig:pLUE:pointlocations}
\end{figure}

Alternatively, applying the usual desingularisation procedure to the differential system \eqref{eq-pLUE-diff} yields the exact same family of surfaces.
We give an illustration of the configuration of exceptional divisors $L_1,\dots,L_8$ arising from the blowups in \Cref{fig:pLUE-soic}, with curves labeled by their classes.

\begin{figure}[htb]
	\begin{tikzpicture}[>=stealth,basept/.style={circle, draw=red!100, fill=red!100, thick, inner sep=0pt,minimum size=1.2mm}]
	\begin{scope}[xshift=0cm,yshift=0cm]
	\draw [black, line width = 1pt] (-0.4,0) -- (2.9,0)	node [pos=0,left] {\small $g=0$};
	\draw [black, line width = 1pt] (-0.4,2.5) -- (2.9,2.5) node [pos=0,left] {\small $g=\infty$};
	\draw [black, line width = 1pt] (0,-0.4) -- (0,2.9) node [pos=0,below] {\small $f=0$};
	\draw [black, line width = 1pt] (2.5,-0.4) -- (2.5,2.9) node [pos=0,below] {\small $f=\infty$};
	\node (b1) at (0,0.5) [basept,label={[xshift =-7pt, yshift=-2pt] \small $p_{1}$}] {};
	\node (b2) at (0,1.5) [basept,label={[xshift=-7pt, yshift=-2pt] \small $p_{2}$}] {};
	\node (b3) at (1.25,2.5) [basept,label={[xshift = -7pt, yshift=0pt] \small $p_{3}$}] {};
	\node (b4) at (2,3.1) [basept,label={[xshift = 0pt, yshift=0pt] \small $p_{4}$}] {};
	\node (b5) at (2.6,3.1) [basept,label={[xshift = 0pt, yshift=0pt] \small $p_{5}$}] {};
	\node (b6) at (3.2,3.1) [basept,label={[xshift = 0pt, yshift=0pt] \small $p_{6}$}] {};
	\node (b7) at (2.5,0) [basept,label={[xshift = 7pt, yshift=-2pt] \small $p_{7}$}] {};
	\node (b8) at (2.5,.95) [basept,label={[xshift = 7pt, yshift=-2pt] \small $p_{8}$}] {};
	\draw [red, line width = 0.8pt, ->] (b4) -- (1.55,2.5) -- (b3);
	\draw [red, line width = 0.8pt, ->] (b5) -- (b4);	
	\draw [red, line width = 0.8pt, ->] (b6) -- (b5);	
	\end{scope}
	\draw [->] (5.75,1)--(3.75,1) node[pos=0.5, below] {$\operatorname{Bl}_{p_{1}\cdots p_{8}}$};
	\begin{scope}[xshift=8cm,yshift=0cm]
	\draw [blue, line width = 1pt] (0,-.6) -- (0,3.4)	node [pos=0, below] {\small $\h_{f}-\L_{1}-\L_{2}$};
	\draw [blue, line width = 1pt] (4,-.6) -- (4,3.4)	node [pos=1, above] {\small $\h_{f}-\L_{7}-\L_{8}$};
	\draw [blue, line width = 1pt] (-0.4,3) -- (4.4,3)	node [pos=0, left] {\small $\h_{g}-\L_{3}-\L_{4}$};
	\draw [blue, line width = 1pt] (2,1.8) -- (2,4.2)	node [pos=1, above] {\small $\L_{4}-\L_{5}$};
	\draw [blue, line width = 1pt] (2.2,2) -- (1,2)	node [pos=0, right] {\small $\L_{3}-\L_{4}$};
	\draw [blue, line width = 1pt] (2.2,4) -- (.6,4)	node [pos=1, left] {\small $\L_{5}-\L_{6}$};
	\draw [red, line width = 1pt] (.8,3.8) -- (.8,4.8)	node [pos=1, above] {\small $\L_{6}$};
	\draw [red, line width = 1pt] (.6,1.5) -- (-0.6,0.5) node [pos=1, left] {\small $\L_{1}$};
	\draw [red, line width = 1pt] (.6,2.5) -- (-0.6,1.5)	node [pos=1, left] {\small $\L_{2}$};
	\draw [red, line width = 1pt] (4.3,1.7) -- (3,0.4)	node [pos=0, right] {\small $\L_{8}$};
	\draw [red, line width = 1pt] (4.3,0.7) -- (3,-.6)	node [pos=0, right] {\small $\L_{7}$};
	\draw [red, line width = 1pt] (3.55,-.4) -- (-.4,-.4)	node [pos=1, left] {\small $\h_g - \L_{7}$};
	\end{scope}
	\end{tikzpicture}
	\caption{The $D_{5}^{(1)}$ Sakai surface from the perturbed Laguerre weight}
	\label{fig:pLUE-soic}
\end{figure}
\begin{figure}[htb]
\begin{equation*}\label{eq:d-roots-LUE}			
	\raisebox{-32.1pt}{\begin{tikzpicture}[
			elt/.style={circle,draw=black!100,thick, inner sep=0pt,minimum size=2mm}]
		\path 	(-1,1) 	node 	(d0) [elt, label={[xshift=-10pt, yshift = -10 pt] $\delta_{0}$} ] {}
		        (-1,-1) node 	(d1) [elt, label={[xshift=-10pt, yshift = -10 pt] $\delta_{1}$} ] {}
		        ( 0,0) 	node  	(d2) [elt, label={[xshift=-10pt, yshift = -12 pt] $\delta_{2}$} ] {}
		        ( 1,0) 	node  	(d3) [elt, label={[xshift=10pt, yshift = -12 pt] $\delta_{3}$} ] {}
		        ( 2,1) 	node  	(d4) [elt, label={[xshift=10pt, yshift = -10 pt] $\delta_{4}$} ] {}
		        ( 2,-1) node 	(d5) [elt, label={[xshift=10pt, yshift = -10 pt] $\delta_{5}$} ] {};
		\draw [black,line width=1pt ] (d0) -- (d2) -- (d1)  (d2) -- (d3) (d4) -- (d3) -- (d5);
	\end{tikzpicture}} \qquad
			\begin{alignedat}{2}
			\delta_{0} &= \mathcal{H}_f -\L_1 - \L_2, &\qquad  \delta_{3} &=\L_4 - \L_5,\\
			\delta_{1} &= \mathcal{H}_f - \L_7 - \L_8, &\qquad  \delta_{4} &= \L_3 - \L_4,\\
			\delta_{2} &= \mathcal{H}_g - \L_3 - \L_4, &\qquad  \delta_{5} &= \L_5 - \L_6.
			\end{alignedat}
\end{equation*}
	\caption{The surface root basis from the perturbed Laguerre weight}
	\label{fig:d-roots-pLUE}	
\end{figure}

We denote the identification of all Picard groups $\Pic(Y_{\boldsymbol{b},s})$ into a single lattice by 
$$\Pic(\mathcal{Y}) = \Span_{\Z}\left\{\h_{f},\h_{g},\L_1,\dots,\L_8\right\},$$
with the intersection form being given by
$$        \h_f\bullet\h_g=1, \quad \h_f\bullet\h_f=\h_g\bullet\h_g=\h_f\bullet\L_i=\h_g\bullet\L_i=0, \quad \L_i\bullet\L_j=-\delta_{ij}, \quad\text{ for }~i,j=1,\dots,8.
$$
and the anticanonical class corresponding to 
$$-\K_{\mathcal{Y}}:=2\h_f+2\h_g-\L_1-\L_2-\L_3-\L_4-\L_5-\L_6-\L_7-\L_8\in\Pic(\Y),$$

\begin{proposition} \label{lem:prop:ACdiv-pLUE}
        The surface $Y_{\boldsymbol{b},s}$ is a Sakai surface of type $\mathcal{R}=D_5^{(1)}$.
\end{proposition}
\begin{proof}
    It is straightforward to show by direct calculation that the unique effective anticanonical divisor of $Y_{\boldsymbol{b},s}$ is that given by $D=- \operatorname{div} \omega$, where $\omega$ is the pullback under the blowup morphism of the rational 2-form on $\p^1\times\p^1$ defined in the affine chart $(f,g)$ by $\frac{ d f\wedge d g}{f}$.
    Its decomposition into irreducible components is 
    $D = D_0+D_1+2D_2+2D_3+D_4+D_5$,
    where $D_i$ are indicated in blue on \Cref{fig:pLUE-soic}.
    The corresponding elements of $\Pic(\mathcal{Y})$ are given in \Cref{fig:d-roots-pLUE},
as well as the Dynkin diagram encoding their intersection configuration, so the surface type $D_5^{(1)}$ is established.
\end{proof}

\subsubsection{Induced dynamics on the Picard lattice}
\begin{lemma} \label{lem:induceddynamicsonPicpLUE}
The isomorphisms $\tilde{\psi}_{\boldsymbol{b},s}$ induce via pushforward $(\tilde{\psi}_{\boldsymbol{b},s})_* : \Pic(Y_{\boldsymbol{b},s})\to \Pic(Y_{\bar{\boldsymbol{b}},s})$ the following action on $\Pic(\Y)$:
\begin{equation}
\Psi : \left\{
	\begin{aligned}
		\mathcal{H}_f 	&\mapsto 5 \h_f + 2 \h_g - \L_1 -  \L_2 -  2\L_3 -  2\L_4 - 2\L_5 - 2\L_6 - \L_7-\L_8, \\\
		\mathcal{H}_g 	&\mapsto  2\h_f + \h_g - \L_3 -\L_4 - \L_5 - \L_6, \\
		\L_1			&\mapsto  2\h_f + \h_g -\L_2 - \L_3 -\L_4 - \L_5 - \L_6, \\
		\L_2			&\mapsto  2\h_f + \h_g -\L_1 - \L_3 -\L_4 - \L_5 - \L_6, \\
		\L_3			&\mapsto  \h_f  - \L_6, \\
		\L_4			&\mapsto  \h_f  - \L_5, \\
		\L_5			&\mapsto  \h_f-\L_4, \\
		\L_6			&\mapsto  \h_f-\L_3, \\
		\L_7			&\mapsto  2\h_f + \h_g -\L_3 - \L_4 -\L_5 - \L_6 - \L_8, \\
		\L_8			&\mapsto  2\h_f + \h_g -\L_3 - \L_4 -\L_5 - \L_6 - \L_7. 
	\end{aligned}
\right.
\end{equation}
\end{lemma}
\begin{proof}
    The proof is a standard computation, e.g. by finding images under $\tilde{\psi}_{\boldsymbol{b},s}$ of curves on $Y_{\boldsymbol{b},s}$.
\end{proof}

\subsubsection{Identification with the standard model}
We present the identification of $\Pic(\X)$ and $\Pic(\Y)$ that matches the linear dynamics in \Cref{lem:induceddynamicsonPicpLUE} with the translation $T_{\operatorname{KNY}}$ associated to the discrete Painlev\'e equation \eqref{eq-KNY-dP-S3}.
\begin{lemma}\label{lem:KNY-to-pLUE-Pic} 
The identification $ \Pic(\mathcal{Y}) \to \Pic(\X)$ of Picard lattices between the standard Kajiwara-Noumi-Yamada
	and perturbed Laguerre surfaces is given, written using equalities by  abuse of notation,     as follows:
	\begin{equation}\label{eq:basis-KNY-LUE}
		\begin{aligned}
			\mathcal{H}_{q} & = 2 \mathcal{H}_f + \mathcal{H}_g - \L_1 - \L_3 - \L_4 - \L_7, &\qquad 
				\mathcal{H}_{f} &= \h_q + \h_p  - \E_5 - \E_6,\\
			\mathcal{H}_{p} &=  \mathcal{H}_{f} + \mathcal{H}_{g} - \L_1 - \L_7,  &\qquad 	
				\mathcal{H}_{g} & = \mathcal{H}_{q} +2  \mathcal{H}_{p} - \E_1 - \E_3 - \E_5 - \E_6, \\
			\mathcal{E}_{1} &= \mathcal{H}_f - \L_1, &\qquad 
				\mathcal{L}_{1}	&= \mathcal{H}_{q}+\mathcal{H}_p - \E_1 - \E_5 - \E_6,\\ 
			\mathcal{E}_{2} &= \L_2, &\qquad 
				\mathcal{L}_{2}	&= \E_2,\\ 
			\mathcal{E}_{3} &= \mathcal{H}_{f} - \L_{7}, &\qquad 
				\mathcal{L}_{3}	&= \h_p - \E_{6},\\ 
			\mathcal{E}_{4} &= \L_8, &\qquad 
				\mathcal{L}_{4}	&= \h_p - \E_{5},\\ 
			\mathcal{E}_{5} &= \mathcal{H}_f + \mathcal{H}_g -\L_1 - \L_4 - \L_7, &\qquad 
				\mathcal{L}_{5}	&= \E_7,\\ 
			\mathcal{E}_{6} &= \mathcal{H}_f+ \mathcal{H}_g - \L_1 - \L_3 - \L_7, &\qquad 
				\mathcal{L}_{6}	&= \E_8,\\ 
			\mathcal{E}_{7} &= \L_5, &\qquad 
				\mathcal{L}_{7}	&= \h_{q}  + \h_p - \E_3 -\E_5 - \E_6,\\ 
			\mathcal{E}_{8} &=  \L_6, &\qquad 
				\mathcal{L}_{8}	&= \E_4.
		\end{aligned}
	\end{equation}
	This results in the following correspondences between the surface roots:
	\begin{equation}\label{eq:geom-KNY-LUE}
	\begin{aligned}
		\delta_{0} &= \mathcal{E}_{1} - \mathcal{E}_{2} = \h_f - \L_1 - \L_2, &\qquad
			\delta_{3} &=\mathcal{H}_{p} - \mathcal{E}_{5} - \mathcal{E}_{7} = 
				\L_4 - \L_5, \\
		\delta_{1} &= \mathcal{E}_{3} - \mathcal{E}_{4} = \h_f - \L_7 - \L_8, &\qquad
			\delta_{4} &= \mathcal{E}_{5} - \mathcal{E}_{6}  = \L_3 -\L_4, \\
		\delta_{2} &=  \mathcal{H}_{q} - \mathcal{E}_{1} - \mathcal{E}_{3} = \h_g - \L_3 - \L_4 &\qquad
			\delta_{5} &= \mathcal{E}_{7} - \mathcal{E}_{8} = \L_5 - \L_6.
	\end{aligned}
	\end{equation}
	The symmetry root basis in $\Pic(\X)$ given in \Cref{fig:a-roots-a3-KNY} maps to $\Pic(\Y)$ according to
	\begin{equation}\label{eq:sym-KNY-LUE}
	\begin{aligned}
			\alpha_{0} &= \mathcal{H}_{p} - \mathcal{E}_{1} - \mathcal{E}_{2}
		=  \h_g - \L_2 - \L_7, \\
			\alpha_{1} &= \mathcal{H}_{q} - \mathcal{E}_{5} - \mathcal{E}_{6} = - \h_g + \L_1 +\L_7, \\
			\alpha_{2} &= \mathcal{H}_{p} - \mathcal{E}_{3} - \mathcal{E}_{4} = 
			\h_g - \L_1 - \L_8,\\
			\alpha_{3} &= \mathcal{H}_{q} - \mathcal{E}_{7} - \mathcal{E}_{8} = 2 \h_f + \h_g - \L_1 -\L_3 -\L_4 -\L_5 -\L_6 - \L_7.
		\end{aligned}
	\end{equation}
\end{lemma}
Under the identification in \Cref{lem:KNY-to-pLUE-Pic}, the discrete system \eqref{eq-pLUE-disc-body}  corresponds to the following action on the symmetry roots
\begin{equation} \label{eq-PsiOnRoots-pLUE}
\Psi: 
\left\{
	\begin{aligned}
		\alpha_0 &\mapsto \alpha_0 - \delta, \\
		\alpha_1 &\mapsto \alpha_1 + \delta , \\
		\alpha_2 &\mapsto \alpha_2 - \delta, \\
		\alpha_3 &\mapsto \alpha_3 + \delta,	
	\end{aligned}
\right.
\end{equation}
which coincides with that of $T_{\operatorname{KNY}}$, given in \cref{eq:translationonroots-KNY}, which generates the KNY discrete Painlev\'e equation \eqref{eq-KNY-dP-S3}.
To obtain a realisation of the identification in \Cref{lem:KNY-to-pLUE-Pic} as a change of variables and parameter matching between the discrete systems \eqref{eq-pLUE-disc-body} and \eqref{eq-KNY-dP-S3}, we first compute root variables of the surface $Y_{\boldsymbol{b},s}$ with respect to the symmetry root basis corresponding to the standard one under the identification in \Cref{lem:KNY-to-pLUE-Pic}.

\begin{lemma} \label{lem:rootvarspLUE}
    The root variables $a_i = \chi(\alpha_i)$, $i=0,1,2,3$, defined with respect to the period map $\chi$ given by the 2-form on $Y_{\boldsymbol{b},s}$ defined by $ \frac{df \wedge dg}{f}$ for the symmetry roots 
   \begin{equation*} 
   \begin{aligned}
    \alpha_{0} &=  \h_g - \L_2 - \L_7, \quad 
	&&\alpha_{1} = - \h_g + \L_1 +\L_7, \quad\\
	\alpha_{2} &=  \h_g - \L_1 - \L_8, \quad
    &&\alpha_{3} =  2 \h_f + \h_g - \L_1 -\L_3 -\L_4 -\L_5 -\L_6 - \L_7,
    \end{aligned}
    \end{equation*}
     are given by 
	\begin{equation}
		a_{0} = \alpha + n,\quad  a_{1} =-  n ,\quad
		a_{2} = \gamma+n , \quad a_{3} = 1 - \alpha - \gamma - n,
	\end{equation}
    so in particular this matches the normalisation $a_0+a_1+a_2+a_3=1$.
\end{lemma}
\begin{proof}
    This is a standard computation based on residue theorem computations along components of the anticanonical divisor, demonstrations of which can be found in \cite{SAKAI2001} as well as in, e.g., \cite{DzhTak:2018:SASGTDPE, DzhFilSto:2020:RCDOPWHWDPE,  DzhFilSto:2022:DERCSOPTRPEGA, ASIDEnotes}.
\end{proof}
	
We are now ready to establish the remaining part of \Cref{mainthm:pLaguerre}.
\begin{proposition}\label{prop:coords-KNY-pLUE} The following change of variables and parameter matching simultaneously identifies the discrete and differential systems \eqref{eq-LUE-disc-body} and \eqref{eq-LUE-diff} from the perturbed Laguerre weight with the KNY discrete Painlev\'e equation \eqref{eq-KNY-dP-S3} and the KNY Hamiltonian form \eqref{eq-KNY-Ham5-sys} of $\pain{V}$ respectively: 
    \begin{equation}\label{eq:KNYtoLUE}
   	 \left\{
		\begin{aligned}
   	 	q&= \frac{(f-1)(f g - g + n)}{s f},\\
   		p&= \frac{s (g-n)}{f g- g + n}  \\
	   	 \end{aligned}
		 \qquad
		 \begin{aligned}
		 a_{0}&= \alpha+ n,		&\quad &a_{1}= -n,\\  
		a_{2}&= \gamma+ n ,   		&\quad &a_{3}= 1 - \alpha - \gamma- n, 
		\end{aligned}
	\qquad t=s.
	\right.
    \end{equation}
\end{proposition}
\begin{proof}
    Deriving the change of variables is a standard computation, the details of which we refer to \cite{DzhFilSto:2020:RCDOPWHWDPE}.
    The matching of the independent variables $t$ and $s$ from the differential systems appears in the process of finding the isomorphism between $X_{\boldsymbol{a},t}$ and $Y_{\boldsymbol{b},s}$ that provides the change of variables, see \cite{DzhFilSto:2022:DERCSOPTRPEGA}.
    The fact that the systems are matched can be verified immediately by direct calculation.
\end{proof}

Note that, if $n$ is assumed complex, the surfaces $Y_{\boldsymbol{b},s}$ constructed above are generic, in the sense that the family exhausts the moduli for its surface type $\mathcal{R}=D_5^{(1)}$. 
This can be seen in the fact that the corresponding root variables $a_0,a_1,a_2,a_3$ and extra parameter $s=t$ are free, which implies that the family of surfaces includes representatives of all isomorphism classes of $D_5^{(1)}$ surfaces, through the Torelli-type theorem \cite[Th. 25]{SAKAI2001}.
In this sense the family of surfaces constructed here has the full symmetry group $\widehat{W}(A_3^{(1)})$. The next example we will study does not.

\subsection{Laguerre weight on a finite interval} % (fold)
\label{sec:LUE}

We consider the discrete system \eqref{eq-LUE-disc} and establish \Cref{mainthm:Laguerre}.

\begin{notation*}
For the discrete and differential systems from the Laguerre weight \eqref{weight-L} and associated $D_5^{(1)}$ surfaces we use the same notation as for the perturbed case: coordinates $(f,g)$; parameters $\alpha,n$; continuous independent variable  $s$;
centres of blowups $p_{i}$; exceptional divisors $L_{i}$.	
We will recycle notation for the birational maps $\psi_{\boldsymbol{b},s}$,  surfaces $Y_{\boldsymbol{b},s}$, and isomorphisms $\tilde{\psi}_{\boldsymbol{b},s}$.
\end{notation*}
\subsubsection{The discrete and the differential systems}

We first recall from \cite{LC17} the discrete and differential systems satisfied by $f = f_n(s)$, $g=g_n(s)$, which are 
\begin{equation}\label{eq-LUE-disc-body}
	\left\{
		\bar{f} f = \frac{g^2 - (\alpha+ 2n)g + (\alpha+n)n}{g^2} ,\qquad
		\ubar{g} + g = \frac{(1-\alpha + s -2n) f + \alpha + 2n +1}{(f-1)^2},
	\right.
\end{equation}
and 
\begin{equation}\label{eq-LUE-diff}
	\left\{
	\begin{aligned}
		\frac{df}{ds} &=\frac{2 f^2 g - 4 f g+ (\alpha-s +2n)f + 2 g - \alpha - 2n}{s} ,\\
		\frac{dg}{ds} &= \frac{g^2(1-f^2) - (\alpha+2n)g + (\alpha+n)n}{s f},
	\end{aligned}
	\right.
\end{equation}
respectively. Note that these can be obtained by setting $\gamma=0$ in the systems \eqref{eq-pLUE-disc-body} and \eqref{eq-pLUE-diff}.

As above, we take $f,g$ as an affine coordinates on $\p^1 \times \p^1$ and consider the discrete system \eqref{eq-LUE-disc-body} as a family of birational maps $\psi_{\boldsymbol{b},s}$
with parameter space $\mathscr{B}\ni \boldsymbol{b}=(\alpha,n)$, and evolution of parameters $\boldsymbol{b}\mapsto\bar{\boldsymbol{b}}$ given by $n\mapsto n+1$.
The independent variable space for the differential system \eqref{eq-LUE-diff} is again $\mathscr{S}=\C\setminus\{0\}\ni s$.
\subsubsection{Surfaces}

The surfaces on which the discrete system \eqref{eq-LUE-disc-body} is regularised is essentially the result of setting $\gamma=0$ in those from the perturbed Laguerre case above.
However, this degeneration causes a pair of centres of blowups to merge in a way such that a nodal curve appears, so we provide the details.

\begin{lemma}
    Through the sequence of eight blowups of points $p_1,\dots,p_8$ given in \Cref{fig:LUE:pointlocations}, we obtain a rational surface $Y_{\boldsymbol{b},s} = \Bl_{p_1 \cdots p_8}(\p^1 \times \p^1)$ such that the birational map defined by the discrete system \eqref{eq-LUE-disc-body} becomes an isomorphism
    $$ \tilde{\psi}_{\boldsymbol{b},s} : Y_{\boldsymbol{b},s} \longrightarrow Y_{\bar{\boldsymbol{b}},s}.$$
\end{lemma}

\begin{figure}[htb]
\begin{equation*}
\begin{aligned}
	&p_1 : (f,g)=(0,n),							&	&p_2 : (f,g)=(0,\alpha+n),  \\
	&p_3 : (f,g)=\left( 1, \infty \right) \quad \leftarrow 	&	&p_4 : (u_3,v_3) = \left(  f-1, \frac{1}{g(f-1)}  \right) = (0, 0)  \\
	&										& 	&\uparrow \\
	&										&	&p_5 : 	(u_4,v_4) = \left( u_3, \frac{v_3}{u_3} \right) = \left(0, \frac{1}{s}\right) \\
	&										& 	&\uparrow \\
	&										&	&p_6 : 	(u_5,v_5) = \left( u_4, \frac{v_4 - \frac{1}{s}}{u_4} \right) = \left(0, \frac{\alpha  - s + 2n -1}{s^2}\right), \\
	&p_7 : (f,g) = (\infty,0)  \quad \leftarrow
    &	&p_8 : (U_7,V_7) = \left( \frac{1}{f g},g \right) = (0, 0).
\end{aligned}
\end{equation*}
    \caption{Blowup point locations from the Laguerre weight}
    \label{fig:LUE:pointlocations}
\end{figure}
We give an illustration of the configuration of exceptional divisors $L_1,\dots,L_8$ arising from the blowups in \Cref{fig:LUE-soic}, with curves labelled by their classes in $\Pic(\mathcal{Y}_{\boldsymbol{b},s})$.
Note that the point $p_8$ is now infinitely near to $p_7$, i.e. it is found on the exceptional divisor $L_7$.
Because $p_8$ lies at the intersection of $L_7$ and the strict transform of the curve on $\p^1 \times \p^1$ defined by $f = \infty$, the divisor $L_7-L_8$ on $Y_{\boldsymbol{b},s}$ is effective.
This constitutes a nodal curve, so we indicate it in green in \Cref{fig:LUE-soic}.

\begin{figure}[htb]
	\begin{tikzpicture}[>=stealth,basept/.style={circle, draw=red!100, fill=red!100, thick, inner sep=0pt,minimum size=1.2mm}]
	\begin{scope}[xshift=0cm,yshift=0cm]
	\draw [black, line width = 1pt] (-0.4,0) -- (2.9,0)	node [pos=0,left] {\small $g=0$};
	\draw [black, line width = 1pt] (-0.4,2.5) -- (2.9,2.5) node [pos=0,left] {\small $g=\infty$};
	\draw [black, line width = 1pt] (0,-0.4) -- (0,2.9) node [pos=0,below] {\small $f=0$};
	\draw [black, line width = 1pt] (2.5,-0.4) -- (2.5,2.9) node [pos=0,below] {\small $f=\infty$};
	\node (b1) at (0,0.5) [basept,label={[xshift =-7pt, yshift=-2pt] \small $p_{1}$}] {};
	\node (b2) at (0,1.5) [basept,label={[xshift=-7pt, yshift=-2pt] \small $p_{2}$}] {};
	\node (b3) at (1.25,2.5) [basept,label={[xshift = -7pt, yshift=0pt] \small $p_{3}$}] {};
	\node (b4) at (2,3.1) [basept,label={[xshift = 0pt, yshift=0pt] \small $p_{4}$}] {};
	\node (b5) at (2.6,3.1) [basept,label={[xshift = 0pt, yshift=0pt] \small $p_{5}$}] {};
	\node (b6) at (3.2,3.1) [basept,label={[xshift = 0pt, yshift=0pt] \small $p_{6}$}] {};
	\node (b7) at (2.5,0) [basept,label={[xshift = 7pt, yshift=-15pt] \small $p_{7}$}] {};
	\node (b8) at (2.95,.95) [basept,label={[xshift = 0pt, yshift=0pt] \small $p_{8}$}] {};
	\draw [red, line width = 0.8pt, ->] (b4) -- (1.55,2.5) -- (b3);
	\draw [red, line width = 0.8pt, ->] (b5) -- (b4);	
	\draw [red, line width = 0.8pt, ->] (b6) -- (b5);	
	\draw [red, line width = 0.8pt, ->] (b8) -- (2.5,0.35) -- (b7);	
	\end{scope}
	\draw [->] (5.75,1)--(3.75,1) node[pos=0.5, below] {$\operatorname{Bl}_{p_{1}\cdots p_{8}}$};
	\begin{scope}[xshift=8cm,yshift=0cm]
	\draw [blue, line width = 1pt] (0,-.6) -- (0,3.4)	node [pos=0, below] {\small $\h_{f}-\L_{1}-\L_{2}$};
	\draw [blue, line width = 1pt] (4,-.6) -- (4,3.4)	node [pos=1, above] {\small $\h_{f}-\L_{7}-\L_8$};
	\draw [blue, line width = 1pt] (-0.4,3) -- (4.4,3)	node [pos=0, left] {\small $\h_{g}-\L_{3}-\L_{4}$};
	\draw [blue, line width = 1pt] (2,1.8) -- (2,4.2)	node [pos=1, above] {\small $\L_{4}-\L_{5}$};
	\draw [blue, line width = 1pt] (2.2,2) -- (1,2)	node [pos=0, right] {\small $\L_{3}-\L_{4}$};
	\draw [blue, line width = 1pt] (2.2,4) -- (.6,4)	node [pos=1, left] {\small $\L_{5}-\L_{6}$};
	\draw [red, line width = 1pt] (.8,3.8) -- (.8,4.8)	node [pos=1, above] {\small $\L_{6}$};
	\draw [red, line width = 1pt] (.6,1.5) -- (-0.6,0.5) node [pos=1, left] {\small $\L_{1}$};
	\draw [red, line width = 1pt] (.6,2.5) -- (-0.6,1.5)	node [pos=1, left] {\small $\L_{2}$};
	\draw [red, line width = 1pt] (4.3,-.5) -- (3.1,0.7)	node [pos=0, right] {\small $\L_{8}$};
	\draw [darkgreen, line width = 1pt] (3.5,0.7) -- (2.2,-.6)	node [pos=1, below] {\small $\L_{7}-\L_{8}$};
	\draw [red, line width = 1pt] (2.7,-.4) -- (-.4,-.4)	node [pos=1, left] {\small $\h_g - \L_{7}$};
	\end{scope}
	\end{tikzpicture}
	\caption{The $D_{5}^{(1)}$ Sakai surface (with nodal curve) from the Laguerre weight}
	\label{fig:LUE-soic}
\end{figure}

\begin{proposition} \label{lem:prop:ACdiv-LUE}
        $Y_{\boldsymbol{b},s}$ is a Sakai surface of type $\mathcal{R}=D_5^{(1)}$ with $\L_7-\L_8\in \Delta^{\operatorname{nod}}$ for every $\boldsymbol{b},s$.
\end{proposition}
\begin{proof}
    Similarly to \Cref{lem:prop:ACdiv-pLUE}, the unique effective anticanonical divisor on $Y_{\boldsymbol{b},s}$ can be shown by direct calculation to be that coming from the 2-form $\frac{df \wedge dg}{f}$.
    The irreducible components of this gives the same expressions for the surface root basis as in \Cref{fig:d-roots-pLUE}, so we do not reproduce it here, and we see that the surface type is as claimed.
    With the locations of centres of blowups as shown in \Cref{fig:LUE:pointlocations}, we get a $(-2)$-curve which is not among the components of the anticanonical divisor, coming from the strict transform of $L_7$ under the blowup of $p_8$.
    This is clearly a rational curve and we have $\L_7-\L_8\in\Delta^{\operatorname{nod}}$.
\end{proof}

\subsubsection{Identification with the standard model}

Again identifying all $\Pic(Y_{\boldsymbol{b},s})$ into the single lattice $\Pic(\Y)$, the linear action induced via pushforwards by the isomorphisms $\tilde{\psi}_{\boldsymbol{b},s}$ is the same as in \Cref{lem:induceddynamicsonPicpLUE}.
To find an identification with the standard model of $D_5^{(1)}$-surfaces, we will need consideration of nodal curves.

\begin{lemma}\label{lem:KNY-to-LUE} 
The identification $ \Pic(\mathcal{Y}) \to \Pic(\X)$ of Picard lattices 
    can be taken to be the same as in \Cref{lem:KNY-to-pLUE-Pic}.
	This results in the same matching of surface and symmetry root bases, but the class of the nodal curve on $Y_{\boldsymbol{b},s}$ corresponds under this identification to
	\begin{equation}
		\L_7 - \L_8 = \alpha_1 + \alpha_2 = \mathcal{H}_q + \mathcal{H}_p - \E_3 - \E_4 - \E_5 -\E_6.
	\end{equation}
\end{lemma}
Under the identification, the discrete system from the Laguerre weight corresponds to the same action on symmetry roots in \Cref{eq-PsiOnRoots-pLUE}, and coincides with that of the KNY discrete Painlev\'e equation \eqref{eq-KNY-dP-S3}.
To obtain the transformation of the discrete system \eqref{eq-pLUE-disc-body} to a special case of this equation, we introduce a constraint on root variables corresponding to the nodal curve.

\begin{proposition} \label{prop:rootvarsLUE}
    The root variables $a_i=\chi(\alpha_i)$, $i=0,1,2,3$, defined with respect to the period map $\chi$ given by the 2-form on $Y_{\boldsymbol{b},s}$ defined by $ \frac{df \wedge dg}{f}$ are given by 
    \begin{equation}
		a_{0} = \alpha + n,\quad  a_{1} =-  n ,\quad
		a_{2} =  n , \quad a_{3} = 1 - \alpha - n.
	\end{equation}
\end{proposition}
\begin{proof}
    Most of the proof is a standard computation, but we give some details related to the presence of the nodal curve.
    The computations in \Cref{lem:rootvarspLUE} of the values of the period map $\chi$ on $\alpha_0,\alpha_1$ and $\alpha_3$ can be done in a way such that they are not affected by the merging of the point $p_8$ onto the intersection of $L_7$ with the strict transform of the curve on $\p^1 \times \p^1$ defined by $f = \infty$. 
    However, certain choices of how to express a simple root as a difference of classes of $(-1)$-curves become subtle now that $\L_7$ does not correspond to an irreducible/prime divisor any more, but rather can be written as $\L_7=(\L_7-\L_8)+\L_8$.
    For example, in computing the root variable $a_0=\chi(\alpha_0)$, for which we recall that $\alpha_0=\h_g-\L_2-\L_7$, in the generic case this could be done by expressing $\alpha_0$ as the difference of classes of two irreducible curves of self-intersection $-1$, using first $\h_g-\L_2$, which is the strict transform of the line of constant $g$ through $p_2$, and second $\L_7$.
    When the nodal curve $\L_7-\L_8$ is present, this choice of expression is no longer possible since $\L_7$ is not an irreducible curve.
    However, one can still choose a different expression as a difference of exceptional curves, e.g. $\alpha_0=(\h_g-\L_7)-(\L_2)$ using the strict transform of the line $g=0$ and the exceptional divisor $\L_2$, and the computation proceeds as usual. 
\end{proof}

In order to obtain the change of variables to the special case of the KNY discrete Painlev\'e equation \eqref{eq-KNY-dP-S3} with parameter constraint $a_1+a_2=0$, we need to construct the subfamily of surfaces $X_{\boldsymbol{a},t}$ with this constraint on root variables.
Fortunately, this is straightforward given that we have identified in \Cref{lem:KNY-to-LUE} that the nodal curve should correspond to 
$$\alpha_1+\alpha_2 = \h_q+\h_p - \E_3-\E_4-\E_5-\E_6.$$
The subfamily of surfaces should be obtained through the same sequence of blowups as in \Cref{fig:LUE:pointlocations}, but with $b_3,b_4,b_5,b_6$ lying on a curve in $\p^1 \times \p^1$ of bidegree (1,1), or more precisely that there should exist a curve of bidegree (1,1) whose proper transform on $X_{\boldsymbol{a},t}$ has class $\alpha_1+\alpha_2$.
The condition for this to happen is precisely that $\chi(\alpha_1+\alpha_2)=a_1+a_2=0$, which is as expected according to \Cref{prop:nodalcurvesrootvars}.

\begin{proposition}\label{prop:coords-KNY-LUE} The following change of variables and parameter matching simultaneously identifies the discrete and differential systems \eqref{eq-LUE-disc-body} and \eqref{eq-LUE-diff} from the Laguerre weight with the KNY discrete Painlev\'e equation \eqref{eq-KNY-dP-S3} and KNY Hamiltonian form \eqref{eq-KNY-Ham5-sys} of $\pain{V}$ respectively, in the special case $a_1 + a_2 = 0$: 
    \begin{equation}\label{eq:KNYtoLUE}
   	 \left\{
		\begin{aligned}
   	 	q&= \frac{(f-1)(f g - g + n)}{s f},\\
   		p&= \frac{s (g-n)}{f g- g + n},  \\
	   	 \end{aligned}
		 \qquad
		 \begin{aligned}
		 a_{0}&= \alpha+ n,		&\quad &a_{1}= -n,\\  
		a_{2}&= n ,   		&\quad &a_{3}= 1 - \alpha - n, 
		\end{aligned}
	\qquad t=s.
	\right.
    \end{equation}
\end{proposition}
\begin{proof}
    This is a standard computation along the same lines as \Cref{prop:coords-KNY-pLUE}, with the only extra consideration necessary being to work with the family of surfaces $X_{\boldsymbol{a},t}$ with constraint $a_1+a_2=0$ on root variable parameters.
\end{proof}

\subsubsection{Symmetry type}

As anticipated in \Cref{subsec:paramconstraintsandobstructionstosymmetries}, the nodal curve and associated parameter constraint mean that the symmetry type in this case, in the sense of the symmetry group of the family of surfaces, will be a subgroup of the generic one $\widehat{W}(A_3^{(1)})$.

\begin{proposition} \label{prop:symmetry-type-LUE}
    The subgroup of $\widehat{W}(A_3^{(1)})=W(A_3^{(1)})\rtimes \Aut(A_3^{(1)})$ whose Cremona action is compatible with the parameter constraint $a_1+a_2=0$ on root variables $\boldsymbol{a}=(a_0,a_1,a_2,a_3)\in \mathscr{A}$ is 
    $$
    \left< s_0, s_1, w_0, w_1, \tau \right> \cong \left(W(A_1^{(1)})\times W(A_1^{(1)})\right)\rtimes \Z/2\Z \cong  W( (A_1 + \underset{|\alpha|^2=4}{A_1})^{(1)} ) \rtimes \Z / 2\Z,$$
    where the generators are given in terms of the simple reflections $w_i$ and $A_3^{(1)}$ Dynkin diagram automorphisms by
    \begin{equation} \label{eq:stabilisergeneratorsLUE}
        \begin{aligned}
            &s_0 = w_0w_1w_0, &&s_1= w_3 w_2 w_3, && \langle s_0,s_1\rangle\cong W(A_1^{(1)}), &&\tau s_0 = s_1 \tau,
            \\
            &s'_0 = \sigma_2 w_2 w_0, &&s'_1= \sigma_1\sigma_2\sigma_1 w_3 w_1, && \langle s'_0,s'_1\rangle\cong W({A_1^{(1)}}), &&\tau s_0' = s_1' \tau,
        \end{aligned}
         \quad 
        \quad \tau = \sigma_1.
    \end{equation}
\end{proposition}
    We ascribe the type $(A_1 + \underset{|\alpha|^2=4}{A_1})^{(1)}$ to the Weyl group in \Cref{prop:symmetry-type-LUE} because it can be described in terms of a lattice in $\Pic(\X)$ which is, in a similar sense to some in the Sakai classification, a root lattice of type $(A_1 + \underset{|\alpha|^2=4}{A_1})^{(1)}$.
    To illustrate this, introduce the sublattice of $Q(A_3^{(1)})$ given by  $Q_1 + Q_2$, where
\begin{equation}
\begin{aligned}
Q_1 &= \Z \beta_0 + \Z \beta_1, 
&&
\begin{aligned}
\beta_0 &= \alpha_0 + \alpha_1 = \h_q + \h_p - \E_1 -\E_2 -\E_5 -\E_6, \\
\beta_1 &= \alpha_2 + \alpha_3 = \h_q + \h_p - \E_3 -\E_4 -\E_7 -\E_8, 
\end{aligned}
&&&\beta_0 + \beta_1 = \delta,
\\
Q_2 &= \Z \beta_0' + \Z \beta_1', 
&&
\begin{aligned}
\beta_0' &= \alpha_0 + \alpha_2 =  2 \h_p - \E_1 -\E_2 - \E_3 - \E_4, \\
\beta_1' &= \alpha_1 + \alpha_3 = 2 \h_q - \E_5 -\E_6 - \E_7 - \E_8, 
\end{aligned}
&&&\beta_0' + \beta'_1 = \delta,
\end{aligned}
\end{equation}
so $Q_1 + Q_2 = \Z \beta_1 \oplus \Z \gamma_1 \oplus \Z \delta,$ where $\oplus$ indicates orthogonal direct sum.
The action of the generators $s_0,s_1,s_0',s_1'$ on this lattice corresponds to the usual realisation of the affine Weyl group as reflections:
\begin{equation} \label{eq:sonbetas}
% \begin{gathered}
s_0  : 
\left\{ 
\begin{aligned}
\beta_0 &\mapsto - \beta_0, \\
\beta_1 &\mapsto \beta_1 + 2 \beta_0,  \\
\beta'_0 &\mapsto \beta'_0, \\
\beta'_1 &\mapsto \beta'_1,
\end{aligned}
\right. 
\quad
s_1  : 
\left\{ 
\begin{aligned}
\beta_0 &\mapsto \beta_0 + 2 \beta_1, \\
\beta_1 &\mapsto - \beta_1,   \\
\beta'_0 &\mapsto \beta'_0, \\
\beta'_1 &\mapsto \beta'_1,
\end{aligned}
\right. 
\quad
s_0'   : 
\left\{ 
\begin{aligned}
\beta_0 &\mapsto \beta_0, \\
\beta_1 &\mapsto \beta_1,   \\
\beta'_0 &\mapsto - \beta'_0,\\
\beta'_1 &\mapsto \beta'_1+2\beta'_0,
\end{aligned}
\right. 
\quad
s_1' : 
\left\{ 
\begin{aligned}
\beta_0 &\mapsto \beta_0, \\
\beta_1 &\mapsto \beta_1,  \\
\beta'_0 &\mapsto \beta'_0 + 2 \beta'_1,\\
\beta'_1 &\mapsto -\beta'_1.
\end{aligned}
\right. 
% \end{gathered}
\end{equation}
Further, the element $\tau$ acts by simultaneously permuting the two pairs $\{\beta_0,\beta_1\}$ and $\{\beta_0',\beta_1'\}$, in cycle notation $\tau=(\beta_0\,\beta_1)(\beta_0' \, \beta_1')$.
Note that the intersection matrix of the elements $\left\{ \beta_0, \beta_1, \gamma_0, \gamma_1 \right\}$ is 
\begin{equation}
\left(\begin{array}{cccc}-2 & 2 & 0 & 0 \\2 & -2 & 0 & 0 \\0 & 0 & -4 & 4 \\0 & 0 & 4 & -4\end{array}\right),
\end{equation}
which is (the multiple by $-1$ of) the direct sum of generalised Cartan matrices associated with root systems $A_1^{(1)}$ and $\underset{|\alpha|^2=4}{A_1^{(1)}}$, as in \cite[Sec. 8.2.23]{KajNouYam:2017:GAPE} and \cite[Sec. 8.2.20]{KajNouYam:2017:GAPE} respectively.
This leads us to ascribe the type $(A_1+\underset{|\alpha|^2=4}{A_1})^{(1)}$ to this `root lattice', and to include this in the description of the symmetry type of the discrete Painlev\'e equation in this section.
The extension by $\tau$ corresponds not to the full extension by automorphisms of the pair of Dynkin diagrams associated to this root system, but rather a single automorphism that acts simultaneously on both; see \Cref{fig:dynkindiagramA1A1}.
\begin{figure}[htb]
    \centering
    $$
    \begin{tikzpicture}[
					elt/.style={circle,draw=black!100,thick, inner sep=0pt,minimum size=1.3ex},scale=0.8]
				\begin{scope}[xshift=0cm,yshift=0cm]
				\path 	(0,0) 	node 	(b0) [elt, label={[xshift=0pt, yshift = -25 pt] $\beta_{0}$} ] {}
				        (1.4,0) node 	(b1) [elt, label={[xshift=0pt, yshift = -25 pt] $\beta_{1}$} ] {};
				\draw [black,thick, double distance = .4ex] (b0) -- (b1);
                \draw [purple, dashed, <->, line width=0.5pt] (b0.north) to[out=60, in=120] node[midway, above] {$\tau$} (b1.north) ;
				% \draw [purple, dashed, line width = 0.5pt] (0.5,-0.3) -- (0.5,0.3)	node [pos=1,above] {\small $\sigma_{1}\sigma_{2}\sigma_{1}\sigma_{2}$};
				\end{scope}
				\begin{scope}[xshift=3cm,yshift=0cm]
				\path 	(0,0) 	node 	(b0) [elt, label={[xshift=0pt, yshift = -25 pt] $\beta'_{0}$} ] {}
				        (1.8,0) node 	(b1) [elt, label={[xshift=0pt, yshift = -25 pt] $\beta'_{1}$} ] {};
				\draw [black,thick, double distance = .4ex] (b0) -- (b1);
                \draw [purple, dashed, <->, line width=0.5pt] (b0.north) to[out=60, in=120] node[midway, above] {$\tau$} (b1.north) ;
				% \draw [purple, dashed, line width = 0.5pt] (1,-0.3) -- (1,0.3)	node [pos=1,above] {\small $\sigma_{1}$};
				\end{scope}
			\end{tikzpicture}
            $$
    \caption{Dynkin diagram of root lattice of type $(A_1+\underset{|\alpha|^2=4}{A_1})^{(1)}$ and automorphism $\tau$.}
    \label{fig:dynkindiagramA1A1}
\end{figure}
\begin{proof}[Proof of \Cref{prop:symmetry-type-LUE}]
    In light of the results reviewed in \Cref{subsec:paramconstraintsandobstructionstosymmetries}, the subgroup we want to describe is the stabiliser in $\widehat{W}(A_3^{(1)})$ of the set of classes of nodal curves, which in this case is 
    \begin{equation} \label{eq:stab:laguerre}
        \left\{ w \in \widehat{W}(A_3^{(1)}) ~|~ w(\alpha_1+\alpha_2) = \alpha_1+\alpha_2 \right\}.
    \end{equation}
    The proof is done via a similar computation to that in \cite{dP2symmetry}, in which the following (setwise) stabiliser of a subset of the symmetry root lattice for Sakai surfaces of type $D_5^{(1)}$ was considered:
    \begin{equation}\label{eq:stab:dP2}
        \left\{ w \in \widehat{W}(A_3^{(1)}) ~|~ w\left( {\{ \alpha_0+\alpha_1,\alpha_2+\alpha_3\}}\right) = {\{ \alpha_0+\alpha_1,\alpha_2+\alpha_3\}} \right\}.
    \end{equation}
    In \cite{dP2symmetry}, the subgroup in \cref{eq:stab:dP2} was shown to be isomorphic to the direct product of two copies of $W(A_1^{(1)})\rtimes\Aut(A_1^{(1)})$, i.e. the extension of $W((A_1+\underset{|\alpha|^2=4}{A_1})^{(1)})$ by the whole group $\Z/2\Z\times\Z/2\Z$ of Dynkin diagram automorphisms (that respect root lengths).

    A set of generators for the stabiliser in \cref{eq:stab:dP2} was computed in \cite{dP2symmetry} in two ways, first by direct computation and second using tools from the normalizer theory of parabolic subgroups of Coxeter groups as surveyed in \cite{yangtranslations}.
    To keep the present paper as self-contained as possible, we give a brief account of how the former approach can be applied to the problem at hand of computing the stabiliser in \cref{eq:stab:laguerre}.

    There is a decomposition of $\widehat{W}(A_3^{(1)})$ into finite and translation parts, and we can write any element of $\widehat{W}(A_3^{(1)})$ as 
    $$T_h \,w \,\sigma, \quad h\in \nought{P}, \quad w \in W(A_3),\quad \sigma \in \{1, \sigma_1\}$$
    Then the problem of computing \eqref{eq:stab:laguerre} is reduced by the translation formula \eqref{eq:translationformula} to computing directly, for each of the finitely many $w \sigma$, the set of translations that send $w \sigma(\alpha_1+\alpha_2)$ back to $\alpha_1+\alpha_2$, i.e. 
    $$ \left\{ T_h ~|~ T_h w \sigma(\alpha_1+\alpha_2)=\alpha_1+\alpha_2\right\} = \left\{ T_h ~|~ w \sigma(\alpha_1+\alpha_2)- \left( w \sigma(\alpha_1+\alpha_2)\bullet h \right) \delta = \alpha_1+\alpha_2\right\}.$$
Doing so leads to a set of generators of the stabiliser, which can be reduced straighforwardly to that in \cref{eq:stabilisergeneratorsLUE}. 
The difference between the structure of the subgroup \eqref{eq:stab:dP2} computed in \cite{dP2symmetry} and the that of the subgroup \eqref{eq:stab:laguerre} is that the former included the full group $\Z/2\Z\times\Z/2\Z$ of Dynkin diagram automorphisms of $(A_1+\underset{|\alpha|^2=4}{A_1})^{(1)}$, whereas the latter only allows the single nontrivial automorphism $\tau$ that acts simultaneously on the two subdiagrams, as shown in \Cref{fig:dynkindiagramA1A1}.

    %     To use this to prove the theorem, note that the subgroup \eqref{eq:stab:dP2} is conjugate by, for example, $w_1 w_3$ to
    %     \begin{equation}\label{eq:stab:LUE}
    % \left\{ w \in \widehat{W}(A_3^{(1)}) ~|~ w\left( {\{ \alpha_1+\alpha_2,\alpha_0+\alpha_3\}}\right) = {\{ \alpha_1+\alpha_2,\alpha_0+\alpha_3\}} \right\},
    % \end{equation}
    % and the generators in \eqref{eq:stabilisergeneratorsLUE} can be obtained from those in \cite[Proof of Th. 1]{dP2symmetry} via such a conjugation.
    
    % So after conjugating this description in terms of generators, we have a similar description of the subgroup in \cref{eq:stab:LUE}.
    % Using this, it is straightforward to verify that restricting to the subgroup of this that fixes pointwise, rather than permutes, the subset ${\{ \alpha_1+\alpha_2,\alpha_0+\alpha_3\}}$ is described in terms of generators in \cref{eq:stabilisergeneratorsLUE} as claimed. 
\end{proof}

\subsection{Generalised semi-classical Meixner weight on $\Z_{\geq 0}$} % (fold)
\label{sec:meixner}

We next consider the discrete system \eqref{eq-Meixner-disc} and establish \Cref{mainthm:genMeixner}.
We further omit proofs when they are standard computations.
\begin{notation*}
For the discrete and differential systems from the generalised Meixner weight \eqref{weight-gM} and associated $D_5^{(1)}$ surfaces we use the following notation: coordinates $(x,y)$; parameters $\beta, \gamma, n$; continuous independent variable $c$; centres of blowups $z_{i}$; exceptional divisors $M_{i}$.	
For the birational maps on $\p^1 \times \p^1$ we will use $\phi_{\boldsymbol{b},c}$, for surfaces $Z_{\boldsymbol{b},c}$, and for isomorphisms $\tilde{\phi}_{\boldsymbol{b},c}$.
\end{notation*}

\subsubsection{The discrete and differential systems}

Recall from \cite{SmetVanAssche,FilipukVanAssche} the discrete and differential systems satisfied by $x=x_n(c)$ and $y=y_n(c)$, which are 
\begin{equation}\label{eq-Meixner-disc-body}
	\left\{
		(\bar{x} + y)(x + y) = \frac{(\gamma-1)}{c^2} y (y-c) \left(y- c \frac{\gamma-\beta}{\gamma-1}\right)  ,\quad
		(x + \ubar{y})(x + y) = \frac{x(x+c)}{x- \tfrac{cn}{\gamma-1}} \left(x + c \frac{\gamma-\beta}{\gamma-1}\right),
	\right.
\end{equation}
and 
\begin{equation}\label{eq-Meixner-diff}
\left\{
	\begin{aligned}
		% c^2 \frac{dx}{dc} &= c (n x+n y+x)-(\gamma -1) x (x+y) + \frac{x (c+x) (c (\gamma -\beta )+(\gamma -1) x)}{x+y}  ,\\
		\frac{dx}{dc} &=  \frac{(n+1)x+n y}{c}-\frac{\gamma -1}{c^2} x (x+y) + \frac{x (x+c) (c (\gamma -\beta )+(\gamma -1) x)}{c^2(x+y)}  ,\\
		 \frac{dy}{dc} &= -  x +  \frac{y}{c} + y \left( -1 +\frac{(\beta -1) (y-c)}{c (x+y)} + \frac{(\gamma-1)(y-c)^2}{c^2(x+y)}  \right)
         % -\frac{(\beta -1) y (c-y)}{c (x+y)}
	\end{aligned}
	\right.
\end{equation}
respectively. 
We assume $\gamma\neq1$ and in this section that $\beta\neq \gamma$.
Note that the differential system \eqref{eq-Meixner-diff} is derived by combining the discrete system \eqref{eq-Meixner-disc-body} with the Toda type differential-difference system
\begin{equation}
    \frac{d}{dc} a_n^2 = \frac{a_n^2}{c} (b_n - b_{n-1}),\quad \frac{db_n}{dc} =\frac{1}{c}(a_{n+1}^2 - a_n^2), 
\end{equation}
derived in \cite{BoelenFilipukVanAssche}, in which $a_n=a_n(c), b_n=b_n(c)$ are related to $x_n(c),y_n(c)$ according to 
$$a_n^2 = n c - (\gamma-1)x_n, \quad b_n = n+\gamma - \beta +c - \frac{(\gamma-1)}{c} y_n.$$
In the same way as in previous sections, we consider the discrete system \eqref{eq-Meixner-disc-body} as a family of birational maps
\begin{equation} \label{eq-Meixner-birationalmap}
\begin{aligned}
    \phi_{\boldsymbol{b},c} : \p^1 \times \p^1 &\dashrightarrow \p^1 \times \p^1,\\
    (x,y)&\mapsto (\bar{x},\bar{y}),
\end{aligned}
\end{equation}
with parameter space $\mathscr{B}\ni \boldsymbol{b}=(\beta,\gamma,n)$, and evolution of parameters $\boldsymbol{b}\mapsto\bar{\boldsymbol{b}}$ given by $n\mapsto n+1$.
The independent variable space for the differential system \eqref{eq-Meixner-diff} is $\mathscr{S}=\C\setminus\{0\}\ni c$.

\subsubsection{Surfaces}

We again follow the identification procedure and construct the surfaces on which the discrete system \eqref{eq-Meixner-disc-body} (and the differential system \eqref{eq-Meixner-diff}) is regularised. 

\begin{lemma}
    Through the sequence of eight blowups of points $z_1,\dots,z_8$ given in \Cref{fig:Meixner:pointlocations}, we obtain a rational surface $Z_{\boldsymbol{b},c} = \Bl_{z_1 \cdots z_8}(\p^1 \times \p^1)$ such that the birational map \eqref{eq-Meixner-birationalmap} defined by the discrete system \eqref{eq-Meixner-disc-body} becomes an isomorphism
    $$ \tilde{\phi}_{\boldsymbol{b},s} : Z_{\boldsymbol{b},s} \longrightarrow Z_{\bar{\boldsymbol{b}},s}.$$
\end{lemma}

\begin{figure}[htbp]
\begin{equation}
\begin{gathered}
z_1 : (x,y)=(0,0), \quad z_2 : (x,y)=(-c,c), \quad  z_3 : (x,y) =  \left(\tfrac{c(\beta-\gamma)}{\gamma-1},\tfrac{c(\gamma-\beta)}{\gamma-1} \right),  \\
z_4 : (x,y)=\left( \frac{c n }{\gamma-1}, \infty\right), \\
\begin{aligned}
z_5 &: (x,y)=(\infty, \infty) 							& \leftarrow   	\qquad  	&z_6 : (U_5,V_5) = \left( \frac{y}{x} , \frac{1}{y}\right) = (0,0) \\
	&										& 			 		&\uparrow \\
	&										& 				&z_7 : (U_6, V_6) = \left(  \frac{U_5}{V_5}, V_5\right) = \left(  \frac{c^2}{\gamma-1},0 \right) \\
	&										& 			 		&\uparrow \\
	&										& 				&z_8 : (U_7, V_7) = \left( \frac{U_6 - \frac{c^2}{\gamma-1}}{V_6}, V_6\right) = \left(  \frac{c^3(2 \gamma- \beta + c + n)}{(\gamma-1)^2} ,0\right).
\end{aligned}
\end{gathered}
\end{equation}    
\caption{Blowup point locations for the generalised semi-classical Meixner weight}
    \label{fig:Meixner:pointlocations}
\end{figure}

We denote the identification of all Picard groups $\Pic(Z_{\boldsymbol{b},c})$ into a single lattice by 
$$\Pic(\mathcal{Z}) = \Span_{\Z}\left\{\h_{x},\h_{y},\M_1,\dots,\M_8\right\},$$
where, similarly to above, $\h_x$ and $\h_y$ correspond to pullbacks of classes of fibres of the projection from $\p^1\times\p^1$ onto its two factors and $\M_i$ comes from the exceptional divisor of the blowup of $z_i$. 
The intersection form is given by
$$        \h_x\bullet\h_y=1, \quad \h_x\bullet\h_x=\h_y\bullet\h_y=\h_x\bullet\M_i=\h_y\bullet\M_i=0, \quad \M_i\bullet\M_j=-\delta_{ij}, \quad\text{ for }~i,j=1,\dots,8.
$$
and the anticanonical class corresponds to 
$$-\K_{\mathcal{Z}}:=2\h_x+2\h_y-\M_1-\M_2-\M_3-\M_4-\M_5-\M_6-\M_7-\M_8\in\Pic(\mathcal{Z}).$$

\begin{proposition} \label{lem:prop:ACdiv-pLUE}
        The surface $Z_{\boldsymbol{b},s}$ is a Sakai surface of type $\mathcal{R}=D_5^{(1)}$.
\end{proposition}
\begin{proof}
    The unique effective anticanonical divisor of $Z_{\boldsymbol{b},c}$ is that coming from the pullback under the blowup morphism of the rational 2-form on $\p^1\times\p^1$ defined in the affine chart $(x,y)$ by $\frac{ d x\wedge d y}{x+y}$.
    Its irreducible components $D_i$ are indicated in blue on \Cref{fig:Meixner-soic}.
    The corresponding elements of $\Pic(\mathcal{Z})$ are given in \Cref{fig:d-roots-Meixner}, as well as the associated Dynkin diagram. 
\end{proof}

\begin{figure}[htb]
	\begin{tikzpicture}[>=stealth,basept/.style={circle, draw=red!100, fill=red!100, thick, inner sep=0pt,minimum size=1.2mm}]
	\begin{scope}[xshift=0cm,yshift=0cm]
	\draw [black, line width = 1pt] (-0.4,0) -- (2.9,0)	node [pos=1,right] {\small $y=0$};
	\draw [black, line width = 1pt] (-0.4,2.5) -- (2.9,2.5) node [pos=0,left] {\small $y=\infty$};
	\draw [black, line width = 1pt] (0,-0.4) -- (0,2.9) node [pos=1,above] {\small $x=0$};
	\draw [black, line width = 1pt] (2.5,-0.4) -- (2.5,2.9) node [pos=0,below] {\small $x=\infty$};
	\draw [black, line width = 1pt] (-0.4,-0.4) -- (2.9,2.9) node [pos=0,below left] {\small $x+y=0$};
	\node (z1) at (0,0) [basept,label={[xshift = -7pt, yshift=0pt] \small $z_{1}$}] {};
	\node (z2) at (.7,.7) [basept,label={[xshift=-7pt, yshift=-7pt] \small $z_{2}$}] {};
	\node (z3) at (1.2,1.2) [basept,label={[xshift = -7pt, yshift=-7pt] \small $z_{3}$}] {};
	\node (z4) at (1.5,2.5) [basept,label={[xshift = 0pt, yshift=0pt] \small $z_{4}$}] {};
	\node (z5) at (2.5,2.5) [basept,label={[xshift = -7pt, yshift=0pt] \small $z_{5}$}] {};
	\node (vert) at (2.5,2.1) {} ;
	\node (z6) at (3,1.7) [basept,label={[xshift = 0pt, yshift=-15pt] \small $z_{6}$}] {};
	\node (z7) at (3.625,1.7) [basept,label={[xshift = 0pt, yshift=-15pt] \small $z_{7}$}] {};
	\node (z8) at (4.25,1.7) [basept,label={[xshift = 0pt, yshift=-15pt] \small $z_{8}$}] {};
%	\draw [red, line width = 0.8pt, ->] (z3) -- (1.3536,1.3536) --  (z2);
	% \draw [red, line width = 0.8pt, ->] (z3) -- (1.2,1.2) --  (z2);
	\draw [red, line width = 0.8pt, ->] (z6) -- (2.5,2.1) -- (z5);
	\draw [red, line width = 0.8pt, ->] (z7) -- (z6);	
	\draw [red, line width = 0.8pt, ->] (z8) -- (z7);	
	\end{scope}
	\draw [->] (6.75,1)--(4.75,1) node[pos=0.5, below] {$\operatorname{Bl}_{z_{1}\cdots z_{8}}$};
	\begin{scope}[xshift=8cm,yshift=-1cm]
	\draw [red, line width = 1pt] (0,0.6) -- (0,5)	node [pos=1, above] {\small $\h_{x}-\M_{1}$};
	\draw [red, line width = 1pt] (-0.2,1) -- (1,-0.2)	node [pos=1, below] {\small $\M_{1}$};
	\draw [red, line width = 1pt] (0.6,0) -- (4.2,0)	node [pos=1, right] {\small $\h_{y}-\M_{1}$};
	\draw [red, line width = 1pt] (.6,1.9) -- (1.4,1.0)	node [pos=0, xshift=+3pt,yshift=+7pt] {\small $\M_{2}$};
	\draw [red, line width = 1pt] (1.5,2.8) -- (2.3,1.9)	node [pos=0, xshift=+3pt,yshift=+7pt] {\small $\M_{3}$};
	% \draw [darkgreen, line width = 1pt] (.4,1.4) -- (1.7,2.7)	node [pos=1, xshift=-28pt,yshift=-3pt] {\small $\M_{2}-\M_{3}$};
	% \draw [red, line width = 1pt] (1.5,2.2) -- (1.5,5)	node [pos=1, above] {\small $\h_{x}-\M_{2}$};
	\draw [blue, line width = 1pt] (.2,.2) -- (4.5,4.5)	node [pos=0, below left] {\small $\h_{x}+\h_y-\M_{1}-\M_{2}-\M_{3}-\M_{5}$};
%	\draw [blue, line width = 1pt] (2.8,3.9) -- (4.3,2.4)	node [pos=1, right] {\small $M_{5}-M_{6}$};
	\draw [red, line width = 1pt] (3.1,3.8) -- (1.9,5)		node [pos=1, xshift=-7pt,yshift=+7pt] {\small $\M_{4}$};
	\draw [red, line width = 1pt] (2.9,-.3) -- (2.9,4.3)		node [pos=0, below] {\small $\h_{x}-\M_{4}$};
	\draw [blue, line width = 1pt] (3.4,4.9) -- (5.1,3.2)	node [pos=0, above] {\small $\M_{5}-\M_{6}$};
	\draw [blue, line width = 1pt] (5.1,3.7) -- (3.7,2.3)	node [pos=0, right] {\small $\M_{6}-\M_{7}$};
	\draw [blue, line width = 1pt] (4.1,3.1) -- (5.2,2)		node [pos=1, right] {\small $\M_{7}-\M_{8}$};
	\draw [red, line width = 1pt]  (4.7,1.9) -- (5.3,2.5)		node [pos=1, right] {\small $\M_{8}$};
	\draw [blue, line width = 1pt]  (3.9,-.3) -- (3.9,2.8)	node [pos=0.5, right] {\small $\h_{x}-\M_{5}-\M_{6}$};
	\draw [blue, line width = 1pt] (-0.3,4.7) -- (3.9,4.7)	node [pos=0, left] {\small $\h_{y}-\M_{4}-\M_{5}$};
	\end{scope}
	\end{tikzpicture}
	\caption{The $D_{5}^{(1)}$ Sakai surface from the generalised semi-classical Meixner weight}
	\label{fig:Meixner-soic}
\end{figure}

\begin{figure}[htb]
\begin{equation*}\label{eq:d-roots-dM}			
	\raisebox{-32.1pt}{\begin{tikzpicture}[
			elt/.style={circle,draw=black!100,thick, inner sep=0pt,minimum size=2mm}]
		\path 	(-1,1) 	node 	(d0) [elt, label={[xshift=-10pt, yshift = -10 pt] $\delta_{0}$} ] {}
		        (-1,-1) node 	(d1) [elt, label={[xshift=-10pt, yshift = -10 pt] $\delta_{1}$} ] {}
		        ( 0,0) 	node  	(d2) [elt, label={[xshift=-10pt, yshift = -12 pt] $\delta_{2}$} ] {}
		        ( 1,0) 	node  	(d3) [elt, label={[xshift=10pt, yshift = -12 pt] $\delta_{3}$} ] {}
		        ( 2,1) 	node  	(d4) [elt, label={[xshift=10pt, yshift = -10 pt] $\delta_{4}$} ] {}
		        ( 2,-1) node 	(d5) [elt, label={[xshift=10pt, yshift = -10 pt] $\delta_{5}$} ] {};
		\draw [black,line width=1pt ] (d0) -- (d2) -- (d1)  (d2) -- (d3) (d4) -- (d3) -- (d5);
	\end{tikzpicture}} \qquad
			\begin{alignedat}{2}
			\delta_{0} &= \mathcal{H}_x + \mathcal{H}_y -\M_{1}-\M_{2}-\M_{3}-\M_{5}, &\qquad  \delta_{3} &=\M_6 - \M_7,\\
			\delta_{1} &= \mathcal{H}_y - \M_4 - \M_5, &\qquad  \delta_{4} &= \mathcal{H}_x - \M_5 - \M_6,\\
			\delta_{2} &= \M_5 - \M_6, &\qquad  \delta_{5} &= \M_7 - \M_8.
			\end{alignedat}
\end{equation*}
	\caption{The surface root basis from the generalised semi-classical Meixner weight}
	\label{fig:d-roots-Meixner}	
\end{figure}
\subsubsection{Induced dynamics on the Picard lattice}

\begin{lemma} \label{lem:induceddynamicsonPicMeixner}
The isomorphisms $\tilde{\phi}_{\boldsymbol{b},c}$ induce via pushforward $(\tilde{\phi}_{\boldsymbol{b},c})_* : \Pic(Z_{\boldsymbol{b},c})\to \Pic(Z_{\bar{\boldsymbol{b}},c})$ the following action on $\Pic(\mathcal{Z})$:
\begin{equation}
\Phi : \left\{
	\begin{aligned}
		\mathcal{H}_x 	&\mapsto 4 \h_x + 3 \h_y - 2 \M_1 - 2 \M_2 - 2 \M_3 - 3 \M_4 - \M_5 - \M_6 - \M_7, \\\
		\mathcal{H}_y 	&\mapsto 2 \h_x + \h_y - \M_1 -\M_2 - \M_3 - \M_4, \\
		\M_1			&\mapsto  \h_x + \h_y -\M_2 - \M_3 -\M_4, \\
		\M_2			&\mapsto  \h_x + \h_y -\M_1 - \M_3 -\M_4, \\
		\M_3			&\mapsto  \h_x + \h_y -\M_1 - \M_2 -\M_4, \\
		\M_4			&\mapsto  \M_8, \\
		\M_5			&\mapsto  2 \h_x + \h_y -\M_1 - \M_2 -\M_3 - \M_4 - \M_7, \\
		\M_6			&\mapsto  2 \h_x + \h_y -\M_1 - \M_2 -\M_3 - \M_4 - \M_6, \\
		\M_7			&\mapsto  2 \h_x + \h_y -\M_1 - \M_2 -\M_3 - \M_4 - \M_5, \\
		\M_8			&\mapsto  \h_x - \M_4. 
	\end{aligned}
\right.
\end{equation}
\end{lemma}

\subsubsection{Identification with the standard model}

We present the identification of $\Pic(\X)$ and $\Pic(\mathcal{Z})$ that matches the linear dynamics in \Cref{lem:induceddynamicsonPicMeixner} with the translation $T_{\operatorname{Sak}}$ associated to the discrete Painlev\'e equation \eqref{eq-Sakai-dP-S3}.
\begin{lemma}\label{lem:KNY-to-Meixner-Pic} 
The identification $\Pic(\mathcal{Z}) \to \Pic(\X)$ of Picard lattices between the standard Kajiwara-Noumi-Yamada
	and generalised Meixner surfaces is given as follows:
\begin{equation}
		\begin{aligned}
			\mathcal{H}_{q} & = \mathcal{H}_x + 2 \mathcal{H}_y - \M_1 - \M_2 - \M_5 - \M_6, 
    %         &\qquad 
				% \mathcal{H}_{x} &= 3 \mathcal{H}_{q} + 2 \mathcal{H}_p - \E_1 - \E_3 - \E_5 - \E_6 - 2 \E_7 - 2 \E_8,
                \\
			\mathcal{H}_{p} &= 2\mathcal{H}_{x} + 3\mathcal{H}_{y} - \M_{1} - 2 \M_{2} - 2 \M_{5} - \M_6 - \M_{7} - \M_8,  
    %         &\qquad 	
				% \mathcal{H}_{y} & = 2 \mathcal{H}_{q} + \mathcal{H}_{p} - \mathcal{E}_{1} - \mathcal{E}_{5} - \mathcal{E}_{7} - \mathcal{E}_{8}, 
                \\
			\mathcal{E}_{1} &= \mathcal{H}_x + \mathcal{H}_y - \M_1 - \M_1 - \M_5, 
    %         &\qquad 
				% \mathcal{F}_{1}	&= \mathcal{H}_{q}+\mathcal{H}_p - \E_1 - \E_7 - \E_8,
                \\ 
			\mathcal{E}_{2} &= \M_3, 
    %         &\qquad 
				% \mathcal{F}_{2}	&= 2 \mathcal{H}_{q}+ \mathcal{H}_p - \E_1 - \E_5 - \E_6 - \E_7 - \E_8,
                \\ 
			\mathcal{E}_{3} &= \mathcal{H}_{x} - \mathcal{M}_{5}, 
    %         &\qquad 
				% \mathcal{F}_{3}	&= \E_{2},
                \\ 
			\mathcal{E}_{4} &= \M_4, 
    %         &\qquad 
				% \mathcal{F}_{4}	&= \E_{4},
                \\ 
			\mathcal{E}_{5} &= \mathcal{H}_x + \mathcal{H}_y -\M_2 - \M_5 - \M_6, 
    %         &\qquad 
				% \mathcal{F}_{5}	&= 2 \mathcal{H}_q - \mathcal{H}_p - \E_1 - \E_3 - \E_5 - \E_7 - \E_8,
                \\ 
			\mathcal{E}_{6} &= \mathcal{H}_y - \M_2, 
    %         &\qquad 
				% \mathcal{F}_{6}	&= \mathcal{H}_{q} + \mathcal{H}_p - \E_5 - \E_7 - \E_8,
                \\ 
			\mathcal{E}_{7} &= \mathcal{H}_x + 2 \mathcal{H}_y - \M_1 - \M_2 - \M_5 - \M_6 - \M_8, 
    %         &\qquad 
				% \mathcal{F}_{7}	&= \mathcal{H}_{q} - \mathcal{E}_{8},
                \\ 
			\mathcal{E}_{8} &=  \mathcal{H}_x + 2 \mathcal{H}_y - \M_1 - \M_2 - \M_5 - \M_6 - \M_7.
    %         &\qquad 
				% \mathcal{F}_{8}	&= \mathcal{H}_{q} - \mathcal{E}_{7}.
		\end{aligned}
	\end{equation}
	This results in the following correspondences between the surface roots:
	\begin{equation}\label{eq:geom-KNY-Meixner}
	\begin{aligned}
		\delta_{0} &= \mathcal{E}_{1} - \mathcal{E}_{2} = \mathcal{H}_x + \mathcal{H}_y - \M_1 - \M_2 -\M_3 - \M_5, &\quad
			\delta_{3} &=\mathcal{H}_{p} - \mathcal{E}_{5} - \mathcal{E}_{7} = 
				\M_6 - \M_7, \\
		\delta_{1} &= \mathcal{E}_{3} - \mathcal{E}_{4} = \mathcal{H}_y - \M_4 - \M_5, &\quad
			\delta_{4} &= \mathcal{E}_{5} - \mathcal{E}_{6}  = \mathcal{H}_x - \M_5 -\M_6, \\
		\delta_{2} &=  \mathcal{H}_{q} - \mathcal{E}_{1} - \mathcal{E}_{3} = \M_5 -\M_6, &\quad
			\delta_{5} &= \mathcal{E}_{7} - \mathcal{E}_{8} = \M_7 - \M_8.
	\end{aligned}
	\end{equation}
	The symmetry root basis in $\Pic(\X)$ given in \Cref{fig:a-roots-a3-KNY} maps to $\Pic(\mathcal{Z})$ according to
    \begin{equation}\label{eq:sym-KNY-Meixner}
	\begin{aligned}
			\alpha_{0} &= \mathcal{H}_{p} - \mathcal{E}_{1} - \mathcal{E}_{2}
		=  \mathcal{H}_{x} + 2 \mathcal{H}_{y} - \M_2 - \M_3 - \M_5 - \M_6 - \M_7 - \M_8, \\
			\alpha_{1} &= \mathcal{H}_{q} - \mathcal{E}_{5} - \mathcal{E}_{6} = - \mathcal{M}_{1} + \mathcal{M}_{2}, \\
			\alpha_{2} &= \mathcal{H}_{p} - \mathcal{E}_{3} - \mathcal{E}_{4} = 
			2 \mathcal{H}_x + 2 \mathcal{H}_y - \M_1 - 2 \M_2 - \M_4 - \M_5 -\M_6 - \M_7 - \M_8,\\
			\alpha_{3} &= \mathcal{H}_{q} - \mathcal{E}_{7} - \mathcal{E}_{8} = -  \mathcal{H}_{x} - 2 \mathcal{H}_y + \M_1 + \M_2 + \M_5 + \M_6 + \M_7 + \M_8.
		\end{aligned}
	\end{equation}
\end{lemma}

Under the identification in \Cref{lem:KNY-to-Meixner-Pic}, the discrete system \eqref{eq-Meixner-disc-body} corresponds to the following action on the symmetry roots
\begin{equation}\label{eq-PhiOnRoots-Meixner}
\Phi : 
\left\{
	\begin{aligned}
		\alpha_0 &\mapsto \alpha_0 - \delta, \\
		\alpha_1 &\mapsto \alpha_1, \\
		\alpha_2 &\mapsto \alpha_2, \\
		\alpha_3 &\mapsto \alpha_3 + \delta,	
	\end{aligned}
\right.
\end{equation}
which coincides with that of $T_{\operatorname{Sak}}$, given in \cref{eq:translationonroots-Sakai}, which generates the Sakai discrete Painlev\'e equation \eqref{eq-Sakai-dP-S3}.

\begin{lemma} \label{lem:rootvarsgenMeixner}
    The root variables $a_i = \chi(\alpha_i)$, $i=0,1,2,3$, defined with respect to the period map $\chi$ given by the 2-form on $Z_{\boldsymbol{b},c}$ defined by $ k\frac{dx \wedge dy}{x+y}$, $k=\frac{1-\gamma}{c}$, for the symmetry roots 
   \begin{equation*} 
   \begin{aligned}
    \alpha_{0} &=  \mathcal{H}_{x} + 2 \mathcal{H}_{y} - \M_2 - \M_3 - \M_5 - \M_6 - \M_7 - \M_8, \\
	\alpha_{1} &= - \mathcal{M}_{1} + \mathcal{M}_{2}, \\
	\alpha_{2} &= 2 \mathcal{H}_x + 2 \mathcal{H}_y - \M_1 - 2 \M_2 - \M_4 - \M_5 -\M_6 - \M_7 - \M_8, \\
    \alpha_{3} &= -  \mathcal{H}_{x} - 2 \mathcal{H}_y + \M_1 + \M_2 + \M_5 + \M_6 + \M_7 + \M_8,
    \end{aligned}
    \end{equation*}
     are given by 
	\begin{equation}
		a_{0} =  n+1,\quad  a_{1} =\gamma-1 ,\quad
		a_{2} = 2-\beta , \quad a_{3} = -\gamma - n,
	\end{equation}
    so in particular this matches the normalisation $a_0+a_1+a_2+a_3=1$.
\end{lemma}

We are now ready to establish the remaining part of \Cref{mainthm:genMeixner}.
\begin{proposition}\label{prop:coords-KNY-Meixner}
The following change of variables and parameter matching simultaneously identifies the discrete and differential systems \eqref{eq-Meixner-disc-body} and \eqref{eq-Meixner-diff} from the generalised Meixner weight with the Sakai discrete Painlev\'e equation \eqref{eq-Sakai-dP-S3} and the KNY Hamiltonian form \eqref{eq-KNY-Ham5-sys} of $\pain{V}$ respectively: 
\begin{equation}\label{eq:KNYtoMeixner}
   	 \left\{
		\begin{aligned}
   	 	q&= \frac{(\gamma-1)y(y-c)}{c^2(x+y)},\\
   		p&= \frac{ c(x+c)}{y-c} + \frac{ c^2 (\gamma+n+1-\beta)(x+y)}{c^2 x + y \left( c(\gamma+c-1) - (\gamma-1)y\right)}, \\
	   	 \end{aligned}
		 \qquad
		 \begin{aligned}
		 a_{0}&= n+1,		&\quad &a_{1}=\gamma-1,\\  
		a_{2}&= 2-\beta,   		&\quad &a_{3}=- \gamma - n, 
		\end{aligned}
	\qquad t=c.
	\right.
    \end{equation}
\end{proposition}

Similarly to the case of the perturbed Laguerre weight in \Cref{sec:defLUE}, if $n$ is assumed generic and complex, the surfaces $Z_{\boldsymbol{b},c}$ are generic and we have the full symmetry group $\widehat{W}(A_3^{(1)})$.

\subsection{Semi-classical Meixner weight on $\Z_{\geq 0}$} % (fold)
\label{sec:degenMeixner}

\begin{notation*}
For the discrete and differential systems from the semi-classical Meixner weight \eqref{weight-M} and associated $D_5^{(1)}$ surfaces we use the same notation as for the generalised case:
coordinates $(x,y)$; parameters $\gamma, n$; continuous independent variable $c$; centres of blowups $z_{i}$; exceptional divisors $M_{i}$; birational maps $\phi_{\boldsymbol{b},c}$; surfaces $Z_{\boldsymbol{b},c}$; isomorphisms $\tilde{\phi}_{\boldsymbol{b},c}$; induced action $\Phi$ on $\Pic(\mathcal{Z})$.
\end{notation*}

\subsubsection{The discrete and differential systems}

We first recall the discrete and differential systems derived in  \cite{BoelenFilipukVanAssche}, which are
\begin{equation}\label{eq-degen-Meixner-disc-body}
	\left\{
		(\bar{x} + y)(x + y) = \frac{(\gamma-1)}{c^2} y (y-c)^2  ,\quad
		(x + \ubar{y})(x + y) = \frac{x(x+c)^2}{x- \tfrac{cn}{\gamma-1}},
	\right.
\end{equation}
and 
\begin{equation}\label{eq-degen-Meixner-diff}
\left\{
	\begin{aligned}
		\frac{dx}{dc} &=  \frac{(n+1)x+n y}{c}-\frac{\gamma -1}{c^2} x (x+y) + \frac{(\gamma -1)x (x+c)^2}{c^2(x+y)}  ,\\
		 \frac{dy}{dc} &= -  x +  \frac{y}{c} + y \left( -1 + \frac{(\gamma-1)(y-c)^2}{c^2(x+y)}  \right),
         % -\frac{(\beta -1) y (c-y)}{c (x+y)}
	\end{aligned}
	\right.
\end{equation}
respectively.
Note that these can be obtained by setting $\beta=1$ in the systems \eqref{eq-Meixner-disc-body} and \eqref{eq-Meixner-diff}.

\subsubsection{Surfaces}

The surfaces on which the discrete system \eqref{eq-degen-Meixner-disc-body} is regularised are essentially the result of setting $\beta=1$ in those from the generalised Meixner case above, but again this degeneration causes a nodal curve to appear, so we provide some details.
\begin{lemma}
    Through the sequence of eight blowups of points $z_1,\dots,z_8$ given in \Cref{fig:degenMeixner:pointlocations}, we obtain a rational surface $Z_{\boldsymbol{b},c} = \Bl_{z_1 \cdots z_8}(\p^1 \times \p^1)$ such that the birational map defined by the discrete system \eqref{eq-degen-Meixner-disc-body} becomes an isomorphism
    $$ \tilde{\phi}_{\boldsymbol{b},s} : Z_{\boldsymbol{b},s} \longrightarrow Z_{\bar{\boldsymbol{b}},s}.$$
\end{lemma}

\begin{figure}[htb]
\begin{equation}
\begin{gathered}
z_1 : (x,y)=(0,0), \quad z_2 : (x,y)=(-c,c) ~~  \leftarrow~~ z_3 : (U_2,V_2) = \left(  \frac{x+c}{y-c},y-c \right) = (-1, 0),  \\
z_4 : (x,y)=\left( \frac{c n }{\gamma-1}, \infty\right), \\
\begin{aligned}
z_5 &: (x,y)=(\infty, \infty) 							& \leftarrow   	\qquad  	&z_6 : (U_5,V_5) = \left( \frac{y}{x} , \frac{1}{y}\right) = (0,0) \\
	&										& 			 		&\uparrow \\
	&										& 				&z_7 : (U_6, V_6) = \left(  \frac{U_5}{V_5}, V_5\right) = \left(  \frac{c^2}{\gamma-1},0 \right) \\
	&										& 			 		&\uparrow \\
	&										& 				&z_8 : (U_7, V_7) = \left( \frac{U_6 - \frac{c^2}{\gamma-1}}{V_6}, V_6\right) = \left(  \frac{c^3(2 \gamma+ c + n-1)}{(\gamma-1)^2} ,0\right).
\end{aligned}
\end{gathered}
\end{equation}    
\caption{Blowup point locations for the semi-classical Meixner weight}
    \label{fig:degenMeixner:pointlocations}
\end{figure}

We give an illustration of the configuration of exceptional divisors $M_1,\dots,M_8$ arising from the blowups in \Cref{fig:degenMeixner:pointlocations}, with curves labelled by their classes in $\Pic(Z_{\boldsymbol{b},s})$.
In constrast to the generic case in \Cref{fig:Meixner-soic}, the point $z_3$ is now infinitely near to $z_2$, lying at the intersection of $M_2$ and the strict transform of the curve on $\p^1 \times \p^1$ defined by $x+y =0$.
This again means there is a nodal curve, which we indicate in green in \Cref{fig:DM-soic}.

\begin{figure}[htbp]
	\begin{tikzpicture}[>=stealth,basept/.style={circle, draw=red!100, fill=red!100, thick, inner sep=0pt,minimum size=1.2mm}]
	\begin{scope}[xshift=0cm,yshift=0cm]
	\draw [black, line width = 1pt] (-0.4,0) -- (2.9,0)	node [pos=1,right] {\small $y=0$};
	\draw [black, line width = 1pt] (-0.4,2.5) -- (2.9,2.5) node [pos=0,left] {\small $y=\infty$};
	\draw [black, line width = 1pt] (0,-0.4) -- (0,2.9) node [pos=1,above] {\small $x=0$};
	\draw [black, line width = 1pt] (2.5,-0.4) -- (2.5,2.9) node [pos=0,below] {\small $x=\infty$};
	\draw [black, line width = 1pt] (-0.4,-0.4) -- (2.9,2.9) node [pos=0,below left] {\small $x+y=0$};
	\node (z1) at (0,0) [basept,label={[xshift = -7pt, yshift=0pt] \small $z_{1}$}] {};
	\node (z2) at (1,1) [basept,label={[xshift=5pt, yshift=-15pt] \small $z_{2}$}] {};
	\node (z3) at (1.2,1.75) [basept,label={[xshift = -7pt, yshift=-7pt] \small $z_{3}$}] {};
	\node (z4) at (1.5,2.5) [basept,label={[xshift = 0pt, yshift=0pt] \small $z_{4}$}] {};
	\node (z5) at (2.5,2.5) [basept,label={[xshift = -7pt, yshift=0pt] \small $z_{5}$}] {};
	\node (vert) at (2.5,2.1) {} ;
	\node (z6) at (3,1.7) [basept,label={[xshift = 0pt, yshift=-15pt] \small $z_{6}$}] {};
	\node (z7) at (3.625,1.7) [basept,label={[xshift = 0pt, yshift=-15pt] \small $z_{7}$}] {};
	\node (z8) at (4.25,1.7) [basept,label={[xshift = 0pt, yshift=-15pt] \small $z_{8}$}] {};
%	\draw [red, line width = 0.8pt, ->] (z3) -- (1.3536,1.3536) --  (z2);
	\draw [red, line width = 0.8pt, ->] (z3) -- (1.2,1.2) --  (z2);
	\draw [red, line width = 0.8pt, ->] (z6) -- (2.5,2.1) -- (z5);
	\draw [red, line width = 0.8pt, ->] (z7) -- (z6);	
	\draw [red, line width = 0.8pt, ->] (z8) -- (z7);	
	\end{scope}
	\draw [->] (6.75,1)--(4.75,1) node[pos=0.5, below] {$\operatorname{Bl}_{z_{1}\cdots z_{8}}$};
	\begin{scope}[xshift=8cm,yshift=-1cm]
	\draw [red, line width = 1pt] (0,0.6) -- (0,5)	node [pos=1, xshift=-10pt,yshift=10pt] {\small $\h_{x}-\M_{1}$};
	\draw [red, line width = 1pt] (-0.2,1) -- (1,-0.2)	node [pos=1, below] {\small $\M_{1}$};
	\draw [red, line width = 1pt] (0.6,0) -- (4.2,0)	node [pos=1, right] {\small $\h_{y}-\M_{1}$};
	\draw [red, line width = 1pt] (.4,1.8) -- (1.3,0.9)	node [pos=1, right] {\small $\M_{3}$};
	\draw [darkgreen, line width = 1pt] (.4,1.4) -- (1.7,2.7)	node [pos=1, xshift=-28pt,yshift=-3pt] {\tiny $\M_{2}-\M_{3}$};
	\draw [red, line width = 1pt] (1.5,2.2) -- (1.5,5)	node [pos=1, xshift=-7pt,yshift=10pt] {\small $\h_{x}-\M_{2}$};
	\draw [blue, line width = 1pt] (.2,.2) -- (4.5,4.5)	node [pos=0, left] {\small $\h_{x}+\h_y-\M_{1,2,3,5}$};
%	\draw [blue, line width = 1pt] (2.8,3.9) -- (4.3,2.4)	node [pos=1, right] {\small $M_{5}-M_{6}$};
	\draw [red, line width = 1pt] (3.1,3.8) -- (1.9,5)		node [pos=0, xshift=9pt,yshift=5pt] {\small $\M_{4}$};
	\draw [red, line width = 1pt] (2.9,-.3) -- (2.9,4.3)		node [pos=0, below] {\small $\h_{x}-\M_{4}$};
	\draw [blue, line width = 1pt] (3.4,4.9) -- (5.1,3.2)	node [pos=0, above] {\small $\M_{5}-\M_{6}$};
	\draw [blue, line width = 1pt] (5.1,3.7) -- (3.7,2.3)	node [pos=0, right] {\small $\M_{6}-\M_{7}$};
	\draw [blue, line width = 1pt] (4.1,3.1) -- (5.2,2)		node [pos=1, right] {\small $\M_{7}-\M_{8}$};
	\draw [red, line width = 1pt]  (4.7,1.9) -- (5.3,2.5)		node [pos=1, right] {\small $\M_{8}$};
	\draw [blue, line width = 1pt]  (3.9,-.3) -- (3.9,2.8)	node [pos=0.5, right] {\small $\h_{x}-\M_{5}-\M_{6}$};
	\draw [blue, line width = 1pt] (-0.3,4.7) -- (3.9,4.7)	node [pos=0, left] {\small $\h_{y}-\M_{4}-\M_{5}$};
	\end{scope}
	\end{tikzpicture}
	\caption{The $D_{5}^{(1)}$ Sakai surface (with nodal curve) from the semi-classical Meixner weight}
	\label{fig:DM-soic}
\end{figure}

The following is established along the same lines as \Cref{lem:prop:ACdiv-LUE}.
\begin{proposition} \label{lem:prop:ACdiv-degenMeixner}
        $Z_{\boldsymbol{b},c}$ is a Sakai surface of type $\mathcal{R}=D_5^{(1)}$ with $\M_2-\M_3\in \Delta^{\operatorname{nod}}$ for every $\boldsymbol{b},c$.
\end{proposition}

\subsubsection{Identification with the standard model}

Again identifying all $\Pic(Z_{\boldsymbol{b},c})$ into the single lattice $\Pic(\mathcal{Z})$, the linear action induced on $\Pic(\mathcal{Z})$ is the same as in \Cref{lem:induceddynamicsonPicMeixner}, so we can use the same identification on the level of Picard lattices as long as we give consideration to the nodal curve.
\begin{lemma}\label{lem:KNY-to-degenMeixner} 
The identification $\Pic(\mathcal{Z}) \to \Pic(\X)$ of Picard lattices can be taken to be the same as in \Cref{lem:KNY-to-Meixner-Pic}.
	This results in the same matching of surface and symmetry root bases, but the class of the nodal curve on $Z_{\boldsymbol{b},c}$ corresponds under this identification to
	\begin{equation}
		\M_2 - \M_3 = \alpha_0+\alpha_1+\alpha_3 = \delta - \alpha_2 = 2\mathcal{H}_q + \mathcal{H}_p - \E_1 - \E_2 - \E_5 -\E_6-\E_7-\E_8.
	\end{equation}
\end{lemma}
Under the identification, we have the same action on symmetry roots as in \cref{eq-PhiOnRoots-Meixner}, which coincides with that of the Sakai discrete Painlev\'e equation \eqref{eq-Sakai-dP-S3}.
The presence of a nodal curve again leads to a constraint on root variables.

\begin{proposition} \label{prop:rootvarsdegenMeixner}
    The root variables $a_i=\chi(\alpha_i)$, $i=0,1,2,3$, defined with respect to the period map $\chi$ given by the 2-form on $Z_{\boldsymbol{b},c}$ defined by $ k\frac{dx \wedge dy}{x+y}$, $k=\frac{1-\gamma}{c}$, are given by 
	\begin{equation}
		a_{0} =  n+1,\quad  a_{1} =\gamma-1 ,\quad
		a_{2} = 1 , \quad a_{3} = -\gamma - n.
	\end{equation}
\end{proposition}

In order to obtain the change of variables to the special case of the KNY discrete Painlev\'e equation \eqref{eq-KNY-dP-S3} with parameter constraint $a_2=1$, we need to construct the subfamily of surfaces $X_{\boldsymbol{a},t}$ with this constraint on root variables.
Again this is straightforward given that we have identified in \Cref{lem:KNY-to-degenMeixner} that the nodal curve should correspond to 
$$\delta- \alpha_2 = 2\mathcal{H}_q + \mathcal{H}_p - \E_1 - \E_2 - \E_5 -\E_6-\E_7-\E_8,$$
so the subfamily of surfaces should have the points $b_1,b_2,b_5,b_6,b_7,b_8$ lying on a curve in $\p^1 \times \p^1$ of bidegree (2,1).
The condition for this to happen is precisely that $\chi(\delta-\alpha_2)=\chi(\alpha_0+\alpha_1+\alpha_3)=1-a_2=0$, which is as expected according to \Cref{prop:nodalcurvesrootvars}.
The following is established similarly to \Cref{prop:coords-KNY-pLUE,prop:coords-KNY-LUE,prop:coords-KNY-Meixner}.

\begin{proposition}\label{prop:coords-KNY-dM} The following change of variables and parameter matching simultaneously identifies the discrete and differential systems \eqref{eq-degen-Meixner-diff} and \eqref{eq-degen-Meixner-disc-body} with the KNY Hamiltonian form \eqref{eq-KNY-Ham5-sys} of $\pain{V}$ and the Sakai discrete Painlev\'e equation \eqref{eq-Sakai-dP} respectively, in the special case $a_2 = 1$: 
    \begin{equation}\label{eq:KNYtodM-5}
   	 \left\{
		\begin{aligned}
   	 	q&= \frac{(\gamma-1)y(y-c)}{c^2(x+y)},\\
   		p&= \frac{ c(x+c)}{y-c} + \frac{ c^2 (\gamma+n)(x+y)}{c^2 x + y \left( c(\gamma+c-1) - (\gamma-1)y\right)}, \\
	   	 \end{aligned}
		 \qquad
		 \begin{aligned}
		 a_{0}&= n+1,		&\quad &a_{1}=\gamma-1,\\  
		a_{2}&= 1,   		&\quad &a_{3}=- \gamma - n, 
		\end{aligned}
	\qquad t=c.
	\right.
    \end{equation}
\end{proposition}

\subsubsection{Symmetry type}

As in the case of the unperturbed Laguerre weight, the nodal curve and associated parameter constraint mean that we have a symmetry type corresponding to a proper subgroup of $\widehat{W}(A_3^{(1)})$.
This is isomorphic to the subgroup in \Cref{prop:symmetry-type-LUE}, but with a different embedding in $\widehat{W}(A_3^{(1)})$, which we give in the following.

\begin{proposition} \label{prop:symmetry-type-Meixner}
    The subgroup of $\widehat{W}(A_3^{(1)})=W(A_3^{(1)})\rtimes \Aut(A_3^{(1)})$ whose Cremona action is compatible with the parameter constraint  $a_2=1$ on root variables $\boldsymbol{a}=(a_0,a_1,a_2,a_3)\in \mathscr{A}$ is 
    $$
    \left< s_0, s_1, w_0, w_1, \tau \right> \cong \left(W(A_1^{(1)})\times W(A_1^{(1)})\right)\rtimes \Z/2\Z \cong  W( (A_1 + \underset{|\alpha|^2=4}{A_1})^{(1)} ) \rtimes \Z / 2\Z,$$
    where the generators are given in terms of the simple reflections $w_i$ and $A_3^{(1)}$ Dynkin diagram automorphisms by
    \begin{equation} \label{eq:stabilisergeneratorsmeixner}
        \begin{aligned}
            &s_0 = w_0, &&s_1= w_3 w_1 w_2 w_1 w_3, && \langle s_0,s_1\rangle\cong W(A_1^{(1)}), &&\tau s_0 = s_1 \tau,
            \\
            &s'_0 = \sigma_2 w_3 w_2 w_0 w_3, &&s'_1= \sigma_1\sigma_2\sigma_1 , && \langle s'_0,s'_1\rangle\cong W({A_1^{(1)}}), &&\tau s_0' = s_1' \tau,
        \end{aligned}
         \quad 
        \quad \tau = \sigma_2\sigma_1\sigma_2 w_2 w_3.
    \end{equation}
\end{proposition}
\begin{proof}
    In this case, the subgroup of symmetries compatible with the nodal curve and parameter constraint is
    \begin{equation}
        \left\{ w \in \widehat{W}(A_3^{(1)}) ~|~ w(\delta-\alpha_2) = \delta-\alpha_2 \right\} = \left\{ w \in \widehat{W}(A_3^{(1)}) ~|~ w(\alpha_2) = \alpha_2 \right\},
    \end{equation}
    which can conjugated by, for example, $w_1$, to 
    \begin{equation}\label{eq:stab:meixner}
        \left\{ w \in \widehat{W}(A_3^{(1)}) ~|~ w\left(\alpha_1+\alpha_2\right) = \alpha_1+\alpha_2 \right\}.
    \end{equation}
    This stabiliser was computed in \Cref{prop:symmetry-type-LUE}, so conjugating the generators in \eqref{eq:stabilisergeneratorsLUE} by $w_1$ gives generators for the subgroup \eqref{eq:stab:meixner}, from which the claimed expressions in \eqref{eq:stabilisergeneratorsmeixner} can be obtained.
\end{proof}

This subgroup can again be described in terms of a root lattice of type $(A_1+\underset{|\alpha|^2=4}{A_1})^{(1)}$ in $\Pic(\X)$, which in this case is $Q_1+Q_2$, where 
\begin{equation}
\begin{aligned}
Q_1 &= \Z \beta_0 + \Z \beta_1, 
&&\quad
\begin{aligned}
\beta_0 &= \alpha_0 = \h_q  - \E_1 -\E_2, \\
\beta_1 &= \alpha_1 + \alpha_2 + \alpha_3 = 2\h_q + \h_p - \E_3 -\E_4 -\E_5-\E_6 -\E_7 -\E_8, 
\end{aligned}
\\
Q_2 &= \Z \beta_0' + \Z \beta_1', 
&&\quad
\begin{aligned}
\beta_0' &= \alpha_0 + \alpha_2 + 2 \alpha_3 = 2 \h_q + 2 \h_p -\E_1 - \E_2 -\E_3 -\E_4 - 2\E_7 - 2 \E_8, \\
\beta_1' &= \alpha_1 - \alpha_3 = - \E_5 - \E_6 + \E_7 + \E_8.
\end{aligned}
\end{aligned}
\end{equation}
The generators $s_0,s_1,s_0',s_1'$ in \Cref{prop:symmetry-type-Meixner} act on the lattice in the same way as shown in \Cref{eq:sonbetas}, and $\tau$ again corresponds to the Dynkin diagram automorphism shown in \Cref{fig:dynkindiagramA1A1}.

\subsection{Relation between differential systems from perturbed Laguerre and generalised Meixner weights}

Since the translations $T_{\operatorname{KNY}}$ and $T_{\operatorname{Sak}}$ are not conjugate in $\widehat{W}(A_3^{(1)})$, there cannot exist any birational change of variables relating the discrete systems \eqref{eq-pLUE-disc-body} and \eqref{eq-Meixner-disc-body} from the perturbed Laguerre and generalised Meixner weights, respectively. 
However, if we disregard the discrete dynamics and consider only the differential systems, a relation between these can be obtained, since conjugation of group elements is not required for this.
While this identification has no exact meaning for the recurrence coefficients themselves, since the discrete systems are disregarded, it was used in \cite{ondifferentialsystems} to obtain new double scaling limits for the differential system from the generalised Meixner weight, based on known ones for the perturbed Laguerre case, in which large $n$ behaviours are governed by solutions of the $\pain{XXXIV}$ equation.

To differentiate between parameters in the differential systems \eqref{eq-pLUE-diff} and \eqref{eq-Meixner-diff}, consider them in the forms 
\begin{equation}\label{eq-pLUE-diff-relsection}
	\left\{
	\begin{aligned}
		\frac{df}{ds} &=\frac{(\gamma_{\operatorname{L}}+2 g) f^2  - 4 f g+ (\alpha_{\operatorname{L}}-\gamma_{\operatorname{L}}-s +2n_{\operatorname{L}})f + 2 g - \alpha_{\operatorname{L}} - 2n_{\operatorname{L}}}{s} ,\\
		\frac{dg}{ds} &= \frac{g^2(1-f^2) - (\alpha_{\operatorname{L}}+2n_{\operatorname{L}}+\gamma_{\operatorname{L}} f^2)g + (\alpha_{\operatorname{L}}+n_{\operatorname{L}})n_{\operatorname{L}}}{s f},
	\end{aligned}
	\right.
\end{equation}
and 
\begin{equation}\label{eq-Meixner-diff-relsection}
\left\{
	\begin{aligned}
		\frac{dx}{dc} &=  \frac{(n_{\operatorname{M}}+1)x+n_{\operatorname{M}} y}{c}-\frac{\gamma_{\operatorname{M}} -1}{c^2} x (x+y) + \frac{x (x+c) (c (\gamma_{\operatorname{M}} -\beta_{\operatorname{M}} )+(\gamma_{\operatorname{M}} -1) x)}{c^2(x+y)}  ,\\
		 \frac{dy}{dc} &= -  x +  \frac{y}{c} + y \left( -1 +\frac{(\beta_{\operatorname{M}} -1) (y-c)}{c (x+y)} + \frac{(\gamma_{\operatorname{M}}-1)(y-c)^2}{c^2(x+y)}  \right),
         % -\frac{(\beta -1) y (c-y)}{c (x+y)}
	\end{aligned}
	\right.
\end{equation}
respectively.
The relation between these is obtained from an identification of $\Pic(\Y)$ and $\Pic(\mathcal{Z})$ that, while not matching the translations associated with the discrete systems, can be realised by an isomorphism $Y_{\boldsymbol{b},s}\to Z_{\boldsymbol{b},c}$ under appropriate parameter matching.
    The identification is given by 
\begin{equation}
		\begin{aligned}
			\mathcal{H}_{f} & = 2\mathcal{H}_x +  \mathcal{H}_y - \M_1 - \M_2 - \M_4 - \M_5, 
            &&
                \\
			\mathcal{H}_{g} &= \h_x,  
            &&
                \\
			\mathcal{L}_{1} &= \h_x-\M_4, 
                \\ 
			\mathcal{L}_{2} &= \M_3, 
                \\ 
			\mathcal{L}_{3} &= \mathcal{H}_{x} - \M_{5}, 
                \\ 
			\mathcal{L}_{4} &= \M_6, 
                \\ 
			\mathcal{L}_{5} &= \M_7, 
                \\ 
			\mathcal{L}_{6} &= \M_8, 
                \\ 
			\mathcal{L}_{7} &= \mathcal{H}_x - \M_1, 
                \\ 
			\mathcal{L}_{8} &=  \mathcal{H}_x - \M_2.
		\end{aligned}
	\end{equation}

The isomorphism realising this leads to the following, in which the parameter matching is obtained by computation of root variables.
The following can be verified by direct computation.
\begin{proposition} \label{prop:relationbetweendifferentialsystems}
    The following change of variables gives a transformation between the differential systems \eqref{eq-pLUE-diff-relsection} and \eqref{eq-Meixner-diff-relsection}:
        \begin{equation}
   	 \left\{
		\begin{aligned}
   	 	x&=  \frac{s\, g }{\gamma_{\operatorname{L}}} ,\\
   		y&=  \frac{s\, g (\gamma_{\operatorname{L}} f+(f-1) g+n_{\operatorname{L}})}{\gamma_{\operatorname{L}} (g-n_{\operatorname{L}})}, \\
        c&=s,        
	   	 \end{aligned}
	\right.
    \quad\text{or inversely,} \qquad
    \left\{
		\begin{aligned}
   	 	f&= \frac{(x+y) ((\gamma_{\operatorname{M}}-1) x-c \,n_{\operatorname{M}})}{(\gamma_{\operatorname{M}}-1) x (c+x)} ,\\
   		g&=\frac{(\gamma_{\operatorname{M}}-1) x}{c}  , \\
        s&=c.
	   	 \end{aligned}
	\right.
    \end{equation}
    under the parameter matching
    \begin{equation}
    \left\{
		\begin{aligned}
        \beta_{\operatorname{M}} &= \alpha_{\operatorname{L}}+\gamma_{\operatorname{L}} + n_{\operatorname{L}}+1,\\
        \gamma_{\operatorname{M}} &= \gamma_{\operatorname{L}}+1,\\
        n_{\operatorname{M}} &= n_{\operatorname{L}},
        \end{aligned}
        \right.
        \quad\text{or inversely,} \qquad
    \left\{
		\begin{aligned}
        \alpha_{\operatorname{L}} &= \beta_{\operatorname{M}} - \gamma_{\operatorname{M}} - n_{\operatorname{L}},\\
        \gamma_{\operatorname{L}} &= \gamma_{\operatorname{M}}-1,\\
        n_{\operatorname{L}} &= n_{\operatorname{M}}.
        \end{aligned}
        \right.
    \end{equation}
\end{proposition}

\section{Discussion} \label{sec:discussionconclusion}

The main point of this paper was to demonstrate that, when specifying examples of discrete Painlev\'e equations that arise in the context of semi-classical orthogonal polynomials,
one should specify not just the surface type $\mathcal{R}$ and generic symmetry type $\mathcal{R}^{\perp}$, but also the conjugacy class $[w]$ of the element that generates the dynamics and the true symmetry type $\mathcal{S}$.

Not only do we believe that this should be taken into consideration in any attempt to match different weights with the Sakai classification scheme, but it seems that the kinds of non-generic discrete Painlev\'e equations we have studied, corresponding to surfaces with nodal curves, appear frequently in similar contexts.
For example, the instances of the discrete Painlev\'e equations labelled by $\dpain{IV}$ and $\dpain{V}$ appearing in \cite{borodinboyarchenko} are also non-generic, and their parameter constraints correspond to nodal curves.
Further, the instance of $\dpain{IV}$ that appears there corresponds to a single nodal curve on a $D_5^{(1)}$ surface, which is exactly the same situation that we saw here. Our description of the translation element and symmetry group for the non-generalised Meixner case in \Cref{sec:degenMeixner} is applicable to the system  derived in \cite[Th. 9.3 (b)]{borodinboyarchenko}.
It would be interesting to determine if there is some fundamental significance of nodal curves in such contexts.

Further, the nodal curves we have been concerned with in this paper are those fixed by the discrete Painlev\'e dynamics under consideration.
Another way that nodal curves appear is as rational curves on which continuous Painlev\'e flows become linearisable as Riccati equations.
Discrete Painlev\'e equations that do not fix such a nodal curve will act as B\"acklund transformations generating hierarchies of hypergeometric-type special solutions, and indeed it is often these solutions that are of interest to express recurrence coefficients of semi-classical orthogonal polynomials.  

Another feature of the examples studied here that we would like to emphasise is the result of \Cref{prop:relationbetweendifferentialsystems}. 
Despite the discrete systems associated with weights (gM) and (pL) being inequivalent, if one does not impose these discrete systems and considers only the differential systems, the two coincide under a change of variables and parameter identification. 
This is a reflection of a key difference in how surface type relates to Painlevé equations between the differential and discrete cases: while there is a unique Painlevé differential equation corresponding to its surface type, there can be multiple inequivalent discrete Painlevé equations on the same surface family.
% It may be interesting to investigate whether this could be exploited, similarly to in \cite{ondifferentialsystems}, to find new double-scaling degenerations between Painlev\'e equations related to other weights, via identifications of differential but not discrete systems.

    Finally, we remark that one could pose a further refined version of the Painlev\'e identification problem, that asks not just for the equation but for a specific solution of interest.
    Often a particular solution of a discrete Painlev\'e equation is the one relevant to the orthogonal polynomials, and to specify such a solution one should also include with the data $(\mathcal{R},\mathcal{R}^{\perp},\mathcal{S},[w])$ an initial condition or other characterising behaviour. 
    In many cases coming from semi-classical orthogonal polynomials, this is a seed solution for a hierarchy of hypergeometric-type special solutions of the relevant Painlev\'e equation.

\subsection*{Acknowledgements} 
 
GF acknowledges  support of the grant entitled ``Geometric approach to  ordinary differential equations" (01.03.2023-29.02.2024) funded under New Ideas 3B  competition within Priority Research Area III    implemented under the “Excellence Initiative – Research University” (IDUB) Programme  (University of Warsaw) (nr  01/IDUB/2019/94). 
The work of the GF is also partially supported by the project PID2021-124472NB-I00  funded by MCIN/AEI/10.13039/501100011033 and by  ``ERDF A way of making Europe".
AS acknowledges the support of Japan Society for the Promotion of Science (JSPS) through KAKENHI grant numbers 21F21775, 22KF0073 and 24K22843.
AD and AS thank Tomoyuki Takenawa for many helpful discussions.
We thank the anonymous referees for their useful comments and suggestions.

\bibliographystyle{amsalpha}

\begin{thebibliography}{DFLS21b}

\bibitem[Adl98]{adlerdiscreteKN}
Vsevolod E. Adler,
\emph{B\"acklund transformation for the Krichever-Novikov equation},
Internat. Math. Res. Notices 1998, no.~1, 1--4. \MR{1601866}

\bibitem[ABS03]{ABSlist}
Vsevolod E. Adler, Alexander I. Bobenko, and Yuri B. Suris,
\emph{Classification of integrable equations on quad-graphs. The consistency approach},
Comm. Math. Phys.  \textbf{233} (2003), no.~3, 513--543. \MR{1962121}

\bibitem[AGS25]{ASIDEnotes}
Gessica Alecci, Michele Graffeo, and Alexander Stokes,
\emph{Classical Algebraic Geometry and Discrete Integrable Systems},
To appear on ``Symmetries and Integrability of Difference Equations - Lecture Notes of ASIDE15", AMS Contemporary Mathematics. arXiv:2510.12647 [math.AG], 2025.


\bibitem[AHJN16]{AHJN}
James Atkinson, Phil Howes, Nalini Joshi, and Nobutaka Nakazono, 
\emph{Geometry of an elliptic difference equation related to Q4},
J. Lond. Math. Soc. (2)   \textbf{93} (2016), no.~3, 763--784. \MR{3509963}

\bibitem[BC09]{BasorChenPV}
Estelle Basor and Yang Chen,
\emph{Painlev\'e V and the distribution function of a discontinuous linear statistic in the Laguerre unitary ensembles}
J. Phys. A \textbf{42} (2009), no.~3, 035203, 18 pp. \MR{2525311}

\bibitem[BFVA11]{BoelenFilipukVanAssche}
Lies Boelen, Galina Filipuk, and Walter Van Assche, 
\emph{Recurrence coefficients of generalized Meixner polynomials and Painlev\'e equations},
J. Phys. A   \textbf{44} (2011), no.~3, 035202, 19 pp. \MR{2749070}

\bibitem[Bor03]{borodin2003}
Alexei Borodin,
\emph{Discrete gap probabilities and discrete Painlevé equations},
Duke Math. J. \textbf{117} (2003), no.3, 489--542. \MR{1979052}

\bibitem[BB03]{borodinboyarchenko}
Alexei Borodin and Dmitriy Boyarchenko, 
\emph{Distribution of the first particle in discrete orthogonal polynomial ensembles},
Comm. Math. Phys. \textbf{234} (2003), no.~2, 287–338.
\MR{1962463}

\bibitem[CI97]{ChenIsmail}
Yang Chen and Mourad E. H. Ismail,
\emph{Ladder operators and differential equations for orthogonal polynomials},
J. Phys. A \textbf{30} (1997), no.~22, 7817–7829. \MR{1616931}

\bibitem[Chi78]{Chiharabook}
Theodore S. Chihara,
\emph{An introduction to orthogonal polynomials},
Math. Appl., Vol. 13
Gordon and Breach Science Publishers, New York-London-Paris, 1978. xii+249 pp. \MR{0481884}

\bibitem[Cla13]{ClarksonOPs}
Peter~A. Clarkson, 
\emph{Recurrence coefficients for discrete orthonormal polynomials and the Painlev\'e equations}, 
J. Phys. A \textbf{46} (2013), no.~18, 185205, 18 pp. \MR{3055670}

\bibitem[Cla19]{Cla:2019:OPPE}
\bysame,
\emph{Open problems for {P}ainlev\'{e} equations}, SIGMA
  Symmetry Integrability Geom. Methods Appl. \textbf{15} (2019), Paper No. 006,
  20. \MR{3904441}

\bibitem[CJ99]{cresswelljoshi}
Clio Cresswell and Nalini Joshi,
\emph{The discrete first, second and thirty-fourth Painlev\'e hierarchies},
J. Phys. A   \textbf{32} (1999), no.~4, 655--669. \MR{1671841}


% \bibitem[Dol83]{dolgachevweylgroup}
% Igor V. Dolgachev,
% \emph{Weyl groups and Cremona transformations},
% in Peter Orlik ed. \emph{Singularities, Part 1 (Arcata, Calif., 1981)}, 
% Proc. Sympos. Pure Math. \textbf{40}
% American Mathematical Society, Providence, RI, 1983, 283–294. \MR{0713067}

% \bibitem[DO88]{dolgachevortland}
% Igor V. Dolgachev and David Ortland,
% \emph{Point sets in projective spaces and theta functions}, 
% Astérisque No. 165 (1988), 210 pp. \MR{1007155}

\bibitem[DFLS24]{hamiltonians}
Anton Dzhamay, Galina Filipuk, Adam Lig\c{e}za, and Alexander Stokes,
\emph{Different Hamiltonians for differential Painlev\'e equations and their identification using a geometric approach},
J. Differential Equations \textbf{399} (2024), 281--334. \MR{4732165}

\bibitem[DFS20]{DzhFilSto:2020:RCDOPWHWDPE}
Anton Dzhamay, Galina Filipuk, and Alexander Stokes, \emph{Recurrence
  coefficients for discrete orthogonal polynomials with hypergeometric weight
  and discrete {P}ainlev\'{e} equations}, J. Phys. A \textbf{53} (2020),
  no.~49, 495201, 29. \MR{4188771}

\bibitem[DFS21]{ondifferentialsystems}
\bysame, 
\emph{On differential systems related to generalized Meixner and deformed Laguerre orthogonal polynomials}, 
Integral Transforms Spec. Funct. \textbf{32} (2021), no.~5-8, 483--492. \MR{4280695}

\bibitem[DFS22]{DzhFilSto:2022:DERCSOPTRPEGA}
\bysame, \emph{Differential equations for the recurrence coefficients of
  semiclassical orthogonal polynomials and their relation to the {P}ainlev\'{e}
  equations via the geometric approach}, Stud. Appl. Math. \textbf{148} (2022),
  no.~4, 1656--1702. \MR{4433344}

\bibitem[DK20]{antonalisa}
Anton Dzhamay and Alisa Knizel,
\emph{ $q$-Racah ensemble and $q$-$\operatorname{P}(E_7^{(1)}/A_1^{(1)})$ discrete Painlev\'e equation},
Int. Math. Res. Not. IMRN (2020), no.24, 9797–9843. \MR{4190389}

\bibitem[DSSW25]{dP2symmetry}
Anton Dzhamay, Yang Shi, Alexander Stokes, and Ralph Willox,
\emph{What is the symmetry group of a d-$\Pain{II}$ discrete Painlev\'e equation?},
Mathematics (2025), 13, 1123.

\bibitem[DT18]{DzhTak:2018:SASGTDPE}
Anton Dzhamay and Tomoyuki Takenawa, \emph{On some applications of {S}akai's
  geometric theory of discrete {P}ainlev\'{e} equations}, SIGMA Symmetry
  Integrability Geom. Methods Appl. \textbf{14} (2018), Paper No. 075, 20.
  \MR{3830210}

\bibitem[FS23A]{quasihamiltonians}
Galina Filipuk and Alexander Stokes, 
\emph{On Hamiltonian structures of quasi-Painlevé equations},
J. Phys. A \textbf{56} (2023), no. 49, Paper No. 495205, 37 pp.
\MR{4671825}

\bibitem[FS23B]{takasakipaper}
\bysame,
\emph{{T}akasaki's rational fourth
 {P}ainlev\'e-{C}alogero system and geometric regularisability of
 algebro-{P}ainlev\'e equations}, 
 Nonlinearity \textbf{36} (2023), no. 10, 5661–5697.
\MR{4646039}

\bibitem[FVA11]{FilipukVanAssche}
Galina Filipuk and Walter Van Assche, 
\emph{Recurrence coefficients of a new generalization of the Meixner polynomials},
SIGMA Symmetry Integrability Geom. Methods Appl.   \textbf{7} (2011), Paper 068, 11 pp. \MR{2861208}

\bibitem[FIK91]{fokasitskitaev}
Athanassios S. Fokas, Alexander R. Its, and Alexander V. Kitaev,
\emph{Discrete Painlev\'e equations and their appearance in quantum gravity},
Comm. Math. Phys.   \textbf{142} (1991), no.~2, 313--344. \MR{1137067}


\bibitem[FO10]{forresterormerod}
P.~J. Forrester and C.~M. Ormerod, 
\emph{Differential equations for deformed Laguerre polynomials}, 
J. Approx. Theory {\bf 162} (2010), no.~4, 653--677. \MR{2606639}


\bibitem[BR04]{GramaniReview}
Basile Grammaticos and Alfred Ramani,
\emph{Discrete Painlevé equations: a review},
Discrete integrable systems, 245–321.
Lecture Notes in Phys., 644
Springer-Verlag, Berlin, 2004.
\MR{2087743}

\bibitem[HHR00]{genCharlier}
Mahouton Hounkonnou, C. Hounga, and Andr\'e Ronveaux,
\emph{Discrete semi-classical orthogonal polynomials: generalized Charlier},
J. Comput. Appl. Math. \textbf{114} (2000), no.~2, 361–366. \MR{1737084}


\bibitem[HDC20]{HuDzhChe:2020:PLUEDPE}
Jie Hu, Anton Dzhamay, and Yang Chen, 
\emph{Gap probabilities in the {L}aguerre
  unitary ensemble and discrete {P}ainlev\'{e} equations}, 
  J. Phys. A \textbf{53} (2020), no.~35, 354003, 18 pp. \MR{4137540}

  \bibitem[HDC25]{qlaguerre}
  Jie Hu, Anton Dzhamay, and Yang Chen,
\emph{On the recurrence coefficients for the $q$-Laguerre weight and discrete Painlev\'e equations},
J. Phys. A \textbf{58} (2025), no.~2, 025211, 21 pp. \MR{4850803}

\bibitem[Ism05]{Ismailbook}
Mourad E. H. Ismail,
\emph{Classical and quantum orthogonal polynomials in one variable},
Encyclopedia Math. Appl., \textbf{98},
Cambridge University Press, Cambridge, 2005. xviii+706 pp. \MR{2191786}


\bibitem[JGTR06]{joshietal}
Nalini Joshi, Basile Grammaticos, Thamizharasi Tamizhmani, and Alfred Ramani,
\emph{From integrable lattices to non-QRT mappings},
Lett. Math. Phys.   \textbf{78} (2006), no.~1, 27--37. \MR{2271126}

% \bibitem[Kac90]{KAC1990}
% Victor G. Kac,
% \emph{Infinite-dimensional Lie algebras}, 
% Cambridge University Press, Cambridge,
% third edition, 1990.
% \MR{1104219}

\bibitem[KNT11]{KNT}
Kenji Kajiwara, Nobutaka Nakazono, and Teruhisa Tsuda,
\emph{Projective reduction of the discrete Painlev\'e system of type $(A_2+A_1)^{(1)}$},
Int. Math. Res. Not. IMRN 2011, no.~4, 930--966. \MR{2773334}


\bibitem[KNY17]{KajNouYam:2017:GAPE}
Kenji Kajiwara, Masatoshi Noumi, and Yasuhiko Yamada, \emph{Geometric aspects
  of {P}ainlev\'e equations}, J. Phys. A \textbf{50} (2017), no.~7, 073001,
  164. \MR{3609039}

% \bibitem[K\"on05]{konig}
% Wolfgang K\"onig,
% \emph{Orthogonal polynomial ensembles in probability theory},
% Probab. Surv. \textbf{2} (2005), 385–447. \MR{2203677}

\bibitem[LDFZ25]{D6lixing}
Xing Li, Anton Dzhamay, Galina Filipuk, and Da-jun Zhang,
\emph{Recurrence relations for the generalized Laguerre and Charlier orthogonal polynomials and discrete Painlevé equations on the $D_6^{(1)}$ Sakai surface},
Math. Phys. Anal. Geom. \textbf{28} (2025), no.~1, Paper No.~5, 30 pp.
\MR{4879077}

% \bibitem[LNV18]{RMTintro}
% Giacomo Livan, Marcel Novaes, and Pierpaolo Vivo,
% \emph{Introduction to random matrices},
% SpringerBriefs Math. Phys. \textbf{26}
% Springer, Cham, 2018. ix+124 pp. \MR{3751374}

% \bibitem[Loo81]{looijenga}
% Eduard Looijenga, 
% \emph{Rational surfaces with an anticanonical cycle},
% Ann. of Math. (2)   \textbf{114} (1981), no.~2, 267-–322. \MR{0632841}


\bibitem[LC17]{LC17}
Shulin Lyu and Yang Chen,
\emph{The largest eigenvalue distribution of the Laguerre unitary ensemble},
Acta Math. Sci. Ser. B (Engl. Ed.)   \textbf{37} (2017), no.~2, 439--462. \MR{3613423}.

\bibitem[Mal23]{malmquist}
J.~Malmquist, \emph{Sur les {\'e}quattions diff{\'e}retielles du second ordre
 dont l'int{\'e}grale g{\'e}n{\'e}rale a ses points critiques fixes}, Ark.
 Mat. Astr. Fys. \textbf{17} (1922-23), no.~8, 1--89.


\bibitem[MC20]{MC20}
Chao Min and Yang Chen,
\emph{Painlev\'e V, Painlev\'e XXXIV and the degenerate Laguerre unitary ensemble},
Random Matrices Theory Appl. \textbf{9} (2020), no. 2, 2050016, 22 pp. \MR{4085376}


\bibitem[Oka79]{Oka:1979:FAESOPCFPP}
Kazuo Okamoto, \emph{Sur les feuilletages associ\'{e}s aux \'{e}quations du
  second ordre \`a points critiques fixes de {P}. {P}ainlev\'{e}}, Japan. J.
  Math. (N.S.) \textbf{5} (1979), no.~1, 1--79. \MR{614694}

\bibitem[Oka80]{Oka:1980:PHAWPE}
\bysame, \emph{Polynomial {H}amiltonians associated with {P}ainlev\'{e}
  equations. {I}}, Proc. Japan Acad. Ser. A Math. Sci. \textbf{56} (1980),
  no.~6, 264--268. \MR{581468}


\bibitem[Oka87]{Oka:1987:SPEIFPEP}
\bysame, \emph{Studies on the {P}ainlev\'{e} equations. {II}. {F}ifth
  {P}ainlev\'{e} equation {$P_{\rm V}$}}, Japan. J. Math. (N.S.) \textbf{13}
  (1987), no.~1, 47--76. \MR{914314}

% \bibitem[Rai13]{rains}
% Eric Rains,
% \emph{Generalized Hitchin systems on rational surfaces},
% arXiv:1307.4033 [math.AG], 2013.

\bibitem[RCG09]{RCG}
Alfred Ramani, Adrian Stefan Carstea, and Basile Grammaticos,
\emph{On the non-autonomous form of the $Q_4$ mapping and its relation to elliptic Painlev\'e equations},
J. Phys. A   \textbf{42} (2009), no.~32, 322003, 8 pp. \MR{2525848}


% \bibitem[RG09]{dPsinfinite}
% Alfred Ramani and Basile Grammaticos,
% \emph{The number of discrete Painlev\'e equations is infinite},
% Phys. Lett. A   \textbf{373} (2009), no. 34, 3028--3031. \MR{2559808}

\bibitem[Ron86]{ronveaux86}
Andr\'e Ronveaux,
\emph{Discrete semiclassical orthogonal polynomials: generalized Meixner},
J. Approx. Theory \textbf{46} (1986), no.~4, 403--407. \MR{0842801}

\bibitem[Sak01]{SAKAI2001}
Hidetaka Sakai, \emph{Rational surfaces associated with affine root systems and
  geometry of the {P}ainlev\'{e} equations}, Comm. Math. Phys. \textbf{220}
  (2001), no.~1, 165--229. \MR{1882403}


\bibitem[Shi25]{yangtranslations}
Yang Shi,
\emph{Translation in affine Weyl groups and its application in discrete integrable systems},
Proc. A. \textbf{481} (2025), no.~2312, Paper No. 20240749, 21~pp.


 \bibitem[Sho39]{shohat}
 Jacques Shohat,
 \emph{A differential equation for orthogonal polynomials},
Duke Math. J. \textbf{5} (1939), no.~2, 401--417. \MR{1546133}

\bibitem[SVA12]{SmetVanAssche}
Christophe Smet and Walter Van Assche,
\emph{Orthogonal polynomials on a bi-lattice},
Constr. Approx.   \textbf{36} (2012), no.~2, 215--242. \MR{2957309}

\bibitem[Sze75]{szegobook}
Gabor Szeg\H o, {\it Orthogonal polynomials}, fourth edition, 
American Mathematical Society Colloquium Publications, Vol. XXIII, Amer. Math. Soc., Providence, RI, 1975. \MR{0372517}

\bibitem[Tak03]{takenawaD5}
Tomoyuki Takenawa,
\emph{Weyl group symmetry of type $D_5^{(1)}$ in the $q$-Painlev\'e V equation},
Funkcial. Ekvac.   \textbf{46} (2003), no.~1, 173--186. \MR{1996297}

\bibitem[TD24]{antonelizaveta}
Elizaveta Trunina and Anton Dzhamay, 
\emph{Orthogonal Polynomials for the Gaussian Weight with a Jump and Discrete Painlevé Equations}, 
Proceedings of the 16th International Symposium on Orthogonal Polynomials, Special Functions and Applications (OPSFA-16), 47–61.
CRM Ser. Math. Phys.
Springer, 2025. \MR{4982239}

\bibitem[VA07]{VanAsscheDESFOP}
Walter Van Assche,
\emph{Discrete Painlev\'e equations for recurrence coefficients of orthogonal polynomials},
Difference equations, special functions and orthogonal polynomials, 687–725.
World Scientific Publishing Co. Pte. Ltd., Hackensack, NJ, 2007. \MR{2451211}


\bibitem[VA18]{walterbook}
\bysame,
\emph{Orthogonal polynomials and Painlevé equations},
Austral. Math. Soc. Lect. Ser., 27
Cambridge University Press, Cambridge, 2018. xii+179 pp. \MR{3729446}

\bibitem[VAF03]{VanAsscheFoup}
Walter Van Assche and Mama Foupouagnigni,
\emph{Analysis of non-linear recurrence relations for the recurrence coefficients of generalized Charlier polynomials},
J. Nonlinear Math. Phys. \textbf{10} (2003), 231--237. \MR{2063533}


\bibitem[ZCZ25]{siqi}
Menghun Zhu, Siqi Chen, and Xuhao Zhang,
\emph{Recurrence coefficients for the time-evolved Jacobi weight and discrete Painlev\'e equations on the $D_5$ Sakai surface},
	arXiv:2511.04176 [math.CA], 2025.
\end{thebibliography}

\providecommand{\bysame}{\leavevmode\hbox to3em{\hrulefill}\thinspace}
\providecommand{\MR}{\relax\ifhmode\unskip\space\fi MR}

\begingroup
\footnotesize % Set font size to footnotesize

\endgroup

\clearpage

\end{document}